\documentclass[11pt, a4paper]{article}
\usepackage{a4wide}

\usepackage{amsmath,amssymb,amsthm,hyperref,color}
\usepackage{float}
\usepackage[font=small]{caption}
\usepackage[font=footnotesize]{subcaption}
\usepackage[inline]{enumitem}
\usepackage{filecontents}
\usepackage{tikz}
\usetikzlibrary{intersections,calc, angles, quotes, arrows.meta}
\usepackage{pgfplots}

\usepackage{hyperref}
\definecolor{oxfordColor}{HTML}{002147}
\hypersetup{
    colorlinks   = true,                    
    linkcolor    = blue!70!black,
    anchorcolor  = gray,
    citecolor    = blue!70!black,
    filecolor    = red,
    menucolor    = green,
    runcolor     = red,
    urlcolor     = [rgb]{0,0.2,0.5}
}
\usepackage{color}
\usepackage{xcolor}

\definecolor{ChetwodeBlue}{HTML}{6b77ad}  
\definecolor{Jaguar}{HTML}{262731}       
\definecolor{TurkishRose}{HTML}{B46275}   
\definecolor{Hunter}{HTML}{3F704D}   

\newtheorem{theorem}{Theorem}
\newtheorem{lemma}[theorem]{Lemma}

\newtheorem{corollary}[theorem]{Corollary}
\newtheorem{remark}[theorem]{Remark}

\newtheorem{definition}[theorem]{Definition}

\newtheorem*{thmmainone}{Theorem~\ref{thm:mainone}}
\newtheorem*{thmmaintwo}{Theorem~\ref{thm:maintwo}}
\newtheorem*{thmmainthree}{Theorem~\ref{thm:mainthree}}
\newtheorem*{thmmainfour}{Theorem~\ref{thm:mainfour}}
\newtheorem*{thmshifts}{Theorem~\ref{thm:approx-shifts}}

\newtheorem*{corshifts}{Lemma~\ref{cor:implement-approx:real-with-complex}}
\newtheorem*{lemshifts}{Lemma~\ref{lem:implement-approx:pc}}
\newtheorem*{lemout}{Lemma~\ref{lem:reduction:norm:outline}}

\newcommand{\ET}{\textsc{Tutte}}
\newcommand{\ZT}{\textsc{ZeroFreeTutte}}
\newcommand{\TT}{\textsc{ThickenedTutte}}
\newcommand{\ZTT}{\textsc{ZeroFreeThickenedTutte}}
\newcommand{\ST}{\textsc{SignTutte}}

\newcommand{\STz}{\textsc{Sign-Tutte}}

\newcommand{\RT}{\textsc{RatioTutte}}
\newcommand{\RTT}{\textsc{RatioThickenedTutte}}
\newcommand{\PT}{\textsc{PlanarTutte}}
\newcommand{\ZPT}{\textsc{ZeroFreePlanarTutte}}
\newcommand{\TPT}{\textsc{ThickenedPlanarTutte}}
\newcommand{\ZTPT}{\textsc{ZeroFreeThickenedPlanarTutte}}
\newcommand{\RPT}{\textsc{RatioPlanarTutte}}
\newcommand{\RTPT}{\textsc{RatioThickenedPlanarTutte}}
\newcommand{\SPT}{\textsc{SignPlanarTutte}}

\newcommand{\SRPT}{\textsc{Sign-Real-PlanarTutte}}

\newcommand{\SPTz}{\textsc{Sign-PlanarTutte}}

\newcommand{\FNI}{\textsc{Factor-}K\textsc{-NormIsing}}
\newcommand{\FNP}{\textsc{Factor-}K\textsc{-NormPotts}}
\newcommand{\FNT}{\textsc{Factor-}K\textsc{-NormTutte}}
\newcommand{\FNPT}{\textsc{Factor-}K\textsc{-NormPlanarTutte}}
\newcommand{\FNPP}{\textsc{Factor-}K\textsc{-NormPlanarPotts}}
\newcommand{\FNTz}{\textsc{Factor-}K\textsc{-NormTutte}}

\newcommand{\FNPTz}{\textsc{Factor-}K\textsc{-NormPlanarTutte}}

\newcommand{\sDATz}{\textsc{-ArgTutte}}

\newcommand{\sDAPTz}{\textsc{-ArgPlanarTutte}}

\newcommand{\sDAPP}{\textsc{-ArgPlanarPotts}}
\newcommand{\sDAT}{\textsc{-ArgTutte}}
\newcommand{\sDAPT}{\textsc{-ArgPlanarTutte}}
\newcommand{\sDAP}{\textsc{-ArgPotts}}
\newcommand{\sDAI}{\textsc{-ArgIsing}}
\newcommand{\pDAT}{\textsc{Distance-}}

\newcommand{\FNJ}{\textsc{Factor-}K\textsc{-NormJones}}
\newcommand{\DAJ}{\textsc{Distance-}\rho\textsc{-ArgJones}}
\newcommand{\DAJpi}{\textsc{Distance-}\pi/3\textsc{-ArgJones}}

\def\numP{\mathrm{\#P}} 

\def\NP{\mathrm{NP}}

\newcommand{\nbits}{b}

\makeatletter
\def\prob#1#2#3{\goodbreak\begin{list}{}{\labelwidth\z@ \itemindent-\leftmargin
                        \itemsep\z@  \topsep6\p@\@plus6\p@
                        \let\makelabel\descriptionlabel}
                \item[\textbf{Name:}]#1   \vspace{-1ex}
               \item[\textbf{Instance:}]                #2   \vspace{-1ex}
                \item[\textbf{Output:}]#3  
                \end{list}}
\makeatother

\title{The complexity of approximating the complex-valued Potts model \footnote{A short preliminary version without proofs will appear on the proceedings of MFCS 2020.}}
\author{
Andreas Galanis \thanks{
  Department of Computer Science, University of Oxford, Wolfson Building, Parks Road, Oxford, OX1~3QD, UK.}
  \and
  Leslie Ann Goldberg \footnotemark[2]
\and
Andr\'es Herrera-Poyatos \footnotemark[2] \thanks{This author is supported by an Oxford-DeepMind Graduate Scholarship and a EPSRC Doctoral Training Partnership.}
 }

\date{18 November 2021}

\begin{document}

\maketitle

\thispagestyle{empty}
\begin{abstract} 
We study the complexity of approximating the partition function of the  $q$-state Potts model and the closely related Tutte polynomial for complex values of the underlying parameters.  Apart from the classical connections with quantum computing and phase transitions in statistical physics, recent work in approximate counting has shown that the behaviour in the complex plane, and more precisely the location of zeros, is strongly connected with the complexity of the approximation problem, even for positive real-valued parameters. Previous work in the complex plane by Goldberg and Guo focused on $q=2$, which corresponds to the case of the Ising model; for $q>2$, the behaviour in the complex plane is not as well understood and most work applies only to the real-valued Tutte plane. 

Our main result is a complete classification of the complexity of the approximation problems for all non-real values of the parameters, by establishing \#P-hardness results that apply even when restricted to planar graphs. Our techniques apply to all $q\geq 2$ and further complement/refine previous results both for the Ising model and the Tutte plane, answering in particular a question raised by Bordewich, Freedman, Lov\'{a}sz and Welsh  in the context of quantum computations.
\end{abstract}

\newpage

\clearpage
\setcounter{page}{1}

\section{Introduction}

The $q$-state Potts model is a classical model of ferromagnetism in statistical physics~\cite{Potts1952, Welsh1993} which generalises the well-known Ising model. On a (multi)graph $G = (V, E)$, configurations of the model are all possible assignments $\sigma: V\rightarrow [q]$ where $[q]=\{1, \ldots, q\}$ is a set of $q$ spins with $q\geq 2$. The model is parameterised by $y$, which is a function of the temperature of the model and is also known as the \emph{edge interaction}. Each configuration $\sigma$ is assigned  weight $y^{m(\sigma)}$ where $m(\sigma)$ denotes the number of monochromatic edges of $G$ under $\sigma$. The partition function of the model is the aggregate weight over all configurations, i.e.,
\begin{equation*}
  Z_{\mathrm{Potts}}(G; q, y) = \sum_{\sigma \colon V \to [q]} y^{m(\sigma)}, 
\end{equation*}
When $q = 2$, this model is known as the Ising model, and we sometimes use the notation $Z_{\mathrm{Ising}}(G; y)$ to denote its partition function.

The Ising/Potts models have an extremely useful generalisation to non-integer values of $q$ via the so-called ``random-cluster'' formulation and the closely related Tutte polynomial. In particular, for numbers $q$ and $\gamma$,  the Tutte polynomial of a graph $G$ is given by
\begin{equation} \label{eq:tutte}
  Z_{\mathrm{Tutte}}(G; q, \gamma) = \sum_{A \subseteq E} q^{k(A)} \gamma^{|A|}, 
\end{equation}
where $k(A)$ denotes the number of connected components in the graph $(V, A)$ (isolated vertices do count).  When $q$ is an integer with $q \ge 2$, we have $Z_{\mathrm{Potts}}(G; q, y)  = Z_{\mathrm{Tutte}}(G; q, y - 1)$, see, for instance,~\cite{Sokal2005}. The Tutte polynomial on planar graphs is particularly relevant in quantum computing since it corresponds to the Jones polynomial of an ``alternating link''~\cite[Chapter 5]{Welsh1993}, and polynomial-time quantum computation can be simulated by additively approximating the Jones polynomial at a suitable value, as we will explain later in more detail, see also~\cite{Bordewich2005} for details.

In this paper, we study the complexity of approximating the partition function of the Potts model and the Tutte polynomial on planar graphs as the parameter $y$ ranges in the complex plane. Traditionally, this problem has been mainly considered in the case where $y$ is a positive real, however recent developments have shown that for various models, including the Ising and Potts models, there is a close interplay between the location of zeros of the partition function in the complex plane and the approximability of the problem, even for positive real values of $y$. 

The framework of viewing partition functions as polynomials in the complex plane of the underlying parameters has been well-explored in statistical physics and has recently gained traction in computer science as well in the context of approximate counting. On the positive side, zero-free regions in the complex plane translate into efficient algorithms for approximating the partition function~\cite{barvinokbook,PR}  and this scheme has lead to a broad range of new algorithms even for positive real values of the underlying parameters~\cite{LSP,peters2019,liu2019ising,JSPab, peters2018location,barvinokregts2019, GLLZ, GLLb,harrow2019classical}. On the negative side, the presence of zeros poses a barrier to this approach and, in fact, it has been demonstrated that zeros mark the onset of computational hardness for the approximability of the partition function~\cite{Goldberg2014,Goldberg2017,Beza,Bezb}. These new algorithmic and computational complexity developments stemming from the complex plane mesh with the statistical physics perspective where zeros have long been studied in the context of pinpointing phase transitions, see e.g.,~\cite{Sokal2005,Welsh1993,Heilmann1972, lieb1981general,yang1952statistical,bena2005statistical}.

For the problem of exactly computing the partition function of the Potts model, Jaeger, Vertigan and Welsh~\cite{Jaeger1990}, as a corollary of a more general classification theorem for the Tutte polynomial, established $\# \mathsf{P}$-hardness unless $(q, y)$ is one of seven exceptional points, see Section~\ref{sec:pre:hardness} for more details; Vertigan~\cite{Vertigan2005} further showed that the same classification applies on planar graphs with the exception of the Ising model ($q=2$), where the problem is in $\mathsf{FP}$. 

For the approximation problem, the only known result that applies for general values $y$ in the complex plane is by Goldberg and Guo~\cite{Goldberg2017}, which addresses the case $q=2$; the case $q\geq 3$ is largely open apart from the case when $y$ is real which has been studied extensively even for planar graphs~\cite{Jerrum1993, Goldberg2008, Goldberg2012Potts, Goldberg2014,  Kuperberg2015, Goldberg2017}. We will review all these results more precisely in the next section, where we also state our main theorems.

\subsection{Our main results} \label{sec:intro:main-results}

In this work, we completely classify the complexity of approximating $Z_{\mathrm{Potts}}(G; q, y)$ for $q\geq 2$ and non-real $y$, even on planar graphs $G$; in fact, our results  also classify the complexity of the Tutte polynomial on planar graphs for reals $q\geq 2$ and non-real $\gamma$. Along the way, we also answer a question for the Jones polynomial raised by Bordewich,~Freedman,~Lov\'{a}sz, and Welsh~\cite{Bordewich2005}.

To formally state our results, we define the computational problems we consider.   Let $K$ and $\rho$ be real algebraic numbers with $K > 1$ and $\rho > 0$. We investigate the complexity of the following problems for any integer $q$ with $q \ge 2$ and any algebraic number $y$.\footnote{For $z\in \mathbb{C}\backslash\{0\}$, we denote by $|z|$ the norm of $z$, by $\mathrm{Arg}(z)\in [0,2\pi)$ the principal argument of $z$ and  by $\arg(z)$ the set $\{\mathrm{Arg}(z) + 2 \pi j : j \in \mathbb{Z}\}$ of all the arguments of $z$, so that for any $a\in \arg(z)$ we have $z = |z|\exp(i a)$.}

\prob{$\FNP(q, y)$}{A (multi)graph $G$.}{ If $Z_{\mathrm{Potts}}(G; q, y)= 0$, the algorithm may output any rational number. Otherwise, it must output a rational number $\hat{N}$ such that
$\hat{N}/ K \le \left|Z_{\mathrm{Potts}}(G; q, y)\right| \le K \hat{N} $.}

A well-known fact is that the difficulty of the problem $\FNP(q, \gamma)$  does not depend on the constant $K > 1$.
This can be proved using standard powering techniques (see~\cite[Lemma 11]{Goldberg2017} for a proof when $q = 2$). In fact,
the complexity of the problem is the same even for $K=2^{n^{1-\epsilon}}$ for any constant $\epsilon>0$ where $n$ is the size of the input.

\prob{$\pDAT \rho \sDAP(q, y)$}{A (multi)graph $G$.}{If $Z_{\mathrm{Potts}}(G; q, y)= 0$, the algorithm may output any rational number. Otherwise, it must output a  rational $\hat{A}$ such that, for some $a \in \arg(Z_{\mathrm{Potts}}(G; q, \gamma))$, $| \hat{A} - a| \le \rho$ .}

In the special case that  $q$ equals $2$, we omit the argument $q$ and write \textsc{Ising} instead of \textsc{Potts} in the name of the problem. Similarly, when the input of the problems is restricted to planar graphs, we write \textsc{PlanarPotts} instead of \textsc{Potts}. We also consider the problems $\FNT(q, \gamma)$ and  $\pDAT \pi/3 \sDAT(q, \gamma)$ for the Tutte polynomial when $q,\gamma$ are algebraic numbers. Note also that, when $q,\gamma$ are real, the latter problem is equivalent to finding the sign of the Tutte polynomial, and we sometimes write $\ST(q, \gamma)$ (and $\SPT(q, \gamma)$ for the planar version of the problem).

Our first and main result is a full resolution of the complexity of approximating $Z_{\mathrm{Potts}}(G; q, y)$ for $q\geq 3$ and non-real $y$. More precisely, we show the following.
\newcommand{\statethmmainone}{Let $q\geq 3$ be an integer, $y\in \mathbb{C}\backslash \mathbb{R}$ be an algebraic number, and $K>1$. Then, the problems $\FNPP(q, y)$ and  $\pDAT \pi/3 \sDAPP(q, y)$ are $\numP$-hard, unless $q=3$ and $y \in \{e^{2 \pi i / 3}, e^{4 \pi i / 3}\}$ when both problems can be solved exactly in polynomial time.}
\begin{theorem} \label{thm:mainone}
\statethmmainone
\end{theorem}
We remark that, for real $y>0$,  the complexity of approximating  $Z_{\mathrm{Potts}}(G; q, y)$ on planar graphs is not fully known, though on general graphs the problem is \#BIS-hard~\cite{Goldberg2012Potts} and NP-hard for $y\in (0,1)$~\cite{Goldberg2008}, for all $q\geq 3$. For real $y<0$, the problem is NP-hard on general graphs when $y \in (-\infty, 1-q]$ for all $q\geq 3$ (\cite{Goldberg2014})\footnote{\label{foot:NPhard}Note, for $y \in (-\infty, 1-q)\cup [0, \infty)$, \#P-hardness is impossible (assuming $\NP\neq \numP$): finding the sign of $Z_{\mathrm{Potts}}(G; q, y)$ is easy, even on non-planar graphs (\cite{Goldberg2014}), and $Z_{\mathrm{Potts}}(G; q, y)$ can be approximated using an NP-oracle. For  $y=1-q$, the same applies when $q\geq 6$; the cases $q\in\{3,4,5\}$ are not fully resolved though~\cite{Goldberg2014} shows that $q=3,4$ are NP-hard, whereas $q=5$ should be easy unless Tutte's 5-flow conjecture is false~\cite[Section 3.5]{Welsh1993}.} and $\numP$-hard on planar graphs when $y\in(1-q,0)$ and $q\geq 5$ (\cite{Kuperberg2015}, see also~\cite{Goldberg2012}). Our techniques for proving Theorem~\ref{thm:mainone} allow us to  resolve the remaining cases $q=3,4$ for $y\in (1-q,0)$ on planar graphs, as a special case of the following theorem that applies for general $q\geq 3$.  This is our second main result.
\newcommand{\statethmmaintwo}{Let $q\geq 3$ be an integer, $y\in (-q+1,0)$ be a real algebraic number, and $K>1$. Then $\FNPP(q, y)$ and $\pDAT \pi/3 \sDAPP(q, y)$ are $\numP$-hard, unless $(q,y) =(4, -1)$ when both problems can be solved exactly in polynomial time.}
\begin{theorem} \label{thm:maintwo}
\statethmmaintwo
\end{theorem}

Our third main contribution is a full classification of the range of the parameters where approximating the partition function of the Ising model is \#P-hard. Note, on planar graphs $G$, $Z_{\mathrm{Ising}}(G; y)$ can be computed in polynomial time for all $y$. For general (non-planar) graphs and non-real $y$, Goldberg and Guo show $\numP$-hardness on the non-real unit circle ($|y|=1$) with $y\neq \pm i$, and establish $\NP$-hardness elsewhere. Our next result shows that the NP-hardness results of~\cite{Goldberg2017} for non-real $y$ can be elevated to  $\numP$-hardness.
\newcommand{\statethmmainthree}{  Let $y\in\mathbb{C}\backslash \mathbb{R}$ be an algebraic number, and $K>1$. Then, $\FNI(y)$ and  $\pDAT \pi/3 \sDAI(y)$ are $\numP$-hard, unless $y=\pm i$ when both problems can be solved exactly in polynomial time.}
\begin{theorem} \label{thm:mainthree}
\statethmmainthree
\end{theorem}
 For real $y$, we  remark that  the problems of approximating $Z_{\mathrm{Ising}}(G; y)$ and determining its sign (when non-trivial) are well-understood:\footnote{Analogously to Footnote~\ref{foot:NPhard}, for $y\in (-\infty,-1)\cup (0,1)$ \#P-hardness is unlikely since the problem can be approximated with an NP-oracle.}  the problem is FPRASable for $y>1$ and NP-hard for $y\in (0,1)$ (\cite{Jerrum1993}),  $\numP$-hard for $y\in (-1,0)$~\cite{Goldberg2017, Goldberg2014}, and equivalent to approximating $\# \textsc{PerfectMatchings}$ for $y<-1$~\cite{Goldberg2008}. For $y=0,\pm 1$, $Z_{\mathrm{Ising}}(G; y)$ can be computed exactly in polynomial time.

\subsection{Consequences of our techniques for the Tutte/Jones polynomials}  \label{sec:intro:tutte}

While our main results are on the Ising/Potts models, in order to prove them it is convenient to work in the ``Tutte world''; this simplifies the proofs and has also the benefit of allowing us to generalise our results to non-integer $q$. The following result generalises Theorem~\ref{thm:mainone} to non-integer $q>2$. 

\newcommand{\statethmmainfour}{Let $q>2$ be a real, $\gamma\in \mathbb{C}\backslash \mathbb{R}$ be an algebraic number, and $K>1$. Then,  $\FNPT(q, \gamma)$ and $\pDAT \pi/3 \sDAPT(q, \gamma)$ are $\numP$-hard, unless $q=3$ and $\gamma+1 \in \{e^{2 \pi i / 3}, e^{4 \pi i / 3}\}$ when both problems can be solved exactly in polynomial time.}
\begin{theorem} \label{thm:mainfour}
\statethmmainfour
\end{theorem}
Our techniques can further be used to elevate previous $\NP$-hardness results of~\cite{Goldberg2014, Goldberg2008} in the Tutte plane to $\numP$-hardness for planar graphs, and answer a question for the Jones polynomial raised by Bordewich et al. in~\cite{Bordewich2005}. A more detailed discussion can be found in Section~\ref{sec:further-consequences}.

\section{Proof outline}

In this section we provide some insight on the proofs of our main results. As mentioned earlier, the proofs  are performed in the context of the Tutte polynomial. 

In previous $\numP$-hardness results~\cite{Goldberg2014,Goldberg2017} for the Tutte polynomial, the main technique was to reduce the exact counting $\# \textsc{MinimumCardinality } (s,t)\textsc{-Cut}$ problem to the problem of approximating  $Z_{\text{Tutte}}(G; q, \gamma)$ using an elaborate binary search based on suitable oracle calls. Key to these oracle calls are gadget constructions which are mainly based on series-parallel graphs which ``implement'' points $(q',\gamma')$; this means that, by  pasting the gadgets appropriately onto a graph $G$, the computation of $Z_{\text{Tutte}}(G; q', \gamma')$ reduces  to the computation of $Z_{\text{Tutte}}(G; q, \gamma)$. Much of the work in~\cite{Goldberg2014,Goldberg2017}, and for us as well, is understanding what values $(q',\gamma')$ can be implemented starting from $(q,\gamma)$.

For planar graphs, while the binary-search technique from~\cite{Goldberg2017} is still useful,  we  have to use  a different overall reduction scheme since  the problem $\# \textsc{MinimumCardinality } (s,t)\textsc{-Cut}$ is not $\# \mathsf{P}$-hard when the input is restricted to planar graphs~\cite{Provan1983}. To obtain our $\numP$-hardness results our plan instead is to reduce the problem of exactly evaluating the Tutte polynomial for some appropriately selected parameters $q', \gamma'$ to the problem of computing its sign and the problem of approximately evaluating it at parameters $q,\gamma$; note, this gives us the freedom to use any parameters $q',\gamma'$ we wish as long as the corresponding exact problem is $\numP$-hard. Then, much of the work consists of understanding what values $(q',\gamma')$ can be implemented starting from $(q,\gamma)$, so we focus on that component first. 

We first review previous constructions in the literature, known as shifts, and then introduce our refinement of these constructions, which we call polynomial-time approximate shifts, and state our main result about them.

\subsection{Shifts in the Tutte plane}\label{sec:shwa}
We say that that there is a \emph{shift} from $(q, \gamma_1)$ to $(q, \gamma_2)$ if there is a graph  $H=(V,E)$ and  vertices $s,t$ such that
  \begin{equation*}
    \gamma_2= q \frac{ Z_{st}(H; q, \gamma_1)}{ Z_{s|t}(H; q, \gamma_1)},
  \end{equation*}
	where $Z_{st}(H; q, \gamma_1)$ is the contribution to $Z_{\text{Tutte}}(H; q, \gamma_1)$ from configurations $A\subseteq E$ in which $s,t$ belong to the same connected component in $(V,A)$, while  $Z_{s|t}(H; q, \gamma_1)$ is the contribution from all other configurations $A$.  In the following, we will usually encounter shifts in the $(x,y)$-parametrisation of the Tutte plane, rather than 
	the $(q,\gamma)$-parameterisation which was used for convenience here.
To translate between these, set   $y = \gamma + 1$ and $(x-1)(y-1) = q$, see~\cite[Chapter 3]{Welsh1993}. We denote by $\mathcal{H}_q$ the hyperbola $\{(x,y)\in \mathbb{C}^2 : (x-1)(y-1) = q \}$, and we will  use both   parametrisations as convenient. Section~\ref{sec:pre:weights} has a more detailed description of shifts that apply to the multivariate Tutte polynomial.
	
	As described earlier, shifts can be used to ``move around'' the complex plane. If one knows hardness for some $(x_2, y_2) \in \mathcal{H}_q$, and there is a shift   from $(x_1, y_1)\in \mathcal{H}_q$ to $(x_2,y_2)$, then one also obtains hardness for $(x_1,y_1)$. This approach has been very effective when attention is restricted to real parameters~\cite{Goldberg2008, Goldberg2012, Goldberg2014}, however, when it comes to non-real parameters, the success of this approach has been limited. To illustrate this, in~\cite{Goldberg2017}, the authors established $\numP$-hardness of the Ising model  when $y_2 \in (-1,0)$, and used this to obtain $\numP$-hardness for $y_1$ on the unit circle by constructing appropriate shifts. However, their shift construction does not extend to general complex numbers, and this kind of result seems unreachable with those techniques.

\subsection{Polynomial-time approximate shifts}

To obtain our main theorems, we instead need to consider what we call polynomial-time approximate shifts;  such a shift from $(x_1,y_1)\in \mathcal{H}_q$ to $(x_2, y_2)\in \mathcal{H}_q$ is an  algorithm that, for any positive integer $n$, computes  in time polynomial in $n$ a graph $G_n$ that $(x_1, y_1)$-implements $(\hat{x}_2, \hat{y}_2)$ with $\left| y_2 - \hat{y}_2 \right| \le 2^{-n}$. In fact, our constructions need to maintain planarity, and we will typically ensure  this by 
either making  every~$G_n$ a  series-parallel graph, in which case we call the algorithm a \textit{polynomial-time approximate series-parallel shift}, or 
by making every~$G_n$  a theta graph, in which case we call the algorithm a \textit{polynomial-time approximate theta shift}.\footnote{\label{foot:series}A theta graph consists of two terminals $s$ and $t$ joined by internally disjoint paths~\cite{Brown2001}. A series-parallel graph  with  terminals $s$ and $t$ can be obtained from the single-edge graph  with edge $(s,t)$ by repeatedly subdividing edges or adding parallel edges~\cite[Chapter 11]{Brandstadt1999}.}

These generalised shifts allow us to overcome the challenges mentioned above and are key ingredients in our reduction. Our main technical theorem about them is the following.
\newcommand{\statethmshifts}{Let $q\geq 2$ be a real algebraic number. Let $x$ and $y$ be algebraic numbers such that $(x, y) \in \mathcal{H}_q$, $y \in (-1,0)\cup (\mathbb{C}\backslash \mathbb{R})$ and $(x,y) \not \in \{(i, -i), (-i, i), (\omega_3, \omega_3^2), (\omega_3^2, \omega_3)\}$, where $\omega_3 = \exp(2 \pi i / 3)$. Then, for any pair of real algebraic numbers $(x', y') \in \mathcal{H}_q$ there is a polynomial-time approximate series-parallel shift from $(x,y)$ to $(x',y')$.}
\begin{theorem} \label{thm:approx-shifts} 
  \statethmshifts
\end{theorem}
The exceptions $\{(i, -i), (-i, i), (\omega_3, \omega_3^2), (\omega_3^2, \omega_3)\}$ are precisely the non-real points of the $(x,y)$ plane where the Tutte polynomial of a graph can be evaluated in polynomial time (see Section~\ref{sec:pre:hardness}). As we will see, being able to $(x,y)$-implement approximations of any number in $(-1, 0)$ is essentially the property that makes the approximation problem $\numP$-hard at $(x,y)$.

We remark that the idea of implementing approximations of a given weight or edge interaction has been explored in the literature, though only when all the edge interactions involved are real. We review these results in Section~\ref{sec:approx-shifts}. 

We study the properties of polynomial-time approximate shifts in Section~\ref{sec:approx-shifts} and prove Theorem~\ref{thm:approx-shifts} in Section~\ref{sec:complex-implementations}. In the next section, we describe some of the techniques used.

\subsubsection{Proof Outline of Theorem~\ref{thm:approx-shifts}}
Shifts, as defined in Section~\ref{sec:shwa}, have a transitivity property: if there is a shift from $(x_1,y_1)$ to $(x_2,y_2)$ and from $(x_2, y_2)$ to $(x_3, y_3)$, then there is a shift from $(x_1,y_1)$ to $(x_3, y_3)$, see Section~\ref{sec:pre:weights} for more details.  

The polynomial-time approximate shift given in Theorem~\ref{thm:approx-shifts} is constructed in a similar way. First, we construct a polynomial-time approximate shift from $(x ,y)$ to some $(x_2, y_2 )$ such that $y_2 \in (-1,0)$, where  $x_2$ and $y_2$ depend on $x,y$. Then, we construct a polynomial-time approximate shift from $(x_2, y_2 )$ to $(x', y')$. Finally, we combine both polynomial-time approximate shifts using an analogue of the transitivity property. 

However, when this approach is put into practice, there is a difficulty that causes various technical complications: we only have mild control in our constructions over the intermediate shift $(x_2, y_2)$. In particular,    even if the numbers $x$ and $y$ are algebraic, we cannot guarantee that $x_2$ and $y_2$ are algebraic, and this causes problems with obtaining  the required transitivity property. Instead, we have to work with  a wider class of numbers, the set $\mathsf{P}_{\mathbb{C}}$ of \emph{polynomial-time computable numbers}. These are numbers that can be approximated efficiently, i.e., for $y\in \mathsf{P}_{\mathbb{C}}$ there is an algorithm that computes $\hat{y}_n \in \mathbb{Q}[i]$ with $|y - \hat{y}_n| \le 2^{-n}$ in time polynomial in $n$~\cite[Chapter 2]{Ko1991}. We denote by $\mathsf{P}_{\mathbb{R}}=\mathbb{R}\cap \mathsf{P}_{\mathbb{C}}$ the set of polynomial-time computable real numbers.

Our polynomial-time approximate shifts are constructed in Section~\ref{sec:complex-implementations}. The first of these  polynomial-time approximate shifts is provided by Lemma~\ref{cor:implement-approx:real-with-complex}.

\newcommand{\statecorshifts}{ Let $q$ be a real algebraic number with $q \ge 2$. Let $x$ and $y$ be algebraic numbers such that $(x, y) \in \mathcal{H}_q$, $y\in (-1,0)\cup (\mathbb{C}\backslash \mathbb{R})$ and $(x,y) \not \in \{(i, -i), (-i, i), (\omega_3, \omega_3^2), (\omega_3^2, \omega_3)\}$, where $\omega_3 = \exp(2 \pi i / 3)$. Then there is a polynomial-time approximate series-parallel shift from $(x,y)$ to $(x',y')$ for some $(x', y') \in \mathcal{H}_q$ with $x', y' \in \mathsf{P}_{\mathbb{R}}$ and $y' \in (0,1)$.}
\begin{lemma} \label{cor:implement-approx:real-with-complex}
  \statecorshifts
\end{lemma}
The  construction  in Lemma~\ref{cor:implement-approx:real-with-complex} is obtained using a theta graph and trying to get a shift  that is very close to the real line. However, we cannot control the point $(x',y')$ that we are approximating, and as mentioned,  $x',y'$ might not be algebraic. The proof of Lemma~\ref{cor:implement-approx:real-with-complex} requires the most technical work in the paper and is given in Section~\ref{sec:here}.

Using Lemma~\ref{cor:implement-approx:real-with-complex}, we have  a series-parallel polynomial-time approximate shift from $(x,y)$ to some $(x', y') \in \mathcal{H}_q$ with $x', y' \in \mathsf{P}_{\mathbb{R}}$ and $y'\in (0,1)$. Next, we have to construct a  polynomial-time approximate shift from $(x', y')$ to $(\hat{x}, \hat{y})$, where $(\hat{x}, \hat{y})$ is the point that we want to shift to in Theorem~\ref{thm:approx-shifts}. In fact, we actually use a theta shift, which also facilitates establishing the required transitivity property later on. Note that since $y'$ is not necessarily algebraic, we can not directly apply the results that have already appeared in the literature on implementing approximations of edge interactions. In the next lemma, we generalise these results to the setting of polynomial-time computable numbers, where we need to address some further complications that arise from computing with polynomial-time computable numbers instead of algebraic numbers. The proof of the lemma is given in Section~\ref{sec:rvrtvtvt}.

\newcommand{\statelemshifts}{ Let $q, x, y \in \mathsf{P}_{\mathbb{R}}$ such that $q > 0$, $(x,y) \in \mathcal{H}_q$, $y$ is positive and $1-q/2 < y < 1$. There is a polynomial-time algorithm that takes as an input:
  \begin{itemize}
  \item two positive integers $k$ and $n$, in unary;
  \item a real algebraic number $w \in [y^k, 1]$.
  \end{itemize}
  The algorithm produces a theta graph $J$ that $(x,y)$-implements $(\hat{x}, \hat{y})$ such that $\big| \hat{y} - w \big| \le 2^{-n}$. The size of $J$ is at most a polynomial in $k$ and $n$, independently of $w$.}
\begin{lemma} \label{lem:implement-approx:pc}
  \statelemshifts
\end{lemma}

Then, we are able to combine the shifts in Lemmas~\ref{cor:implement-approx:real-with-complex} and~\ref{lem:implement-approx:pc} via a transitivity property for polynomial-time approximate shifts (see Lemma~\ref{lem:approx-shifts:transitivity} in Section~\ref{sec:approx-shifts}), and therefore prove Theorem~\ref{thm:approx-shifts}, see Section~\ref{sec:complex-implementations} for the details.

\newcommand{\statelemtransitivity}{Let $q \in \mathsf{P}_{\mathbb{C}}$ with $q \ne 0$ and let $(x_j, y_j) \in \mathcal{H}_q$ for each $j \in \{1,2,3\}$. Let us assume that the following hypotheses hold:
  \begin{enumerate}
    \item $x_2$ and $y_2$ are polynomial-time computable;
    \item $y_2 \not \in \{1\} \cup \left(1-q/2 + i q \mathbb{R}\right)$;
    \item there is a polynomial-time approximate shift from $(x_1, y_1)$ to $(x_2, y_2)$;
    \item there is a polynomial-time approximate theta shift from $(x_2, y_2)$ to $(x_3, y_3)$.
  \end{enumerate}
  Then there is a polynomial-time approximate shift from $(x_1, y_1)$ to $(x_3, y_3)$. Moreover, if the polynomial-time approximate shift from $(x_1, y_1)$ to $(x_2, y_2)$ is series-parallel, then the polynomial-time approximate shift from $(x_1, y_1)$ to $(x_3, y_3)$ is also series-parallel.}

\subsection{The reductions}

In Section~\ref{sec:hardness:approximate-shifts} we show how to use a polynomial-time approximate shift from $(x_1, y_1)$ to $(x_2, y_2)$ to reduce the problem of approximating the Tutte polynomial at $(x_2, y_2)$ to the same problem at $(x_1, y_1)$. The following lemma gives such a reduction for the problem of approximating the norm, we also give an analogous result for approximating the argument.

\newcommand{\statelemredout}{Let $q\neq 0$, $\gamma_1$ and $\gamma_2\neq 0$ be algebraic numbers, and $K>1$. 
For $j\in\{1,2\}$, 
let $y_j = \gamma_j+1$ and $x_j = 1+q/\gamma_j$. If there is a polynomial-time series-parallel approximate shift from $(x_1, y_1)$ to $(x_2, y_2)$, then we have a reduction from $\FNT(q, \gamma_2)$ to $\FNT(q, \gamma_1)$. This reduction also holds for the planar version of the problem.}
\begin{lemma} \label{lem:reduction:norm:outline}
\statelemredout
\end{lemma}

In order to prove Lemma~\ref{lem:reduction:norm:outline}, we need some lower bounds on the norm of the partition function $Z_{\mathrm{Tutte}}(G; q, \gamma)$. This kind of lower bound plays an important role in several hardness results on the complexity of approximating partition functions~\cite{Goldberg2017, Bezb}. Here, we have to work a bit harder than usual since we have two (algebraic) underlying parameters (in the case of Tutte), and we need to use results in algebraic number theory, see Section~\ref{sec:hardness:lower-bound} for details. 

By combining Theorem~\ref{thm:approx-shifts} and Lemma~\ref{lem:reduction:norm:outline} with existing hardness results, we obtain our hardness results for non-real edge interactions in Section~\ref{sec:mainproof}. On the way, we collect some hardness on real parameters as well that strengthen previous results in the literature, and  part of Section~\ref{sec:hardness} is devoted to this. The main reason behind these improvements is that previous work on real parameters used reductions from  approximately counting minimum cardinality $(s,t)$-cuts~\cite{Goldberg2014, Goldberg2017}, the minimum 3-way cut problem~\cite{Goldberg2008}, or maximum independent set for planar cubic graphs~\cite{Goldberg2012}, which are either easy on  planar graphs or the parameter regions they cover are considerably smaller or cannot be used to conclude $\numP$-hardness. We instead reduce the exact computation of $Z_{\mathrm{Tutte}}(G; q, \gamma)$ to its approximation,  which has the advantage that the problem that we are reducing from is $\numP$-hard for planar graphs~\cite{Vertigan2005}.  Interestingly, our reduction requires us to apply an algorithm of Kannan, Lenstra and Lov\'asz~\cite{Kannan1988} to reconstruct the minimal polynomial of an algebraic number from an additive approximation of the number. The lower bounds on the partition function $Z_{\mathrm{Tutte}}(G; q, \gamma)$ that are gathered in Section~\ref{sec:hardness:lower-bound} also play a role in this reduction, the details will be given in Section~\ref{sec:hardness:real}.

 \section{Preliminaries} \label{sec:pre}

\subsection{The multivariate Tutte polynomial}

  The random cluster formulation of the multivariate Tutte polynomial is particularly convenient when working with implementations (as we will see in Section~\ref{sec:pre:weights}), and is defined as follows. Let $G = (V, E)$ be a graph. For any $\gamma \colon E \to \mathbb{C}$ and $q \in \mathbb{C}$, the \emph{multivariate Tutte polynomial of $G$} is
  \begin{equation} \label{eq:tutte:def:general}
    Z_{\text{Tutte}}(G; q, \gamma) = \sum_{A \subseteq E} q^{k(A)} \prod_{e \in A} \gamma_e.
  \end{equation}
  We will make use of the following notation. Let $s$ and $t$ be two distinct vertices of G. We define
  \begin{equation*}
    Z_{st}(G; q, \gamma) = \sum_{\substack{A \subseteq E: \\ s \text{ and } t \text{ in the same component}}} q^{k(A)} \prod_{e \in A} \gamma_e.
  \end{equation*}
  Analogously, let $Z_{s\vert t}$ be the contribution to $Z_{\text{Tutte}}(G; q, \gamma)$ from the configurations $A \subseteq E$ such that $s$ and $t$ are in different connected components in $(V, A)$. That is, $Z_{s \vert t}(G; q, \gamma) = Z_{\text{Tutte}}(G; q, \gamma) - Z_{st}(G; q, \gamma)$.

\subsection{Implementing weights, series compositions and parallel compositions} \label{sec:pre:weights}

  In this section, we define implementations, shifts, series compositions and parallel compositions. The definitions and results that we give are standard and can also be found, for instance, in~\cite[Section 4]{Jaeger1990}, ~\cite[Section 2.1]{Goldberg2012} or \cite[Section 4]{Sokal2005}.

  Let $q \in \mathbb{C}$ with $q \ne 0$. The value of $q$ is fixed across all this section. Let $H$ be a weighted graph with weight function $\hat{\gamma}$. Let $s$ and $t$ be two distinct vertices of $H$, which are usually referred to as terminals. 
  We say that the graph  $H$ $\hat{\gamma}$-\emph{implements} the weight $w$ with respect to the terminals $s$ and $t$ if
  \begin{equation*}
    w = q \frac{ Z_{st}(H; q, \hat{\gamma})}{ Z_{s|t}(H; q, \hat{\gamma})}.
  \end{equation*}
  We say that $H$ $\hat{\gamma}$-implements the weight $w$ if there are terminals $s$ and $t$ such that $H$ $\hat{\gamma}$-implements the weight $w$ with respect to $s$ and $t$. These definitions are motivated by Lemma~\ref{lem:implementations:multivariate}, whose proof is a straightforward computation involving the definitions of implementations and the multivariate Tutte polynomial.

  \begin{lemma}[{\cite[Equation 2.2]{Goldberg2012}}] \label{lem:implementations:multivariate}
    Let $G$ and $H$ be two graphs with weight functions $\gamma$ and $\hat{\gamma}$ respectively. Let $f$ be an edge of $G$ with weight $\gamma_f$ such that $H$ $\hat{\gamma}$-implements $\gamma_f$ with respect to terminals $s$ and $t$. Let $G_f$ be the graph constructed by considering the union of $G$ and $H$, identifying the terminals $s$ and $t$ with the endpoints of $f$ in $G$ and removing $f$. Let $\gamma'$ be the weight function on $G_f$ that inherits the weights from $\gamma$ and $\hat{\gamma}$. Then
    \begin{equation*}
      Z_{st}\left(G_f; q, \gamma'\right)  = \frac{Z_{s | t}\left(H; q, \hat{\gamma}\right)}{q^2} Z_{st}\left(G; q, \gamma\right), \quad  Z_{s|t}\left(G_f; q, \gamma'\right)  = \frac{Z_{s | t}\left(H; q, \hat{\gamma}\right)}{q^2} Z_{s | t}\left(G; q, \gamma\right).
    \end{equation*}
    In particular, we have $Z_{\text{Tutte}}\left(G_f; q, \gamma'\right) = \frac{Z_{s | t}\left(H; q, \hat{\gamma}\right)}{q^2} Z_{\text{Tutte}}\left(G; q, \gamma \right)$. Moreover, if $G$ $\gamma$-implements a weight $w$, then $G_f$ also $\gamma'$-implements $w$.
  \end{lemma}

  Therefore, if we can compute $Z_{s \mid t}(H; q, \hat{\gamma})$ efficiently, then computing $Z_{\text{Tutte}}(G; q, \gamma)$ is as hard as computing $Z_{\text{Tutte}}(G_f; q, \gamma')$. This observation leads to some of the reductions that appear in this paper.

  In the remaining sections we usually assume that the weights are constant, that is, each edge of the graph has the same weight, and we will make it clear when this is not the case. In the constant weight function case Lemma~\ref{lem:implementations:multivariate} can be applied to each edge of the graph constructed by copying $G$ and substituting each edge  $f$ in $G$ by a copy of $H$ (identifying the endpoints of $f$ with $s$ and $t$). Let $\alpha_1, \alpha_2 \in \mathbb{C}$. We say that there is a \emph{shift} from $(q, \alpha_1)$ to $(q, \alpha_2)$ if there is a graph $H$ that $\alpha_1$-implements $\alpha_2$. An important property of shifts is transitivity; if there are shifts from $(q, \alpha_1)$ to $(q, \alpha_2)$ and from $(q, \alpha_2)$ to $(q, \alpha_3)$, then there is a shift from $(q, \alpha_1)$ to $(q, \alpha_3)$. This is a consequence of Lemma~\ref{lem:implementations:multivariate}. Let $y_1 = \alpha_1 + 1$ and $y_2 = \alpha_2 + 1$. We define $x_1$ and $x_2$ by $q = (x_1-1) (y_1-1) = (x_2-1)(y_2-1)$, which is the change of variables that relates the Tutte polynomial and $Z_{\text{Tutte}}$. We equivalently refer to the shift from $(q, \alpha_1)$ to $(q, \alpha_2)$ as a shift from $(x_1, y_1)$ to $(x_2, y_2)$, and we also say that $H$ $(x_1,y_1)$-implements $(x_2, y_2)$. This notation is convenient to express many of the shifts considered in this paper.

  To conclude this section we introduce two tools that will provide us with many examples of implementations and shifts: parallel compositions and series compositions. For each $j \in \{1,2\}$, let $G_j$ be a graph, let $s_j$ and $t_j$ be two terminals of $G_j$, and let $\gamma_j$ be a weight function such that $G_j$ $\gamma_j$-implements a weight $w_j$ with respect to $s_j$ and $t_j$.

\emph{Parallel compositions.} The parallel composition of $(G_1, s_1, t_1)$ and $(G_2, s_2, t_2)$ is the graph $G$ constructed by considering the union of $G_1$ and $G_2$ and identifying $s_1$ with $s_2$ and $t_1$ with $t_2$. Let $\hat{\gamma}$ be the weight function on $G$ inherited from $\gamma_1$ and $\gamma_2$. It is well-known and easy to check that $G$ $\hat{\gamma}$-implements the weight
    \begin{align*}
      w & = (1 + w_1)(1 + w_2) - 1
    \end{align*}
    with respect to the terminals $s_1$ and $t_1$. Let $(x_1, y_1)$ and $(x_2, y_2)$ be the Tutte coordinates of $(q, w_1)$ and $(q, w_2)$ respectively (so $y_j = w_j + 1$ and $(x_j-1)(y_j-1) = q$). Then the Tutte coordinates of $(q, w)$ are $(x, y)$ with $y = y_1 y_2$ and $(x-1)(y-1) = q$. Let $\Upsilon$ be a graph with two vertices $s,t$ and one edge joining them, and let $\Upsilon^n$ be the parallel composition of $n$ copies of $(\Upsilon, s,t$) (so $\Upsilon^n$ has two vertices and $n$ edges joining them). Then $\Upsilon^n$ $(x,y)$-implements $(x', y')$ with $y' = y^n$ and $(x'-1)(y'-1) = q$. This is known as an \emph{$n$-thickening} of $(x,y)$ and it yields a shift from $(x,y)$ to $(x', y^n)$. \vskip 0.2cm

\emph{Series compositions.} The series composition of $(G_1, s_1, t_1)$ and $(G_2, s_2, t_2)$ is the graph $G$ constructed by considering the union of $G_1$ and $G_2$ and identifying $t_1$ with $s_2$. Let $\hat{\gamma}$ be the weight function on $G$ inherited from $\gamma_1$ and $\gamma_2$. It is well-known and easy to check that $G$ $\hat{\gamma}$-implements the weight
    \begin{align*}
      w & = \frac{w_1 w_2}{w_1 + w_2 + q}
    \end{align*}
    with respect to the terminals $s_1$ and $t_2$. Note that $w$ satisfies
    \begin{equation} \label{eq:stretching:coordinates}
      \left( 1 + \frac{q}{w} \right) = \left( 1 + \frac{q}{w_1} \right) \left( 1 + \frac{q}{w_2} \right).
    \end{equation}
     Let $(x_1, y_1)$ and $(x_2, y_2)$ be the Tutte coordinates of $(q, w_1)$ and $(q, w_2)$ respectively (so $y_j = w_j + 1$ and $(x_{j}-1)(y_{j}-1) = q$). Then, in view of \eqref{eq:stretching:coordinates}, the Tutte coordinates of $(q, w)$ are $(x, y)$ with $x = x_1 x_2$ and $(x-1)(y-1) = q$. Let $\Upsilon$ be a graph with two vertices $s,t$ and one edge joining them, and let $\Upsilon_n$ be the series composition of $n$ copies of $(\Upsilon, s, t$) (so $\Upsilon_n$ is a path graph with $n$ edges). Then $\Upsilon_n$ $(x,y)$-implements $(x', y')$ with $x' = x^n$ and $(x'-1)(y'-1) = q$. This is known as an \emph{$n$-stretching} of $(x,y)$ and it yields a shift from $(x,y)$ to $(x^n, y')$.

For series-parallel and theta graphs (see Footnote~\ref{foot:series}), these constructions give that either $Z_{s \mid t}(G; q, \gamma) = 0$, or 
the series-parallel graph~$G$ 
(with terminals~$s$ and~$t$)
$\gamma$-implements a weight $w(G, s, t; q, \gamma)$ that can be computed from the recursive definition of series-parallel graphs in polynomial time.  In particular, let $\Theta_{(l_1, \dots, l_{m})}$ be the theta graph with $m$ internal paths of lengths $l_1, \ldots, l_m$.   In  this case,\footnote{We should mention that we will make use of the $\Theta$ asymptotic notation in this paper and this notation should not be confused with that of theta graphs.} we   have that
\begin{equation} \label{eq:theta:implement}
  w\left(\Theta_{(l_1, \dots, l_{m})}, s, t; q, \gamma \right) = \prod_{j = 1}^m \left( 1 + \frac{q}{x^{l_j} - 1}\right) - 1,
\end{equation} 
where $x = 1 + q / \gamma$.  
Series-parallel graphs can be built using  series and parallel compositions.
The following definition is equivalent to the one in Footnote~\ref{foot:series}.
A graph $G$ is \emph{series-parallel} (with terminals $s$ and $t$) if either $G$ is the graph with two vertices $s$ and $t$ and one edge joining them, or $G$ is the parallel or series composition of $(G_1, s_1, t_1)$ and $(G_2, s_2, t_2)$, where $s = s_1$, $t= t_2$ and $G_j$ is a series-parallel graph  with terminals $s_j$ and $t_j$~\cite[Chapter 11]{Brandstadt1999}. 

Finally, the \emph{size of a graph} $G = (V, E)$ is the integer $\mathrm{size}(G) = \left| V \right| + \left| E \right|$. Note that the size of $\Theta_{(l_1, \dots, l_{m})}$ is $2 \sum_{j = 1}^m l_j - m + 2$.

\section{Polynomial-time approximate shifts} \label{sec:approx-shifts}

Implementing a specific weight cannot always be achieved. Nonetheless, sometimes we can implement an approximation of  the desired weight with as much precision as we need. These implementations have been exploited several times in the literature on Tutte polynomials and the Ising model; see~\cite{Goldberg2008, Goldberg2012, Goldberg2014, Goldberg2019}. Here we collect some of these results appearing in~\cite{Goldberg2014}, which in turn are based on arguments in~\cite{Goldberg2012}; here, we follow the presentation in~\cite{Goldberg2019} (that was stated for $q=2$).

\begin{lemma}[{\cite[Lemma 22]{Goldberg2019}}, {\cite[Lemma 5]{Goldberg2014}}] \label{lem:implement-approx:literature:1} Let $x$ and $y$ be real algebraic numbers such that $y \not\in [-1, 1]$ and $(x-1) (y-1) = q > 0$. There is a polynomial-time algorithm that takes as an input:
    \begin{itemize}
    \item two positive integers $n$ and $k$, in unary;
    \item a real algebraic number $y' \in [1, |y|^{k}]$.
    \end{itemize}
    This algorithm produces a theta graph $G$ that $(x,y)$-implements $(\hat{x}, \hat{y})$ such that $\big| y' - \hat{y} \big| \le 2^{-n}$. The size of $G$ is at most a polynomial in $n$ and $k$, independently of $y'$.
  \end{lemma}
  
  In Lemma~\ref{lem:implement-approx:pc} (Section~\ref{sec:complex-implementations}), we give a similar result to Lemma~\ref{lem:implement-approx:literature:1} where the numbers $x$ and $y$ may not be algebraic. The fact that the graph $G$ computed in Lemma~\ref{lem:implement-approx:literature:1} is a theta graph is not directly stated in the statement of~\cite[Lemma 5]{Goldberg2014} but it can easily be inferred from the proof. This also applies to Lemma~\ref{lem:implement-approx:literature:2}.

  \begin{lemma}[{\cite[Lemma 22]{Goldberg2019}}, {\cite[Lemma 7]{Goldberg2014}}] \label{lem:implement-approx:literature:2} Let $x_1, y_1, x_2, y_2$ be real algebraic numbers such that $y_1 \in (-1,1)$, $y_2 \not\in [-1, 1]$, and $(x_1-1) (y_1-1) = (x_2-1) (y_2-1) = q < 0$. There is a polynomial-time algorithm that takes as an input:
    \begin{itemize}
    \item two positive integers $n$ and $k$, in unary;
    \item a real algebraic number $y' \in [1, |y_1|^{-k}]$.
    \end{itemize}
    This algorithm produces a theta graph $G = (V, E)$ and a weight function $\hat{\gamma} \colon E \to \{y_1-1, y_2-1\}$ such that $G$ $\hat{\gamma}$-implements $(\hat{x}, \hat{y})$ with $\big| y' - \hat{y} \big| \le 2^{-n}$. The size of $G$ is at most a polynomial in $n$ and $k$, independently of $y'$.
  \end{lemma}
  
  \begin{corollary}\label{cor:implement-approx:literature}
    Let $x_1, y_1, x_2, y_2$ be real algebraic numbers such that $y_1 \in (-1,0)\cup (0,1)$, $y_2 \not\in [-1, 1]$, $(x_1-1) (y_1-1) = (x_2-1) (y_2-1) = q$, $q \ne 0$. There is a polynomial-time algorithm that takes as an input:
    \begin{itemize}
    \item two positive integers $n$ and $k$, in unary;
    \item a positive real algebraic number $y'$ such that $|y'| \in [|y_1|^{k}, |y_1|^{-k}]$.
    \end{itemize}
    This algorithm produces a theta graph $G = (V, E)$ and a weight function $\hat{\gamma} \colon E \to \{y_1-1, y_2-1\}$ such that $G$ $\hat{\gamma}$-implements $(\hat{x}, \hat{y})$ with $\big| y' - \hat{y} \big| \le 2^{-n}$. The size of $G$ is at most a polynomial in $n$ and $k$, independently of $y'$. Moreover, if either $y_1 < 0$ or $y_2 < 0$, then the restriction that $y'$ is positive can be replaced with a restriction that $y'$ is non-zero.
  \end{corollary}  
  \begin{proof}
    This result easily follows from Lemmas~\ref{lem:implement-approx:literature:1} and~\ref{lem:implement-approx:literature:2} by an argument of Goldberg and Jerrum (see the proof of~\cite[Lemma 2]{Goldberg2014}). We include here their argument for completeness. The case when $y' \ge 1$ has been covered in Lemmas~\ref{lem:implement-approx:literature:1} and~\ref{lem:implement-approx:literature:2}. First, let us assume that $y' \in (0,1)$. We have $1 \le y' \cdot y_1^{-2k} \le |y_1|^{-2k}$ and using Lemmas~\ref{lem:implement-approx:literature:1} and~\ref{lem:implement-approx:literature:2} we can implement $\tilde{y}$ with $| \tilde{y} - y' \cdot y_1^{-2k} | \le 2^{-n}$. We have $| y_1^{2k} \tilde{y} - y' | \le 2^{-n}$, so we set $\hat{y} = y_1^{2k} \tilde{y}$. The graph $G$ is the parallel composition of the graph used to implement $ \tilde{y}$ and $2k$ edges with weight $y_1$. Finally, let us assume that there is $i \in \{1, 2\}$ such that $y_i < 0$, and let us consider the case where $y'$ is negative. We implement an approximation $\hat{y}'$ of $y' / y_i > 0$, and return $\hat{y} = \hat{y}' y_i$.
  \end{proof}

  The graphs $G$ produced by the algorithms given in Lemma~\ref{lem:implement-approx:literature:1}, Lemma~\ref{lem:implement-approx:literature:2} and Corollary~\ref{cor:implement-approx:literature} are theta graphs. One may wonder which weights can be approximated as in these results. This leads to the following definition. Let $(x_1,y_1), (x_2,y_2) \in \mathcal{H}_q$. Let $\gamma_1 = y_1-1$ and $\gamma_2 = y_2 - 1$. We say that there is a \emph{polynomial-time approximate shift} from $(q, \gamma_1)$ to $(q, \gamma_2)$ or, equivalently, from $(x_1,y_1)$ to $(x_2, y_2)$, if there is an algorithm that, for any positive integer $n$, computes  in polynomial time in $n$ a graph $G_n$ that $(x_1, y_1)$-implements $(\hat{x}_2, \hat{y}_2)$ with $\left| y_2 - \hat{y}_2 \right| \le 2^{-n}$. If the graph $G_n$ computed by this algorithm is always a theta graph (resp. a series-parallel graph), then we say that this is a \emph{polynomial-time approximate theta shift} (resp. \textit{polynomial-time approximate series-parallel shift}). Lemma~\ref{lem:implement-approx:literature:1} gives polynomial-time approximate theta shifts from $(x_1, y_1)$ to $(x_2, y_2)$ when the considered numbers are real algebraic, $y_1 \not\in [-1,1]$, $y_2 \in [1, \infty)$ and $q > 0$. Note that shifts are a particular case of polynomial-time approximate shifts. Moreover, due to the transitivity property of shifts, if there is a shift from $(x_1, y_1)$ to $(x_2, y_2)$ and there is a polynomial-time approximate shift from $(x_2, y_2)$ to $(x_3, y_3)$, then there is a polynomial-time approximate shift from $(x_1, y_1)$ to $(x_3, y_3)$. In fact, polynomial-time approximate shifts exhibit some of the properties of shifts;
  in Lemma~\ref{lem:approx-shifts:composition} we show that they behave well with respect to parallel and series compositions and
  in Lemma~\ref{lem:approx-shifts:transitivity} we show that they are transitive under certain conditions.  In Section~\ref{sec:complex-implementations} we give more examples of polynomial-time approximate shifts, some of which will be constructed by transitivity. These approximate shifts play an important role in our hardness proofs.

  \begin{lemma} \label{lem:approx-shifts:composition}
    Let $q \in \mathbb{C} \setminus \{0\}$ and let $(x_j, y_j) \in \mathcal{H}_q$ for each $j \in \{1,2,3\}$. Let us assume that there are polynomial-time approximate shifts from $(x_1, y_1)$ to $(x_2, y_2)$, and from $(x_1, y_1)$ to $(x_3, y_3)$.
    Let $(x_4, y_4), (x_5, y_5) \in \mathcal{H}_q$ with $y_4 = y_2 y_3$ and $x_5 = x_2 x_3$. Then:
    \begin{enumerate}
    \item \label{it:f4f545} there is a polynomial-time approximate shift from $(x_1, y_1)$ to $(x_4, y_4)$;
    \item \label{it:f4f545b} there is a polynomial-time approximate shift from $(x_1, y_1)$ to $(x_5, y_5)$.
    \end{enumerate}
    Moreover, if the polynomial-time approximate shifts from $(x_1, y_1)$ to $(x_2, y_2)$ and  $(x_3, y_3)$ are series-parallel, then the obtained polynomial-time approximate shifts are also series-parallel.
  \end{lemma}
  \begin{proof}
    For $j \in \{2,3\}$, let $G_{n,j}$ be the graph computed by the polynomial-time approximate shift from $(x_1, y_1)$ to $(x_j, y_j)$, so $G_{n,j}$ $(x_1, y_1)$-implements $(\hat{x}_j, \hat{y}_j)$ with $\left| y_j - \hat{y}_j \right| \le 2^{-n}$, for certain terminals $t_j$ and $s_j$.

For Item~\ref{it:f4f545}, let $P_n$ be the parallel composition of $(G_{n,2}, s_2, t_2)$ and $(G_{n,3}, s_3, t_3)$. The graph $P_n$ gives a shift from $(x_1, y_1)$ to $(\hat{x}_4, \hat{y}_2 \hat{y}_3) \in \mathcal{H}_q$. Since $\left| y_3 - \hat{y}_3 \right| \le 2^{-n}$, we have $\left| \hat{y}_3 \right| \le \left| y_3 \right| + 1$ and
      \begin{align*}
        \left| y_2 y_3 - \hat{y}_2 \hat{y}_3 \right| & \le \left| y_2 - \hat{y}_2 \right| \left| \hat{y}_3 \right| + \left| y_3 - \hat{y}_3 \right| \left| y_2 \right| \le 2^{-n} \left( \left| y_3 \right| + 1 + \left| y_2 \right| \right).
      \end{align*}
      Therefore, for $k$ large enough, the graphs $P_{n+k}$ give  a polynomial-time approximate shift from $(x_1, y_1)$ to $(x_4, y_4)$ with $y_4 = y_2 y_3$.

For Item~\ref{it:f4f545b}, the proof is analogous but now we define the graph $S_n$ as the series composition of $(G_{n,2}, s_2, t_2)$ and $(G_{n,3}, s_3, t_3)$, which gives a shift from $(x_1, y_1)$ to $(\hat{x}_2 \hat{x}_3, \hat{y}_4) \in \mathcal{H}_q$. 

   Note that if the original polynomial-time approximate shifts are series-parallel, then the obtained ones are also series-parallel by the definition of series-parallel graphs.
  \end{proof}
  
  When it comes to hardness results, we are only interested in algebraic numbers. However, we will have to consider polynomial-time approximate shifts from $(x_1, y_1)$ to $(x_2, y_2)$ such that the numbers involved are not algebraic. This is due to the fact that, even if $x_1$ and $y_1$ are algebraic, $x_2$ and $y_2$ might not be. Nonetheless, in that case we can ensure that $x_2$ and $y_2$ are polynomial-time computable. A real number $x$ is \emph{polynomial-time computable} if there is a function $\phi \colon \mathbb{N} \to \mathbb{Q}$ that is computable in polynomial time (with the input written in unary notation, i.e., $0^n$) such that $\left| x - \phi(n) \right| \le 2^{-n}$ for all $n \in \mathbb{N}$~\cite[Chapter 2]{Ko1991}. The definition given in~\cite[Chapter 2]{Ko1991} uses dyadic rational numbers instead of rational numbers, but these two definitions are easily seen to be equivalent. We denote the set of polynomial-time computable real numbers by $\mathsf{P}_{\mathbb{R}}$. One can easily show that the set $\mathsf{P}_{\mathbb{R}}$ is a field. Real algebraic numbers are in $\mathsf{P}_{\mathbb{R}}$ because we can approximate them as closely as we want by applying Sturm sequences and binary search~\cite{Emiris2008}.  We say that a complex number $z$ is \emph{polynomial-time computable} if $z = x + i y$ for some $x, y \in \mathsf{P}_{\mathbb{R}}$. We denote the set of polynomial-time computable complex numbers by $\mathsf{P}_{\mathbb{C}}$. Algebraic numbers are in $\mathsf{P}_{\mathbb{C}}$ (their real and imaginary parts are real algebraic numbers). It turns out that $\mathsf{P}_{\mathbb{C}}$ is an algebraically-closed field~\cite[Chapter 2]{Ko1991}. In particular, for $z \in \mathsf{P}_{\mathbb{C}}$, we have $\left| z \right| \in \mathsf{P}_\mathbb{R}$.

  If there is a polynomial-time approximate theta shift from $(x_1, y_1)$ to $(x_2, y_2)$ and $x_1$ and $y_1$ are algebraic, then we can compute in polynomial time in $n$ an algebraic number that additively approximates $y_2$ up to an additive error $2^{-n}$. Since we can approximate algebraic numbers by rational numbers efficiently, it follows that $x_2$ and $y_2$ are polynomial-time computable. However, if we only know that $x_1$ and $y_1$ are polynomial-time computable, then it is not clear if $x_2$ and $y_2$ are polynomial-time computable or not. Lemma~\ref{lem:pc:essential} gives a partial answer to this question and plays a key role in our transitivity result for polynomial-time approximate shifts (Lemma~\ref{lem:approx-shifts:transitivity}). First, we need to prove some lemmas on polynomial-time computable numbers.
   
\begin{lemma} \label{lem:pr:bounds}
  Let $z \in \mathsf{P}_{\mathbb{C}}$. There is an algorithm that computes $b_1 \in \mathbb{Q}$ with $|z| \le b_1$. Moreover, if $z \ne 0$, then there is an algorithm that computes $b_2 \in \mathbb{Q}$ with $0 < b_2 \le |z|$.
\end{lemma}
\begin{proof}
  Let $x = \left| z \right|$. From $x \in \mathsf{P}_{\mathbb{R}}$, it follows that we can compute a sequence $\hat{x}_n \in \mathbb{Q}$ such that $\left| x - \hat{x}_n \right| \le 2^{-n}$, that is, we have $x \in [ \hat{x}_n - 2^{-n}, \hat{x}_n + 2^{-n} ]$. This computation for $n = 1$ gives the upper bound $\hat{x}_1 + 1/2$. Note that the  sequences $\hat{x}_n-2^{-n}$ and $\hat{x}_n + 2^{-n}$ converge to $x$. Hence, if $x \ne 0$, then there must be $n$ such that $0 < \hat{x}_n - 2^{-n} \le x$. We compute $\hat{x}_n$ until this inequality happens, obtaining the desired lower bound.
\end{proof}

\begin{lemma} \label{lem:pr:stretching}
  Let $z \in \mathsf{P}_{\mathbb{C}}$ with $\left| z \right| \ne 1$. There is a polynomial-time algorithm that takes as inputs two positive integers $n$ and $k$ and computes a positive integer $r(n, k)$ such that
  \begin{enumerate}
    \item \label{item:pr:stretching:1} $r(n, k)$ is increasing in $k$;
    \item \label{item:pr:stretching:2} $r(n, k) = n + \Theta (k)$;
    \item \label{item:pr:stretching:3} if $\left| z - \hat{z} \right| \le 2^{-r(n, k)}$, then $\left| \frac{1}{z^k - 1}  - \frac{1}{\hat{z}^k - 1} \right| \le 2^{-n}$.
  \end{enumerate}
\end{lemma}

\begin{proof}
  By Lemma~\ref{lem:pr:bounds}, we can compute an integer $t\geq 0$ such that $ 2^{-t} \le ||z| -1|$ and $|z| \le 2^t$. Note that for every integer $k\geq 1$ we have the bound $ \big| z^k - 1 \big| \ge 2^{-t}$. Indeed, if $|z| < 1$, then 
  \begin{equation*}
    2^{-t} \le 1 - |z| \le 1 - |z|^k \le \big|z^k - 1\big|  
  \end{equation*}
  and when $|z| > 1$, we analogously find that $2^{-t} \le |z| - 1 \le |z|^k - 1 \le \big| z^k - 1 \big|$. 

  Let $n$ and $k$ be the inputs of our algorithm. Let $r(n,k) = n+(t+1)(k+1)$, and note that $r$ is increasing in $k$ and $r(n,k) = n + \Theta \left( k \right)$, establishing Items~\ref{item:pr:stretching:1} and~\ref{item:pr:stretching:2}. 
	
	For Item~\ref{item:pr:stretching:3}, consider $\hat{z}$  such that $\left| z - \hat{z} \right| \leq 2^{-r(n,k)}$.  Since $|\hat{z}| \le |z| + 2^{-r(n, k)} \le 2^{t+1}$, for every $j\in \{0,\hdots,k-1\}$ we have $\left| \hat{z} \right|^j \left| z \right|^{k - 1 - j} \le 2^{t(k-1)+j}$ and hence
  \begin{align*}
    \big| z^k - \hat{z}^k \big| & = \Big| (z - \hat{z}) \sum\nolimits_{j = 0}^{k-1} \hat{z}^{j} z^{k - 1 - j}  \Big|  \le \left| z - \hat{z} \right|\sum\nolimits_{j = 0}^{k-1} \left| \hat{z} \right|^j \left| z \right|^{k-1-j} \\
                                   & \le \left| z - \hat{z} \right| \sum\nolimits_{j = 0}^{k-1} 2^{t(k-1)+j}  < \left| z - \hat{z} \right| 2^{t(k-1)+k} \le 2^{-(n+2t+1)}. 
  \end{align*}
Moreover, we have that $ \big| \big| z^k - 1 \big| - \big| \hat{z}^k - 1 \big| \big| \le \big| z^k - \hat{z}^k  \big| < 2^{-(t+1)}$  and, thus,
  \begin{equation*}
    \big| \hat{z}^k - 1 \big|  \ge \big| z^k - 1 \big|  - 2^{-(t+1)} \ge 2^{-(t+1)},
  \end{equation*}
  where we used that $ \big| z^k - 1 \big| \ge 2^{-t}$. Therefore, we find that
  \begin{equation*}
    \left| \frac{1}{z^k - 1} - \frac{1}{\hat{z}^k -1} \right| = \left| \frac{z^k - \hat{z}^k}{(z^k -  1)(\hat{z}^k -1)} \right| \le 2^{2t+1} \big| z^k - \hat{z}^k \big| \le 2^{-n}. \qedhere
  \end{equation*}
\end{proof}

\begin{lemma} \label{lem:pc:essential}  Let $q \in \mathsf{P}_{\mathbb{C}}$ with $q \ne 0$ and let $\gamma \in \mathsf{P}_{\mathbb{C}}$ with $\gamma \not \in \{0\} \cup  -q/2 + i q \mathbb{R}$. There is a polynomial-time algorithm that takes as an input:
  \begin{itemize}
  \item a positive integer $n$;
  \item a theta graph $G = \Theta_{(l_1, \ldots, l_m)}$ with terminals $s$ and $t$.
  \end{itemize}
  This algorithm computes $f(n, G)$ such that
  \begin{enumerate}
    \item $f(n,G) = n + \Theta \left( \mathrm{size}(G) \right)$;
    \item for any $\hat{\gamma}$ with $\left| \gamma - \hat{\gamma} \right| \le 2^{-f(n,G)}$, we have $\left| w(G, s, t; q, \gamma) - w(G, s, t; q, \hat{\gamma}) \right| \le 2^{-n}$.
  \end{enumerate}
\end{lemma}
\begin{proof}
  Let $y = \gamma + 1$ and $x = 1 + q / \gamma$. Note that $\left| x \right| = 1$ if and only if $\left| \gamma + q \right| = \left| \gamma \right|$. By basic geometry, the latter statement is equivalent to $\gamma \in -q/2 + i q \mathbb{R}$. Hence, by hypothesis, $\left| x \right| \ne 1$. There are two cases:
  \begin{itemize}
    \item $\left| x \right| < 1$. Then for any positive integer $k$ we have
      \begin{equation*}
        \left| 1 + \frac{q}{x^k - 1} \right| \le 1 + \frac{\left| q \right|}{1 - \left| x \right|^k} \le 1 + \frac{\left| q \right|}{1 - \left| x \right|} =  1 + \frac{\left| q \right|}{\left| 1 - \left| x \right|\right|} .
      \end{equation*}
    \item $\left| x \right| > 1$. Then for any positive integer $k$ we have
      \begin{equation*}
        \left| 1 + \frac{q}{x^k - 1} \right| \le 1 + \frac{\left| q \right|}{\left| x \right|^k - 1} \le 1 + \frac{\left| q \right|}{\left| x \right| - 1} = 1 + \frac{\left| q \right|}{\left| 1 - \left| x \right| \right|} .
      \end{equation*}
  \end{itemize}
  Since $q, x \in \mathsf{P}_{\mathbb{C}}$, we can apply Lemma~\ref{lem:pr:bounds} along with the above bounds to compute a non-negative integer $t_x$ such that $\left| 1 + q / (x^k - 1) \right| \le 2^{t_x}$ for every positive integer $k$. Lemma~\ref{lem:pr:bounds} also allows us to compute non-negative integers $t_q$ and $t_\gamma$ such that $\left| q \right| \le 2^{t_q}$ and $2^{-t_\gamma} \le \left| \gamma \right|$.

  Let $n$ and $G = \Theta_{(l_1, \ldots, l_m)}$ be the inputs of our algorithm. Let $k = \max \{l_1, \ldots, l_m\}$. Since $\left| x \right| \ne 1$, we can compute $g(n, G) = r(n + (t_x+1)(m+1)+t_q, k)$, where $r$ is as in Lemma~\ref{lem:pr:stretching} for the polynomial-time computable number $x$. We compute $f(n, G) = g(n, G)  + t_q + 2 t_\gamma + 1$. We claim that $f$ satisfies the statement. In view of the properties of $r$, we have
  \begin{equation*}
    f(n, G) = g(n, G) + \Theta \left( 1 \right) = n + \Theta \left( \mathrm{size}(G) \right).
  \end{equation*}  
  We define $y_j = 1 + q / \left( x^{l_j} - 1 \right)$ for every $j \in \{1, \ldots, m\}$. Recall that in \eqref{eq:theta:implement} we argued that 
  \begin{equation*}
    w(G, s, t; q, \gamma) = \prod_{j = 1}^m y_j - 1.  
  \end{equation*}
  Let $\hat{\gamma}$ with $\left| \gamma - \hat{\gamma} \right| \le 2^{-f(n, G)}$. Let $\hat{y} = \hat{\gamma} + 1$ and $\hat{x} = 1 + q / (\hat{y}-1)$. Then  
  \begin{equation*}
    w(G, s, t; q, \hat{\gamma}) = \prod_{j = 1}^m \hat{y}_j - 1,  
  \end{equation*}
  where $\hat{y}_j = 1+ q / \left( \hat{x}^{l_j} - 1 \right)$. Since $\left| \gamma - \hat{\gamma} \right| \le 2^{-f(n, G)} \le 2^{-t_\gamma - 1}$, we have $\left| \hat{\gamma} \right| \ge \left| \gamma \right| - 2^{-t_\gamma - 1} \ge 2^{-t_\gamma - 1}$ and
  \begin{align*}
     \left| x - \hat{x} \right| = \left| \frac{q}{\gamma} - \frac{q}{\hat{\gamma}} \right| & = \left|q  \frac{\hat{\gamma}-\gamma}{\gamma \hat{\gamma}} \right|  
		\le  \left| q \right| \left| \hat{\gamma}-\gamma \right| 2^{2t_{\gamma}+1} \le 2^{t_q + 2t_{\gamma}+1 - f(n, G)} = 2^{-g(n, G)}.
  \end{align*}
  In light of the properties of $r$ (Lemma~\ref{lem:pr:stretching}) and the fact that $l_j \le k$, it follows that
  \begin{equation*}
    \left| y_j - \hat{y}_j \right| = \left| \frac{q}{x^{l_j}-1} - \frac{q}{\hat{x}^{l_j}-1} \right| \le \left| q \right| 2^{-n - (t_x+1)(m+1) - t_q} \le 2^{-n - (t_x+1)(m+1)}
  \end{equation*}
  for every $j \in \{1, \ldots, m\}$. Thus, we have $\left| \hat{y}_j \right| \le \left|y_j\right| + 1 \le 2^{t_x+1}$. We obtain
  \begin{align*}
    \left| \prod\nolimits_{j = 1}^m y_j -  \prod\nolimits_{j = 1}^m \hat{y}_j \right| 
                                                                                      & = \left| \sum\nolimits_{j = 1}^{m}  \left(y_j - \hat{y}_j \right) \prod\nolimits_{s = 1}^{j-1} \hat{y}_s \prod\nolimits_{s = j+1}^{m} y_s\right| 
																																											< \sum_{j = 1}^{m} \left| y_j - \hat{y}_j \right|  2^{t_x(m-1)+j-1} \\
                                                                                      & \le 2^{t_x(m-1)}\sum_{j = 1}^m 2^{-n - (t_x + 1)(m+1) + j-1}  \le  2^{-n-m-2} \sum_{j = 1}^{m} 2^j  < 2^{-n}.
  \end{align*}
  Equivalently, $\left| w(G, s, t; q, \gamma) - w(G, s, t; q, \hat{\gamma}) \right| < 2^{-n}$ as we wanted to prove.
\end{proof}

We now prove the main transitivity property of polynomial-time approximate shifts that we will use in our constructions.
\begin{lemma} \label{lem:approx-shifts:transitivity} 
  \statelemtransitivity
\end{lemma}
\begin{proof}
  Let $\gamma_j = y_j - 1$ for every $j \in \{1, 2, 3\}$. Let $n$ be a positive integer. We give an algorithm that constructs a graph $J_n$, in polynomial time in $n$, such that $J_n$ $\gamma_1$-implements $\hat{\gamma}_3$ with $\left| \gamma_3 - \hat{\gamma}_3\right| \le 2^{-n}$. This algorithm is as follows. First, we use the approximate theta shift from $(x_2, y_2)$ to $(x_3, y_3)$ to compute a theta graph $G_2$ with terminals $s_2$ and $t_2$ such that
  \begin{equation}\label{eq:t4t3tf4}
    \left| \gamma_3 - w(G_2, s_2, t_2; q, \gamma_2) \right| \le 2^{-n-1}.
  \end{equation}
  The size of $G_2$ is at most polynomial in $n$. In light of Lemma~\ref{lem:pc:essential}, we can compute, in polynomial time in $n$, a positive integer $f(n+1, G_2)$ such that for any $\hat{\gamma}_2$ with $\left| \gamma_2 - \hat{\gamma}_2 \right| \le 2^{-f(n+1, G_2)}$, we have 
  \begin{equation}\label{eq:t4t3tf4b}
    \left|w(G_2, s_2, t_2; q, \gamma_2) - w(G_2, s_2, t_2; q, \hat{\gamma}_2) \right| \le 2^{-n-1}.
  \end{equation}
  We also have $f(n+1, G_2) = n + \Theta \left( \mathrm{size}(G_2) \right)$, so $f(n+1, G_2)$ is bounded by a polynomial in $n$. Now we use the approximate shift from $(x_1, y_1)$ to $(x_2, y_2)$ to compute, in polynomial time in $n$, a graph $G_1$ such that $G_1$ $\gamma_1$-implements $\hat{\gamma}_2$ with $\left| \gamma_2 - \hat{\gamma}_2 \right| \le 2^{-f(n+1,G_2)}$. Combining \eqref{eq:t4t3tf4} and \eqref{eq:t4t3tf4b} with the triangle inequality, we obtain 
$\big| \gamma_3 - w(G_2, s_2, t_2; q, \hat{\gamma}_2) \big|\le 2^{-n}$.

  Finally, we construct a graph $J_n$ as a copy of $G_2$ where every edge is substituted by a copy of $G_1$ as in Lemma~\ref{lem:implementations:multivariate}. Since the sizes of $G_1$ and $G_2$ are polynomial in $n$, the size of $J_n$ also is polynomial in $n$. Recall that $G_2$ $\hat{\gamma}_2$-implements $\hat{\gamma}_3 =  w(G_2, s_2, t_2; q, \hat{\gamma}_2)$ and $G_1$ $\gamma_1$-implements $\hat{\gamma}_2$. Therefore, the graph $J_n$ $\gamma_1$-implements $\hat{\gamma}_3$, and $ \left| \gamma_3 - \hat{\gamma}_3 \right| \le 2^{-n}$, as we wanted to obtain. Finally, if the polynomial-time approximate shift from $(x_1, y_1)$ to $(x_2, y_2)$ is series-parallel, then the graphs $J_n$ are easily seen to be series-parallel, and the result follows.
\end{proof}

\section{Polynomial-time approximate shifts with complex weights} \label{sec:complex-implementations}

In this section we show how to implement approximations of real weights when the original weight is a non-real algebraic number. As a consequence of our results, for any real algebraic number $q$ with $q \ge 2$ and any pair of algebraic numbers $(x,y) \in \mathcal{H}_q$ with $y \not \in \mathbb{R}$ and $(x,y) \not \in \{(-i,i),(i-i), (\omega_3^2, \omega_3), (\omega_3, \omega_3^2)\}$, where $\omega_3 = \exp (2 \pi i / 3)$, there is a polynomial-time approximate shift from $(x,y)$ to any pair of real algebraic numbers $(x', y') \in \mathcal{H}_q$ (see Theorem~\ref{thm:approx-shifts}). Our approach to prove Theorem~\ref{thm:approx-shifts} is as follows. First, we show that there is $(x', y') \in \mathcal{H}_q$ with $y' \in (0,1)$ such that there is a polynomial-time approximate theta shift from $(x,y)$ to $(x',y')$ (see Lemma~\ref{lem:implement-approx:real-with-complex}). Since $x$ and $y$ are algebraic, it follows that $x'$ and $y'$ are polynomial-time computable. Secondly, we extend part of Lemma~\ref{lem:implement-approx:literature:1} to the case where the numbers involved are only known to be polynomial-time computable (see Lemma~\ref{lem:implement-approx:pc}). Finally, we use the transitivity property given in Lemma~\ref{lem:approx-shifts:transitivity} to combine both results in the proof of Theorem~\ref{thm:approx-shifts}.

\subsection{Computing with algebraic numbers} \label{sec:complex-implementations:algebraic}

In our proofs we use and develop some algorithms on algebraic numbers. We gather these algorithms in this section. We represent an algebraic number $z$ as its minimal polynomial $p$ and a rectangle $R$ of the complex plane such that $z$ is the only root of $p$ in $R$. We can compute the addition, subtraction, multiplication, division and conjugation of algebraic numbers in polynomial time in the length of their representations, see~\cite{Strzebonski1997} for details. As a consequence, we can also compute the real and imaginary parts of $z$ and the norm of $z$, which are algebraic numbers themselves, in polynomial time. Note that an algebraic number is $0$ if and only if its minimal polynomial is $x$, which can be easily checked in this representation. Hence, we can also determine in polynomial time whether two algebraic numbers $z_1$ and $z_2$ are equal by checking if $z_1 - z_2$ is $0$. 

When $z$ is a real algebraic number, we can simply represent it as its minimal polynomial $p$ and an interval $I$ with rational endpoints such that $z$ is the only root of $p$ in $I$. If we are given a real algebraic number $z$ with this representation, then we can approximate it as closely as we want by applying Sturm sequences and binary search~\cite{Emiris2008}. In fact, for $z_1$ and $z_2$ real algebraic numbers, Sturm sequences also allow us to check whether $z_1 \ge z_2$ in time polynomial  in the length of the representations of $z_1$ and $z_2$. See~\cite{Emiris2008} for more details and complexity analysis.

A \emph{root of unity} is a complex number $z$ such that $z^k = 1$ for some positive integer $k$. The smallest positive integer $n$ such that $z^n = 1$ is the \emph{order of $z$}. Note that roots of unity are algebraic numbers.  The roots of unity of order $n$ share the same minimal polynomial,   known as the $n$-th cyclotomic polynomial, whose degree is $\varphi(n)$, the Euler phi function. We can determine whether an algebraic number $z$ is a root of unity by checking whether its minimal polynomial is cyclotomic, see~\cite{Bradford1989} for a polynomial-time algorithm. If $z$ is a root of unity, then we can easily compute its order from its representation; we compute the smallest $n$ such that the minimal polynomial of $z$ divides $z^n - 1$. This computation runs in polynomial time in the length of the representation of $z$ as a consequence of the elementary bound $\varphi(n) \ge \sqrt{n / 2}$.

Another operation that we can perform in polynomial time is checking if the argument of an algebraic number is in a fixed interval.

\begin{lemma} \label{lem:algebraic:argument} Let $a, b \in \mathbb{Q} \cap [0,1]$ with $a \le b$. Then for any algebraic number $z$ we can check whether $\mathrm{Arg}(z) \in [2 \pi a, 2\pi b]$ in time polynomial in the length of the representation of $z$.
\end{lemma}
\begin{proof}
  We can split the interval $[2 \pi a, 2\pi b]$ into intervals of length at most $\pi / 2$ and check if $\mathrm{Arg}(z)$ belongs to any of those intervals. Hence, let us assume for the sake of simplicity  that $[2 \pi a, 2\pi b] \subseteq [0, \pi /2]$. The other cases are analogous. Note that $e^{2 \pi i a}$ and $e^{2 \pi i b}$ are roots of unity and, in particular, algebraic. Thus, we can compute $z_a = z e^{- 2 \pi i a}$ and $z_b = z e^{2 \pi i (1/4-b)}$. We have $\mathrm{Arg}(z_a) \in [0,\pi/2]$ if and only if $\mathrm{Arg}(z) \in [2 \pi a, \pi / 2 + 2 \pi a]$, and $\mathrm{Arg}(z_b) \in [0,\pi/2]$ if and only if $\mathrm{Arg}(z) \in [-\pi/2 +2\pi b, 2 \pi b]$.  We conclude that $\mathrm{Arg}(z) \in [2 \pi a, 2\pi b]$ if and only if $\mathrm{Arg}(z_a) \in [0,\pi/2]$ and $\mathrm{Arg}(z_b) \in [0,\pi/2]$.  Finally, note that, for any algebraic number $y$, since $\mathrm{Re}(y)$ and $\mathrm{Im}(y)$ are algebraic, we can determine if $\mathrm{Arg}(y) \in [0, \pi/2]$ or not by checking the inequalities $\mathrm{Re}(y) \ge 0$ and $\mathrm{Im}(y) \ge 0$.
\end{proof}

In the rest of this section we show how to efficiently compute a sequence $\sigma(n)$ such that $\mathrm{Arg} ( z^{\sigma(n)} ) \in [2\pi a, 2\pi b]$ for every $n$. We will use the following well-known result, see, e.g.,~\cite[Section 1.2]{brin_stuck_2002}: if $z \in \mathbb{C}$ is not a root of unity and $\left| z \right| = 1$, then $\{z^j : j \in \mathbb{N}\}$ is dense in the unit circle. 

\begin{lemma} \label{lem:sigma-n:1}
  Let $a,b \in \mathbb{Q} \cap [0,1]$ with $a < b$. Let $z$ be an algebraic number such that $|z| = 1$ and $z$ is not a root of unity. Then there exists a sequence of positive integers $\{\sigma(n)\}$ and a positive integer $k$ such that such that:
  \begin{enumerate}
  \item $k$ can be computed from $z$;
  \item $\sigma(n)$ can be computed in polynomial time in $n$;
  \item $n \le \sigma(n) \le n+k-1$ for every positive integer $n$;
  \item $\mathrm{Arg} ( z^{\sigma(n)} ) \in [2a\pi, 2b\pi] + 2 \pi \mathbb{Z}$ for every positive integer $n$.
  \end{enumerate}
\end{lemma}
\begin{proof}
  Our algorithm to compute $\sigma(n)$ is as follows. Set $\sigma(0) = 0$. We compute $\sigma(n)$ as the smallest integer such that $n \le \sigma(n)$ and $\mathrm{Arg} ( z^{\sigma(n)} ) \in [2 a \pi, 2 b \pi]$. We can check whether $\mathrm{Arg} ( z^{\sigma(n)} ) \in [2 a \pi, 2 b \pi]$ or not by applying the procedure given in Lemma~\ref{lem:algebraic:argument}.
  
  We show that $\sigma(n)$ is well-defined. Let $\theta = \mathrm{Arg} (z)$. Since $z$ is not a root of unity, $\{z^j : j \in \mathbb{N}\}$ is dense in the unit circle, as we have discussed in the previous paragraph. Therefore, there is $q \in \mathbb{N}$ such that $\mathrm{Arg} ( z^q ) \in [0, 2(b-a)\pi]$.
  Note that we can compute $q$ in constant time with the help of Lemma~\ref{lem:algebraic:argument}. Let $\tau = \mathrm{Arg}(z^q)$. Since $z$ is not a root of unity, we find that $\tau \ne 0$. Let $t = \lceil 2 \pi / \tau \rceil$. Since $t$ is the smallest positive integer such that $t \tau \ge 2 \pi$, $t$ can be computed by sequentially determining which of the following intervals contains the argument of $z^{qj}$: $(0, \pi /2)$, $(\pi/2, \pi)$, $(\pi, 3\pi/2)$ or $(3\pi/2, 2 \pi)$. Hence, we can compute $k = t q$. For each positive integer $n$, since $t \tau \ge 2\pi$ and $\tau < 2(b-a) \pi$, there is $p_n \in \{0, \ldots, t-1\}$ such that $n \theta + p_n \tau \in [2a\pi, 2b\pi] + 2\pi \mathbb{Z}$. The integer $m_n = n + p_n q$ satisfies $n \le m_n \le n + k -1$ and
  \begin{equation*}
    m_n \theta \in n \theta + p_n \tau + 2 \pi \mathbb{Z} \subseteq [2a\pi, 2b\pi] + 2 \pi \mathbb{Z}.
  \end{equation*}
  We conclude that $\sigma(n)$ is well-defined and $n \le \sigma(n) \le m_n \le n + k -1$, so our algorithm computes $\sigma(n)$ in polynomial time in $n$.
\end{proof}

\begin{lemma} \label{lem:sigma-n:2} Let $z$ be a root of unity of order $k$ with $k \not \in \{1, 2, 4\}$. Then there exists a sequence of positive integers $\left\{\sigma(n)\right\}$ and an integer $l$ such that:
  \begin{enumerate}
  \item $\sigma(n)$ can be computed in polynomial time in $n$;
  \item $n \le \sigma(n) \le n+k-1$ for every positive integer $n$;
  \item $z^{\sigma(n)} = e^{2 \pi i l / k}$ for every positive integer $n$;
  \item $\pi < 2 \pi l / k < 3 \pi / 2$.
  \end{enumerate}
\end{lemma}
 
\begin{proof}
  Let $\theta = \mathrm{Arg}(z)$. Since $\theta \ne 0$, we can write $\theta = 2 \pi j / k$ for some integer $j$ coprime with $k$. We consider two cases.

\textbf{Case I:} $\mathbf{k = 3}$. Then either we have $\theta = 2 \pi / 3$ and we compute $\sigma(n) \in \{n, n+1, n+2\}$ with $\sigma(n) \equiv 2 \pmod{ 3 }$, or we have $\theta = 4 \pi / 3$ and we compute $\sigma(n) \in \{n, n+1, n+2\}$ with $\sigma(n) \equiv 1 \pmod{ 3 }$. In any case, we have $\sigma(n)\theta \in 4 \pi / 3 + 2 \pi \mathbb{Z}$, that is, $z^{\sigma(n)} = e^{4 \pi i / 3}$ for any positive integer $n$.

\textbf{Case II:} $\mathbf{k \geq  5}$. Then there is an integer $l$ such that $k/2 < l < 3k/4$, that is, $2 \pi l / k \in (\pi, 3\pi/2)$. The Euclidean algorithm gives two integers $t_1, t_2$ such that $t_1 j + t_2 k = 1$. We compute $\sigma(n) \in \{n, \ldots, n+k-1\}$ such that $\sigma(n) \equiv t_1 l \pmod{ k }$. We can write $\sigma(n) = t_1l + q_n k$ for some integer $q_n$. We have
    \begin{equation*}
        \sigma(n) \theta  = t_1 l \frac{2 \pi j}{k} + q_n 2 \pi j = l (1 - t_2 k) \frac{2 \pi}{k} +
        q_n 2 \pi j = \frac{2 \pi l }{k} +(q_n j - l t_2) 2 \pi
    \end{equation*}
 and, equivalently, $z^{\sigma(n)} = e^{2 \pi i l / k}$ for every positive integer $n$. 
\end{proof}

\begin{corollary} \label{cor:sigma-n}
  Let $z$ be an algebraic number such that $z \not \in \mathbb{R} \cup i \mathbb{R}$. Let $\theta = \mathrm{Arg} (z)$. Then there exists a sequence of positive integers $\{\sigma(n)\}$, a positive integer $k$ and a positive rational number $C$ such that such that:
  \begin{enumerate}
  \item $k$ and $C$ can be computed from $z$;
  \item $\sigma(n)$ can be computed in polynomial time in $n$;
  \item $n \le \sigma(n) \le n+k-1$ for every positive integer $n$;
  \item $\cos(\sigma(n) \theta) \le - C$ and $\sin(\sigma(n) \theta) \le - C$ for every positive integer $n$.
  \end{enumerate}
\end{corollary}
\begin{proof}
  We may assume that $|z| = 1$ since, otherwise, we can compute the algebraic number $z / |z|$ and apply the following algorithm to this quantity. We invoke either Lemma~\ref{lem:sigma-n:1} for  $a=7/12$ and $b = 8/12$ or Lemma~\ref{lem:sigma-n:2}, depending on whether $z$ is a root of unity or not, which can be checked as explained at the beginning of this section. In any case, we find a sequence $\sigma$ and a positive integer $k$ that satisfy the first three assertions announced in the statement. In the non-root of unity case, we have $\cos (\sigma(n) \theta) \le \cos (2 \pi b) < 0$ and $\sin (\sigma(n) \theta) \le \sin (2 \pi a) < 0$ for every positive integer $n$. In the root of unity case, the sequences $\cos (\sigma(n) \theta)$ and $\sin (\sigma(n) \theta)$ are negative constants. In any case, we can compute a positive rational number $C$ such that  $\cos (\sigma(n) \theta) \le -C$ and $\sin (\sigma(n) \theta) \le - C$ for every positive integer $n$.
\end{proof}

\begin{corollary} \label{cor:algebraic:large-real-part}
  Let $z$ be an algebraic number with $\left| z \right| > 1$. Then for any $x \in \mathbb{Q}$ with $x > 0$, we can compute $n$ such that $\mathrm{Re}\left(z^n\right) \ge x$. Moreover, if $z \not \in [0, \infty)$, then we can compute $m$ such that $\mathrm{Re}\left(z^m\right) \le - x$.
\end{corollary}
\begin{proof}
  Let $z = R e^{i \theta}$ for some $\theta \in [0, 2\pi)$ and $R > 1$. We determine if $z / |z| = e^{i \theta}$ is a root of unity or not, and compute its order as explained before. If $e^{i \theta}$ is a root of unity of order $k$, then $z^k \in (1, \infty)$, so computing $n$ is straightforward. If $e^{i\theta}$ is not a root of unity, then, in view of Lemma~\ref{lem:sigma-n:2} for $a = 1/12$ and $b = 1/ 6$, we can compute a sequence $\sigma$ such that $\sigma(j) \ge j$ and $\sigma(j) \theta \in [\pi/6, \pi/3] + 2 \pi \mathbb{Z}$ for every positive integer $j$. We find that $\mathrm{Re} ( z^{\sigma(j)} ) \ge R^{\sigma(j)}\cos ( \pi / 3 ) \ge R^j / 2$. Hence, we can compute $j$ large enough such that $\mathrm{Re}( z^{\sigma(j)} ) \ge x$ and we choose $n = \sigma(j)$.

  Now let us assume that $z \not \in [0, \infty)$. Note that $e^{i \theta} \ne 1$. If $e^{i \theta}$ is a root of unity of order $2$ or $4$, then the result is trivial. If $\theta \not \in \{0, \pi/2, \pi, 3\pi/2\}$, then, by invoking Corollary~\ref{cor:sigma-n}, we compute $\sigma$ and a positive rational number $C$ such that  $\sigma(j) \ge j$ and $\cos (\sigma(j) \theta) \le -C$ for every positive integer $j$. We find that $\mathrm{Re} ( z^{\sigma(j)} ) \le -CR^{\sigma(j)} \le -C R^{j}$. Hence, we can compute $j$ large enough such that $\mathrm{Re} ( z^{\sigma(j)} ) \le -x$ and we choose $m = \sigma(j)$.
\end{proof}

\subsection{Some shifts for non-real algebraic numbers}

In this section we gather some of the shifts that we use in our proofs. Let $q$ be a real algebraic number with $q \ge 2$ and let $(x,y) \in \mathcal{H}_q$ be a pair of algebraic numbers. We are interested in computing a shift from $(x,y)$ to $(x_1, y_1) \in \mathcal{H}_q$ with $x_1 \not \in \mathbb{R}$ and $\left| x_1 \right| > 1$ whenever possible. The existence of this shift turns out to be closely related to the hardness of approximating $\left| Z_{\text{Tutte}}(G; q, \gamma) \right|$ with $\gamma = y-1$; when we can compute such a shift the approximation problem is $\# \mathsf{P}$-hard, as we will see in Section~\ref{sec:hardness}. Recall that one can evaluate the Tutte polynomial of a graph in polynomial time at any of the points in $\{(-i, i), (i, -i), (\omega_3^2, \omega_3), (\omega_3, \omega_3^2)\}$, where $\omega_3 = \exp (2 \pi i / 3)$ (see Section~\ref{sec:pre:hardness}). These are the points for which our results fail to construct the desired shift.

The results of this section involve computations that might not run in polynomial time in the length of the representation of the algebraic numbers $q$, $x$ and $y$ involved. However, when applying these results, the numbers $q$, $x$ and $y$ are constants and, hence, this will not affect the complexity of our algorithms.

\begin{remark} \label{rem:norm-x}
  Let $q$ be a positive real number and let $(x,y) \in \mathcal{H}_q$. From $(x-1)(y-1) = q$ it follows that $x$ is real if and only if $y$ is real. Note that $x = 1 + q/(y-1)= (y+q-1)/(y-1)$. As noted in the proof of Lemma~\ref{lem:pc:essential}, we find that $\left| x \right| = 1$ if and only if $\left| y+q-1 \right| = \left| y-1 \right|$, that is, $y$ is on the line $1-q/2+i \mathbb{R}$. Moreover, $\left| x \right| > 1$ if and only if $\mathrm{Re}( y ) > 1-q/2$. Note that when $q \ge 2$ and $\mathrm{Re}( y ) > 0$, we have $\mathrm{Re}( y ) > 1-q/2$ and, thus, $\left| x \right| > 1$. These observations will be applied several times in this section.
\end{remark}

\begin{lemma} \label{lem:norm-shift}
  Let $q$ be a real algebraic number with $q \ge 2$. Let $x$ and $y$ be algebraic numbers such that $(x, y) \in \mathcal{H}_q$ and $\mathrm{Arg}(y) \not \in \{0, \pi/2, 2\pi/3, \pi, 4\pi/3, 3\pi/2\}$. Then we can compute a theta graph $J$ that $(x,y)$-implements $(x_1, y_1)$ with $\left| x_1 \right| > 1$ and $x_1 \not \in \mathbb{R}$.
\end{lemma}
\begin{proof}
  We show how to compute $n$ such that $\mathrm{Re}(y^n) > 0$ and $\mathrm{Im}(y^n) > 0$. For such a $n$, we let $y_1 = y^n$ and $x_1 = 1 + q / (y_1 - 1)$, so Remark~\ref{rem:norm-x} ensures that $\left| x_1 \right| > 1$ and $x_1 \not \in \mathbb{R}$. Hence, we can return $J$ as the graph with two vertices and $n$ edges joining them. Since $y$ and $\left| y \right|$ are algebraic numbers, we can compute the algebraic number $y / \left| y \right|$. We can detect if  $y / \left| y \right|$ is a root of unity or not as explained in Section~\ref{sec:complex-implementations:algebraic}. There are two cases:
  \begin{enumerate}[label=(\roman*)]
    \item \label{item:norm-shifts:1} $y / \left| y \right|$ is not a root of unity. Then we can apply Lemma~\ref{lem:sigma-n:1} with $a = 1/12, b = 1/6$ and $z=y^n$ to compute the smallest positive integer $n$ such that $\mathrm{Arg}(y^n) \in [\pi/6, \pi/3]$. Recall that such an integer exists because $\{(y / \left| y \right|)^j : j \in \mathbb{N}\}$ is dense in the unit circle. Finally, since $\mathrm{Arg}(y^n) \in [\pi/6, \pi/3]$, we have $\mathrm{Re}(y^n) > 0$ and $\mathrm{Im}(y^n) > 0$.
    \item \label{item:norm-shifts:2} $y / \left| y \right|$ is a root of unity of order $r$ with $r \ge 5$. Recall that we can compute $r$ by sequentially computing the powers of $y / \left| y \right|$ until we obtain $1$. Then we have $\left(y / \left| y \right|\right)^{r+1} = e^{i2 \pi / r}$. Note that the real and imaginary parts of $e^{i2 \pi / r} = \cos(2\pi/r) + i \sin(2 \pi / r)$ are positive. \qedhere
  \end{enumerate}
\end{proof}

Note that the argument given in Lemma~\ref{lem:norm-shift} strongly uses the fact that $q \ge 2$, that is, $1-q/2 \le 0$. A proof of a version of Lemma~\ref{lem:norm-shift} with $q \in (0, 2)$ is unknown to us. Now we deal with the cases $\mathrm{Arg}(y) \in \{ \pi/2, 2\pi/3,  4\pi/3, 3\pi/2\}$, where the exemptions $(-i, i), (i, -i), (\omega_3^2, \omega_3), (\omega_3, \omega_3^2)$ arise. Note that $(-i, i), (i, -i) \in \mathcal{H}_2$ and $(\omega_3^2, \omega_3), (\omega_3, \omega_3^2) \in \mathcal{H}_3$. In fact, one can easily check that these are the only pairs $(x,y)$ such that $\left| y \right| = 1$ and $q \in \{2, 3\}$.

\begin{lemma} \label{lem:shift:case-3}
  Let $q$ be a real algebraic number with $q \ge 2$. Let $x$ and $y$ be algebraic numbers such that $(x, y) \in \mathcal{H}_q$, $y \ne 0$ and $\mathrm{Arg} (y) \in \{2 \pi / 3, 4 \pi / 3\}$. If $q \ne 3$ or $\left| y \right| \ne 1$, then we can compute a series-parallel graph $J$ that $(x,y)$-implements $(x_1, y_1)$ with $\left| x_1 \right| > 1$ and $x_1 \not \in \mathbb{R}$.
\end{lemma}
\begin{proof}
  Note that $y / \left| y \right|$ is a root of unity of order $3$. We have $\mathrm{Re}(y) = \left| y \right| \cos(2 \pi / 3) =  - \left| y \right| / 2 < 0$. Let $x = 1 + q / (y-1)$. We consider three cases.
	
\textbf{Case I:} $\mathrm{Re}(y) > 1 - q / 2$. Then, by Remark~\ref{rem:norm-x}, $\left| x \right| > 1$. We return $J$ as the graph with $2$ vertices and one edge joining them.

\textbf{Case II:} $\mathrm{Re}(y) < 1 - q / 2$. Then $\left| x \right| < 1$. Let $y_n = 1 + q / (x^n -1)$. An $n$-stretch gives a shift from $(x,y)$ to $(x^n, y_n)$. Since $x \not \in \mathbb{R}$, there are infinitely many values of $n$ such that $y_n \not \in \mathbb{R}$. Note that $y_n$ converges to $1 - q \in (-\infty, -1]$, and the distance between $1-q$ and the set of complex points $\{z \in \mathbb{C} : \mathrm{Arg}(z) \in \{\pi/2, 2\pi/3, 4\pi/3, 3\pi/2\} \}$ is larger than $0$. Hence, we can compute $n$ such that $\mathrm{Arg}(y_n) \not \in \{0, \pi/2, 2\pi/3, \pi, 4\pi/3, 3\pi/2\}$. Since $(x^n, y_n) \in \mathcal{H}_q$, the result follows from applying Lemma~\ref{lem:norm-shift} to $(x^n, y_n)$, the transitivity property of shifts and noticing that the obtained graph is series-parallel. 

\textbf{Case III:} $\mathrm{Re}(y) = 1 - q / 2$. Note that $q > 2$ because for $q = 2$ we would obtain $\mathrm{Re}(y) = 0$. We distinguish three subcases:
      \begin{itemize}
        \item $\left| y \right| > 1$. We compute the smallest positive integer $n$ such that $\mathrm{Arg}(y^n) = 2\pi / 3$ and $\mathrm{Re}(y^n) = - \left| y \right|^n / 2 < 1 - q / 2$. The proof is concluded by applying Case II to $(x_n, y^n)$, where $x_n = 1 + q / (y^n - 1)$, the transitivity property of shifts and noticing that the obtained graph is series-parallel.
        \item $\left| y \right| < 1$. We compute the smallest positive integer $n$ such that $\left| y \right|^n < q-2$ and $\mathrm{Arg}(y^n) = 2\pi / 3$. We have $\mathrm{Re}(y^n) > 1 - q / 2$ (otherwise by applying $\mathrm{Re}(y^n) = - \left| y \right|^n / 2$ we would find that $\left| y \right|^n \ge q - 2$), so $\left| x_n \right| > 1$ for $x_n = 1 + q / (x^n - 1)$. We return $J$ as the graph with two vertices and $n$ edges joining them.
        \item $\left| y \right| = 1$. Then $1 - q / 2 = \mathrm{Re}(y) = - \left| y \right| / 2 = - 1 /2$. It follows that $q = 3$, but this case ($\left| y \right| = 1$ and $q \ne 3$) was excluded in the hypothesis. 
      \end{itemize}
This finishes the proof.
\end{proof}

\begin{lemma} \label{lem:shift:case-4}
  Let $q$ be a real algebraic number with $q \ge 2$. Let $y$ be an algebraic number such that $y \ne 0$ and $\mathrm{Arg} (y) \in \{\pi / 2, 3 \pi / 2\}$.  
  \begin{enumerate}
  \item If $ q > 2$, then we can compute a theta graph $J$ that $(x,y)$-implements $(x_1, y_1)$ with $\left| x_1 \right| > 1$ and $x_1 \not \in \mathbb{R}$.
  \item \label{item:shift:case-4:2} If $q = 2$ and $\left| y \right| \ne 1$, then we can compute a series-parallel graph $J$ that $(x,y)$-implements $(x_2, y_2)$ with $y_2 \in (-1, 0)$.
  \end{enumerate}  
\end{lemma}
\begin{proof}
  The hypotheses $y \ne 0$ and $\mathrm{Arg} (y) \in \{\pi / 2, 3 \pi / 2\}$ are equivalent to $y \ne 0$ and $\mathrm{Re}(y) = 0$. Let $x = 1 + q / (y-1)$. If $q > 2$, then $1 - q/2 < 0 = \mathrm{Re}(y)$ and $\left| x \right| > 1$ as a consequence of Remark~\ref{rem:norm-x}, so we return the graph with two vertices and one edge joining them as $J$. The second claim (case $q = 2$) has been studied in~\cite[Lemma 3.15]{Goldberg2017}, where the graph constructed is a $2$-thickening of a $k$-stretching.
\end{proof}

\begin{corollary} \label{cor:norm-shift}
  Let $q$ be a real algebraic number with $q \ge 2$. Let $x$ and $y$ be algebraic numbers such that $(x, y) \in \mathcal{H}_q$, $y \not \in (-\infty, -1] \cup [0,\infty)$ and $(x,y) \not \in \{(i, -i), (-i, i), (\omega_3, \omega_3^2), (\omega_3^2, \omega_3)\}$, where $\omega_3 = \exp(2 \pi i / 3)$. Then we can compute $(x_2, y_2) \in \mathcal{H}_q$ and a series-parallel graph $J$ such that $\left| x_2 \right| < 1$ and $J$ $(x,y)$-implements $(x_2, y_2)$.
\end{corollary}
\begin{proof}
  First, we assume that $y \not \in \mathbb{R}$. The case $q = 2$ and $y \in \mathbb{R} i$ is covered in Lemma~\ref{lem:shift:case-4}, so we assume that $q \ne 2$ or $y \not \in \mathbb{R}i$. By applying Lemmas~\ref{lem:norm-shift},~\ref{lem:shift:case-3} or~\ref{lem:shift:case-4} (depending on the argument of $y$), we can compute a graph $J$ that $(x,y)$-implements $(x_1, y_1)$ with $\left| x_1 \right| > 1$ and $x_1 \not \in \mathbb{R}$. We apply Corollary~\ref{cor:algebraic:large-real-part} with $z = x_1$ in order to compute $n$ such that $\mathrm{Re}(x_1^n) > 1$. A $n$-stretching of $(x_1, y_1)$ gives a shift from $(x_1, y_1)$ to $(\hat{x}, \hat{y})$, where $\hat{x} = x_1^n$ and $\hat{y} = 1 + q / (\hat{x} - 1)$. We have $\mathrm{Re} (\hat{y}) = 1 + q (\mathrm{Re}(\hat{x}) - 1) / \left| \hat{x} - 1 \right|^2 > 1$, so $\left| \hat{y} \right| > 1$. There are two cases:
  \begin{itemize}
    \item $\hat{y} \not \in \mathbb{R}$. We apply Corollary~\ref{cor:algebraic:large-real-part} with  $z = \hat{y}$ to compute $t$ such that $\mathrm{Re}(\hat{y}^t) < 1 - q/2 < 0$. We set $y_2 = \hat{y}^t$ and $x_2 = 1 + q / (y_2 - 1)$. By the transitivity property of shifts, we have a shift from $(x,y)$ to $(x_2, y_2)$.  Since $\mathrm{Re}(y_2) < 1 - q/2$, we conclude that $\left| x_2 \right| < 1$~(Remark~\ref{rem:norm-x}). 
    \item $\hat{y} \in \mathbb{R}$. Hence, we have $\hat{y} \in (1,\infty)$. We can compute a positive integer $l$ such that the norm of $y' = \hat{y}^l y$ is larger than $1$. Note that $y' = \hat{y}^l y \not \in \mathbb{R}$. A parallel composition yields  a shift from $(x,y)$ to $(x', y')$, where $x' = 1 + q / (y' - 1)$. We compute the graph $J$ by applying the previous case to $(x', y')$.
    \end{itemize}
  Now we deal with the case $y \in (-1,0)$. A $2$-thickening gives us a shift from $(x, y)$ to $(a_1, b_1)$, where $b_1 = y^2 \in (0,1)$ and $a_1 = 1 + q / (b_1 - 1) < 1 - q \le -1$. A $2$-stretching gives us a shift from $(a_1, b_1)$ to $(a_2, b_2)$, where $a_2 = a_1^2 > 1$ and $b_2 = 1 + q / (a_2 - 1) > 1$. We compute a positive integer $j$ such that $b_2^j y < -q$ and, with the help of a $j$-thickening, construct a shift from $(a_2, b_2)$ to $(a_3, b_3)$ with $b_3 = b_2^j$. The transitivity property of shifts allows us to construct a shift from $(x,y)$ to $(a_3, b_3)$. To conclude the proof, we apply a parallel composition between the latter shift and the identity shift from $(x,y)$ to $(x,y)$, obtaining a shift from $(x,y)$ to $(x_2, y_2)$ with $y_2 = b_3 y$. Recall that $b_3 y = b_2^j y < -q$, so $q / (y_2 - 1) \in (-1, 0)$ and $x_2 = 1 + q / (y_2 - 1) \in (0, 1)$.

  Finally, note that the graphs considered in this proof are series-parallel.
\end{proof}

\subsection{An approximate shift to $(0, 1-q)$}

In Lemma~\ref{lem:implement-weight:approx-1:pc} and Corollary~\ref{cor:implement-weight:approx-1} we give a polynomial-time approximate series-parallel shift from $(x,y)$ to $(0,1-q)$ under certain conditions.

\begin{lemma} \label{lem:implement-weight:approx-1:pc} 
  Let $q \in \mathsf{P}_{\mathbb{R}}$ with $q > 0$.  Let  
  $(x,y) \in \mathcal{H}_q$ such that $x, y \in \mathsf{P}_{\mathbb{C}}$  and $\mathrm{Re}(y) < 1-q/2$. Then there is a polynomial-time approximate theta shift from $(x, y)$ to $(0, 1-q)$.
\end{lemma}
\begin{proof}
  Let $x = 1 + q/(y-1)$. In light of Remark~\ref{rem:norm-x}, we have $\left| x \right| < 1$. Therefore, the weight $y_j = 1 + q / (x^j - 1)$ implemented by an $j$-stretch converges to $1-q$ as $j \to \infty$. We have
  \begin{equation} \label{eq:approx:1-q}
    \left| q-1 + y_j \right| = \left| \frac{q x^j}{x^j - 1} \right| \le \frac{q \left|x\right|^j}{1 - \left|x\right|^j} \le \frac{q \left|x\right|^j}{1 - \left|x\right|} .
  \end{equation}
  We use \eqref{eq:approx:1-q} to give a a polynomial-time approximate theta shift from $(x,y)$ to $(0,1-q)$. Let $n$ be a positive integer, so the desired accuracy of the quantity in \eqref{eq:approx:1-q} is $2^{-n}$. We are going to return a path graph with $j$ edges for $j$ large enough. It remains to show how to compute $j$ from $n$. Since $q, \left| x \right| \in \mathsf{P}_{\mathbb{R}}$, we can compute $b, c \in \mathbb{Q}$ such that $q \le c$ and $0 < b \le 1 - \left| x \right|$ (Lemma~\ref{lem:pr:bounds}). Hence, $\left| x \right| \le 1 - b < 1$, and it suffices to compute $j$ with $j \ge \log_{1-b} ( 2^{-n} b / c )$.
\end{proof}

\begin{corollary} \label{cor:implement-weight:approx-1}
  Let $q$ be a real algebraic number with $q \ge 2$. Let $x$ and $y$ be algebraic numbers such that $(x,y) \in \mathcal{H}_q$, $y \not \in (-\infty, -1] \cup [0, \infty)$ and $(x,y) \not \in \{(-i, i), (i, -i), (\omega_3^2, \omega_3), (\omega_3, \omega_3^2)\}$, where $\omega_3 = \exp(2\pi i/3)$. Then there is a polynomial-time approximate series-parallel shift from $(x, y)$ to $(0, 1-q)$.
\end{corollary}
\begin{proof}
  From Corollary~\ref{cor:norm-shift} we obtain a shift from $(x,y)$ to $(x_2, y_2)$ with $\left| x_2 \right| < 1$ or, equivalently, $\mathrm{Re}(y_2) < 1 - q/2$. The result follows from applying Lemma~\ref{lem:implement-weight:approx-1:pc} to $(x_2, y_2)$ and the transitivity property of shifts.
\end{proof}

\subsection{An approximate shift to $(x',y')$ with $y' \in (0,1)$}\label{sec:here}

In Lemma~\ref{lem:approx-argument} we show that if a sequence $z_n$ of complex numbers has certain properties, then there is $w \in (0, 1) \cap \mathsf{P}_{\mathbb{R}}$ that is the limit of $\prod_{j = 1}^n z_j^{e_j}$ for some non-negative integers $e_1, e_2, \ldots$ that we can compute. Then we apply this result to a subsequence of $\{y_n\}$, where $(x^n, y_n)$ is the pair implemented by an $n$-stretch of $(x,y)$, obtaining a polynomial-time approximate theta shift from $(x,y)$ to some $(x',y')$ with $y' \in (0,1)$ (Lemma~\ref{lem:implement-approx:real-with-complex}). First, we need the following elementary results.

\begin{lemma} \label{lem:sin-x-bound}
  We have $\sin(x) \le x \le \pi \sin(x) / 2$ for every $x \in [0,\pi/2]$.
\end{lemma}
\begin{proof}
  First, we prove that $\sin(x) \le x$ for every $x \in [0, \pi/2]$. Let $f(x) = x - \sin(x)$. We have $f'(x) = 1 - \cos(x) > 0$ for every $x \in [0, \pi/2]$. Hence, $f$ is strictly increasing in $[0, \pi/2]$. Since $f(0) = 0$, we obtain $x - \sin(x) \ge 0$ for every $x \in [0, \pi/2]$.

  Now we prove that $x \le \pi \sin(x) / 2$ for every $x \in [0, \pi/2]$. Let $g(x) = \pi \sin(x) / 2 - x$ for every $x \in [0, \pi / 2]$. We have $g'(x) = \pi \cos(x) / 2 - 1$. Let $y \in [0, \pi/2]$ such that $\cos(y) = 2 / \pi$. Note that $g'(x) > 0$ in $[0, y)$, $g(y) = 0$ and $g'(x) < 0$ in $(y, \pi/2]$. Hence, $g$ only reaches a minimum at $x \in \{0, \pi/2\}$. Since $g(0) = g(\pi/2) = 0$, we conclude that $0 \le \pi \sin(x) / 2 - x$ for every $x \in [0, \pi/2]$.
\end{proof}

\begin{lemma} \label{cor:pc:limit}
  Let $\{z_n\}$ be a sequence of algebraic complex numbers such that:
  \begin{enumerate}
  \item we can compute two rational numbers $C$ and $R$ such that $C > 0$, $R \in (0,1)$ and $\left| z - z_n \right| \le C R^n$ for every positive integer $n$;
  \item we can compute the representation of the algebraic number $z_n$ in polynomial time in $n$.
  \end{enumerate}
  Then $z \in \mathsf{P}_{\mathbb{C}}$, i.e., $z$ is polynomial-time computable.
\end{lemma}
\begin{proof}
Let $n$ be an arbitrary positive integer. For $j = \lceil \log_R \left(  2^{-n-1}/ C \right) \rceil$ we have $\left| z - z_j\right| \le 2^{-n-1}$. Note that $j=\Theta(n)$ and hence $z_j$ is an algebraic number whose representation we can compute in time polynomial in $n$. So, we can also compute $\hat{z}_j \in \mathbb{Q}[i]$ such that $\left| z_j - \hat{z}_j \right| \le 2^{-n-1}$ in time polynomial in $n$. Then, we have that
  \begin{equation*} \label{eq:pc:limit}
    \left| z - \hat{z}_n \right| \le \left| z_n - z \right| +  \left| z_n - \hat{z}_n \right| \le 2^{-n}.
  \end{equation*}
Since $n$ was arbitrary, we have that $z$ is polynomial-time computable.  
\end{proof}

\begin{lemma} \label{lem:approx-argument}
  Let $r, c \in (0,1) \cap \mathbb{Q}$. Let $\{z_n\}$ be a sequence of algebraic complex numbers with:
  \begin{enumerate}
  \item $|z_n| < 1$ for every positive integer $n$;
  \item $z_n = 1 - f(n) + i g(n)$ with $f, g \colon \mathbb{Z}^+ \to (0,1)$;
  \item $c r^n \le f(n) \le r^n / 2$ and $c r^n \le g(n) \le r^n/2$ for every positive integer $n$.
  \end{enumerate}
  Then there is $w \in (0,1)$ and a bounded sequence of positive integers $\{e_n\}$ such that
  \begin{equation*}
    \left| \prod\nolimits_{j = 1}^{n} z_j^{e_j} - w \right| \le \left( \frac{\pi}{2} + \frac{\pi}{ c (1- r)} \right) r^n
  \end{equation*}
  for every positive integer $n$. Moreover, if the representation of the algebraic number $z_{n}$ can be computed in polynomial time in $n$, then $w \in \mathsf{P}_{\mathbb{R}}$ and $e_n$ can be computed in polynomial time in $n$.
\end{lemma}

\begin{proof}
   We can write $z_n = \rho_n e^{i \theta_n}$ for some $\rho_n \in (0,1)$ and $\theta_n \in (0, \pi/2)$. Note that $1 - f(n) < \rho_n$. Let $h(n) = 1 - \rho_n$. We obtain
   \begin{equation} \label{eq:hn} 0 < h(n) < f(n) \le r^n / 2
   \end{equation}
   for every positive integer $n$. We have
   \begin{equation*}
     \sin \left(\theta_n \right) = \frac{\mathrm{Im}\left(z_n\right)}{\rho_n} = \frac{g(n)}{1 - h(n)} .
   \end{equation*} 
   In view of Lemma~\ref{lem:sin-x-bound}, we obtain
   \begin{equation*} \frac{g(n)}{1 - h(n)} \le \theta_n \le \frac{\pi g(n)}{2(1 - h(n))}.
   \end{equation*}
   Since $0 < h(n) \le 1/2$ (see \eqref{eq:hn}), it follows that
   \begin{equation} \label{eq:bound-theta} g(n) \le \theta_n \le \pi g(n).
   \end{equation} 
   As a consequence, we find that, for any integer $n$ with $n \ge 2$,
   \begin{equation} \label{eq:bound-thetas}
     \begin{aligned}
       \frac{\theta_{n-1}}{\theta_n} & \le \pi \frac{g(n-1)}{ g(n)} \le \frac{\pi }{2 c r},
     \end{aligned}
   \end{equation}
   where we used the fact that $c r^n \le g(n) \le r^n/2$. The bounds \eqref{eq:hn}, \eqref{eq:bound-theta} and \eqref{eq:bound-thetas} will be used several times in this proof.
    
   Let $\tau_0 = 0$. We define $\tau_n$ and $e_n$ by induction on $n$. Let $e_n$ be the largest integer such that $\tau_{n-1} + e_n \theta_n \le 2\pi$ and let $\tau_n = \tau_{n-1} + e_n \theta_n$. By definition, $\left\{ \tau_n \right\}$ is an increasing sequence that is bounded above by $2\pi$. Moreover, we have $2\pi - \theta_n < \tau_n$, since $\tau_n + \theta_n \le 2\pi$ contradicts the definition of $e_n$. That is, we have $0 \le 2\pi - \tau_n < \theta_n$. We show that $e_n$ is bounded. Note that $e_1 \le 2\pi / \theta_1 \le 2\pi / (c r)$, where we used that $c r \le g(1) \le \theta_1$ (recall \eqref{eq:bound-theta}). For $n \ge 2$ we have
   \begin{equation*}
     0 \le e_n = \frac{\tau_n - \tau_{n-1}}{\theta_n} \le \frac{2\pi - \tau_{n-1}}{\theta_n} <  \frac{\theta_{n-1}}{\theta_n} \le \frac{\pi}{2 c r},
   \end{equation*}
   where we applied \eqref{eq:bound-thetas}. By combining the latter inequality with the case $n = 1$ we conclude that
   \begin{equation} \label{eq:bound-dn}
     0 \le e_n \le \frac{2\pi}{cr}
   \end{equation}
   for every positive integer $n$.

   The sequence $\{e^{i \tau_n}\}$ converges to $1$. In fact, we show that it does so exponentially fast. Note that the derivative of $e^{it}$ has constant norm $1$. Therefore, $e^{it}$ is a Lipschitz function with constant $1$, that is, $\left| e^{it} - e^{is} \right| \le \left| s-t \right|$ for every $s,t \in \mathbb{R}$. It follows that
   \begin{equation} \label{eq:approx-argument:1} \left| 1 - e^{i \tau_n} \right| = \left| e^{i2 \pi}
       - e^{i \tau_n} \right| \le \left| 2\pi - \tau_n \right| < \theta_n \le \pi g(n) \le
     \frac{\pi}{2} r^n
   \end{equation}
   for every positive integer $n$, where we applied \eqref{eq:bound-theta}.
    
   Now we study the sequence $\{x_n\}$ for $x_n = \prod_{j = 1}^n \rho_j^{e_j}$. Since $\rho_j \in (0,1)$, $\{x_n\}$ is decreasing and has a limit $w \in [0,1)$. We claim that this is the real number in $(0,1)$ announced in the statement. First, we prove that $w > 0$. Let $b = \lceil 2 \pi / (cr)\rceil$. In view of \eqref{eq:bound-dn}, we have
   \begin{equation*}
     x_n \ge \prod_{j = 1}^n \rho_j^b = \Big( \prod_{j = 1}^n \left( 1 - h(j) \right)\Big)^b.
   \end{equation*}
   Recall that a product of the form $\prod_{j = 1}^n \left( 1 - a_n \right)$ with $a_n \in [0,1)$ converges to a positive number if and only if $\sum_{j = 1}^n a_n$ converges~\cite[Proposition 3.1]{stein2010}. From \eqref{eq:hn} we obtain
   \begin{equation*}
     \sum_{n = 1}^{\infty} h(n) \le \frac{1}{2} \sum_{n = 1}^{\infty} r^n = \frac{r}{2(1-r)}
   \end{equation*}
   and, thus, $\prod_{j = 1}^n \left( 1 - h(j) \right)$ converges to a real number $L$ with $L > 0$. We conclude that $w \ge L^b > 0$, as we wanted to prove. Now we show that $\{x_n\}$ converges exponentially fast to $w$. Note that $x_{n} = \left( 1 - h(n) \right)^{e_n} x_{n-1}$ and, thus, for $n \ge 2$, we have
   \begin{align*}
     0 \le  x_{n-1} - x_n & = x_{n-1} \left( 1 - \left( 1 - h(n) \right)^{e_n} \right) \\
                          & \le  1 - \left( 1 - h(n) \right)^{e_n}
                            \le h(n) e_n \le \frac{\pi}{c r} r^n,
   \end{align*}
   where we used the fact that $(1 - x)^k \ge 1 - k x$ for every $x \in (0,1)$ and $k \in \mathbb{Z}^+$, and the bounds on $h(n)$ and $e_n$ (see \eqref{eq:hn} and \eqref{eq:bound-dn}). We obtain
   \begin{align*}
     \left| x_{n+q} - x_n \right| \le \sum_{j = 1}^{q} \left| x_{n+j} - x_{n+j-1}\right| \le \frac{\pi}{cr}\sum_{j = 1}^{q} r^{n+j} = \frac{\pi \left( 1 - r^{q} \right)}{c ( 1 - r )} r^{n}
   \end{align*}
   for any positive integers $n$ and $q$. Hence, by making $q$ tend to $\infty$ we conclude that
   \begin{equation} \label{eq:approx-argument:2}
     \left| x_n - w \right| \le \frac{\pi}{c (1-r)}r^n
   \end{equation}
   for every positive integer $n$.

   In light of \eqref{eq:approx-argument:1} and \eqref{eq:approx-argument:2}, we obtain  for every positive integer $n$ that 
   \begin{align*}
     \left| \prod\nolimits_{j = 1}^{n}  z_{j}^{e_j} - w \right| & \le   \left| \prod\nolimits_{j = 1}^{n}  z_j^{e_j} - x_n\right|  + \left| x_n - w \right|  = \left| x_n \right| \left| \prod\nolimits_{j = 1}^{n} e^{i e_j \theta_j} -  1\right|  + \left| x_n - w \right| \\
                                                                & \le \left| \prod\nolimits_{j = 1}^{n} e^{i e_j \theta_j} - 1\right|  + \left| x_n - w \right|  = \left|  e^{i  \tau_{n}} -  1\right|  + \left| x_{n} - w \right|  \le \frac{\pi}{2} r^{n} + \frac{\pi}{c (1-r)} r^{n}.
   \end{align*} 
   
   Finally, we argue that if the representation of $z_n$ can be computed in polynomial time in $n$, then $e_n$ can be computed in polynomial time in $n$ and we have $w \in \mathsf{P}_{\mathbb{R}}$. Note that $e_1$ is the smallest positive integer such that $\mathrm{Arg} \left(z_1^{e_1}\right) \in [3\pi/2, 2 \pi) \cup \{0\}$ and $\mathrm{Arg} (z_1^{e_1+1}) \in (0, \pi / 2]$ and, thus, $e_1$ can be computed by sequentially applying Lemma~\ref{lem:algebraic:argument} with intervals $[3\pi/2, 2 \pi]$ and $[0, \pi / 2]$, with the $z$ of Lemma~\ref{lem:algebraic:argument} equal to $z^k$ for every positive integer $k \le e_1 + 1$. This takes constant time since the quantities and objects involved are constant. For $n \ge 2$, let us assume that we have computed $e_1, \ldots, e_{n-1}$, and let $y_{n-1} = \prod_{j = 1}^{n-1} z_j^{e_j}$ (so $\tau_{n-1} = \mathrm{Arg} \left( y_{n-1} \right)$). Since the sequence $\{e_n\}$ is bounded and the length of the representation of $z_n$ is bounded by a polynomial in $n$, the computation of $y_{n-1}$ takes polynomial time in $n$. Then $e_n$ is the smallest non-negative integer such that $\mathrm{Arg} \left( y_{n-1} z_n^{e_n}\right) \in [3\pi/2, 2 \pi) \cup \{0\}$ and $\mathrm{Arg} \left( y_{n-1} z_n^{e_n+1}\right) \in (0, \pi / 2]$, and we can compute $e_n$ again by sequentially applying Lemma~\ref{lem:algebraic:argument} with intervals $[3\pi/2, 2 \pi]$ and $[0, \pi / 2]$, with the $z$ of Lemma~\ref{lem:algebraic:argument} equal to $z^k$ for every positive integer $k \le e_n + 1$. There is a bounded number of applications of Lemma~\ref{lem:algebraic:argument} because $e_n$ is bounded, and each application takes polynomial time in $n$ because the length of the representation of $y_{n-1} z_n^{k}$ is polynomial in $n$ for any $k \in \{1, 2, \ldots, e_n\}$. We conclude that $w$ is the limit of a sequence of algebraic numbers that converges exponentially fast and the representation of its $n$-th element can be computed in polynomial time in $n$. As a consequence, we have $w \in \mathsf{P}_{\mathbb{R}}$ by Lemma~\ref{cor:pc:limit}.
 \end{proof}

  \begin{lemma} \label{lem:implement-approx:real-with-complex} Let $q$ be a real algebraic number with $q > 0$. Let $x$ and $y$ be algebraic numbers such that $(x,y) \in \mathcal{H}_q$, $y \not \in \mathbb{R}$ and $\left| x \right| > 1$. Then there is a polynomial-time approximate theta shift from $(x,y)$ to $(x',y')$ for some $(x', y') \in \mathcal{H}_q$ with $y' \in (0,1) \cap \mathsf{P}_{\mathbb{R}}$.
  \end{lemma}

  \begin{proof}
    Since $y \not \in \mathbb{R}$, we have $x \not \in \mathbb{R}$ (Remark~\ref{rem:norm-x}). Let us write $x = R e^{i \theta}$ for some $R > 1$ and $\theta \in (0, 2\pi)$. An $m$-stretch gives a shift from $(x,y)$ to $(x^m, y_m)$ with $y_m = (x^m + q - 1) / (x^m - 1)$. By plugging $x = R e^{i \theta}$ in the definition of $y_m$ and multiplying by $R^m e^{-i m \theta} - 1$ in the numerator and denominator, we obtain
    \begin{equation} \label{eq:approx-real-with-complex:wl}
        y_m  = \frac{R^{2m}-q+1 + (q-2) R^m  \cos \left( m \theta \right) - i q R^m \sin \left( m \theta \right)}{1 + R^{2m} - 2 R^m \cos \left( m \theta \right)}.
    \end{equation}
    If $\theta \in \{\pi / 2, 3 \pi /2\}$, that is, $x \in i \mathbb{R}$, then for $m \equiv 2 \pmod {4}$ we have $ \cos \left( m \theta \right) = -1$,  $ \sin \left( m \theta \right) = 0$  and
    \begin{equation*}
      y_m = \frac{\left( 1 + R^m \right)^2 - q \left( 1 + R^m \right)}{\left( 1 + R^m \right)^2} = \frac{ 1 + R^m - q }{ 1 + R^m }.
    \end{equation*}
    Hence, for $m \equiv 2 \pmod {4}$ such that $1 + R^m > q$, we have $y_m \in (0, 1)$, so we can choose $y' = y_m$ and we are done.

    In the rest of the proof we assume that $\theta \not\in \{\pi / 2, 3 \pi /2\}$. We are going to apply Lemma~\ref{lem:approx-argument} to a subsequence of $y_m$. First, we invoke Corollary~\ref{cor:sigma-n} with $z = x$ in order to find a sequence $\sigma(m)$, a positive integer $k$ and a positive rational $C$ that satisfies:
    \begin{itemize}
    \item $\sigma(m)$ can be computed in polynomial time in $m$;
    \item $k$ and $C$ can be computed in constant time from $x$;
    \item $m \le \sigma(m) \le m + k - 1$ for every positive integer $m$;
    \item $\sin(\sigma(m)\theta) \le -C$ and $\cos(\sigma(m)\theta) \le -C$ for every positive integer $m$.
    \end{itemize}
    It follows that
    \begin{equation*}
     \mathrm{Re} \left( x^{\sigma(m)}\right) = \mathrm{Re} \left( R^{\sigma(m)} e^{i\sigma(m)\theta} \right) \le - C R^{\sigma(m)} \le - C R^{m}. 
    \end{equation*}
    Since $R > 1$, we can compute a positive integer $m_1$ such that for $m \ge m_1$ we have $\mathrm{Re} ( x^{\sigma(m)}) < 1 - q / 2$ and, thus, $\left| y_{\sigma(m)} \right| < 1$ (recall that $y_m = \left(x^m + q - 1 \right) / \left(x^m - 1 \right)$ and Remark~\ref{rem:norm-x}).
    Let
    \begin{align*}
      a_m & = 1-\mathrm{Re}(y_m) =  \frac{ q  - q R^m \cos \left( m \theta \right)}{1 + R^{2m} - 2 R^m \cos \left( m \theta \right)}; \\
      b_m & = \mathrm{Im}(y_m) = \frac{ - q R^m \sin \left( m \theta \right)}{1 + R^{2m} - 2 R^m \cos \left( m \theta \right)};
    \end{align*}
    that is, $y_m = 1 - a_m + i b_m$. We have
    \begin{equation*}
      R^{2\sigma(m)} \le 1 + R^{2\sigma(m)} - 2 R^{\sigma(m)} \cos \left( \sigma(m)\theta \right) \le 4 R^{2 \sigma(m)}.
    \end{equation*}
    Therefore, we obtain
    \begin{equation}
      \label{eq:bounds:ab}
      \frac{q C}{4}R^{-\sigma(m)}  \le a_{\sigma(m)}  \le 2 q R^{-\sigma(m)}, \qquad   \frac{q C}{4} R^{-\sigma(m)}   \le b_{\sigma(m)} \le q R^{-\sigma(m)}.
    \end{equation}
    We compute a positive integer $m_2$ such that $m_2 \ge \log_R (4q)$ and $m_2 \ge m_1$. We also compute a rational number $c$ with $c \in (0, q C R^{-m_2-k-1} / 4)$. Note that computing these quantities takes constant time. Let $f(m) = a_{\sigma(m+m_2)}$ and $g(m) = b_{\sigma(m+m_2)}$. In view of \eqref{eq:bounds:ab} and the inequalities $R^{-m-k+1} \le R^{-\sigma(m)} \le R^{-m}$, we find that
    \begin{equation} \label{eq:bounds:fn-gn}
        c R^{-m}  \le f(m) \le \frac{1}{2}R^{-m}, \qquad         c R^{-m} \le g(m) \le \frac{1}{2}R^{-m},
    \end{equation}
    for any positive integer $m$. The sequence $\{z_m\} = \{y_{\sigma(m+m_2)}\}$ satisfies
    \begin{itemize}
    \item $\big| z_m \big| < 1$ for every positive integer $m$;
    \item $z_m = 1 - f(m) + i g(m)$ with $f, g \colon \mathbb{Z}^+ \to (0,1)$;
    \item $f$ and $g$ are bounded as in \eqref{eq:bounds:fn-gn}.
    \item $z_m$ is an algebraic number whose representation can be computed in polynomial time in $m$. This is due to the facts that  $z_m = (x^{\sigma(m+m_2)} + q - 1) / (x^{\sigma(m+m_2)} - 1)$, $\sigma(m)$ can be computed in polynomial time in $m$, and $\sigma(m) = O(m)$.
    \end{itemize}
    Therefore, we can apply Lemma~\ref{lem:approx-argument} to the sequence $\{z_m\}$ for $r = R^{-1}$. There are $y' \in (0,1) \cap \mathsf{P}_{\mathbb{R}}$ and a bounded sequence of positive integers $\{e_m\}$ such that
    \begin{equation*}
      \left| \prod\nolimits_{j = 1}^{m} z_{j}^{e_j} - y' \right| \le \left( \frac{\pi}{2} + \frac{\pi}{c (1-1/R)} \right) R^{-m}
    \end{equation*}
    for every positive integer $m$. Moreover, we can compute $e_m$ in polynomial time in $m$. Let $M = \pi / 2 + \pi / (c (1-1/R))$. For any positive integer $n$, we can compute an integer $m$ with $m \ge \log_{1/R} \left(2^{-n} / M\right)$ and $m = \Theta(n)$ in polynomial time in $n$. We obtain
    \begin{equation*}
      \left| \prod\nolimits_{j = 1}^{m} z_j^{e_j} - y' \right| \le 2^{-n}.
    \end{equation*}
    This gives the following polynomial-time approximate theta shift from $(x,y)$ to $(x', y')$, where $x' = 1 + q / (y'-1)$. For each positive integer $n$ we return a graph $J_n$ that is the parallel composition of the path graphs that are used to implement the weights $y_{\sigma(j+m_2)}$, each one repeated $e_j$ times, for $j \in \{1, \ldots, m\}$. The graph $J_n$ $(x,y)$-implements $(\hat{x}, \hat{y}) \in \mathcal{H}_q$ for $\hat{y} = \prod\nolimits_{j = 1}^m z_j^{e_j} = \prod_{j=1}^m y_{\sigma(j+m_2)}^{e_j}$.
  \end{proof}

  \begin{corshifts}
    \statecorshifts
  \end{corshifts}
  
  \begin{proof} If $y \in (-1,0)$, then a $2$-thickening of $(x,y)$ gives the result. Hence, let us assume that $y \not \in (-1,0)$ in the rest of the proof. There are two cases:  
    \begin{itemize}
      \item $q \ne 2$ or $y \not \in i \mathbb{R}$. We apply either Lemma~\ref{lem:norm-shift}, Lemma~\ref{lem:shift:case-3} or Lemma~\ref{lem:shift:case-4}, depending on $\mathrm{Arg}(y)$, to find a shift from $(x,y)$ to $(x_1, y_1) \in \mathcal{H}_q$ with $y_1 \not \in \mathbb{R}$ and $\left| x_1 \right| > 1$. The graph of this shift is series-parallel. Then we apply Lemma~\ref{lem:implement-approx:real-with-complex} to obtain a polynomial-time approximate theta shift from $(x_1, y_1)$ to some $(x', y') \in \mathcal{H}_q$ with $y' \in (0,1) \cap \mathsf{P}_{\mathbb{R}}$. The result follows from the transitivity property of shifts.
      \item $q = 2$ and $y \in i \mathbb{R}$. Since $y \ne \pm i$, Lemma~\ref{lem:shift:case-4} gives a shift from $(x,y)$ to $(x', y')$ for some $(x',y') \in \mathcal{H}_q$ with $y' \in (-1,0)$. A $2$-thickening of $(x',y')$ gives the result. 
      \end{itemize}
    The fact that  $x' \in \mathsf{P}_{\mathbb{R}}$ follows from $x' = 1 + q / (y'-1)$ and $y' \in \mathrm{P}_{\mathbb{R}}$.
  \end{proof}

\subsection{Approximate shifts for polynomial-time computable real numbers}\label{sec:rvrtvtvt}

In this subsection we show how we can obtain a polynomial-time approximate shift from $(x,y)$ to $(x', y')$ for any $(x,y), (x',y') \in \mathcal{H}_q$ when $q \ge 2$, $y \in (0,1) \cap \mathsf{P}_{\mathbb{R}}$ and $y'$ is a positive real algebraic number (Lemma~\ref{lem:implement-approx:pc}). This extends a particular case of Lemma~\ref{lem:implement-approx:literature:1} to polynomial-time computable numbers. Our proof follows the same approach as that of~\cite[Lemma 22]{Goldberg2019} but we have to overcome some difficulties that arise when working with the class of numbers $\mathsf{P}_{\mathbb{R}}$. These difficulties will become apparent in the proof, but the reader that is familiar with the literature might want to skip the proof. Then we combine this result and Lemma~\ref{cor:implement-approx:real-with-complex} to prove Theorem~\ref{thm:approx-shifts}, the main result of Section~\ref{sec:complex-implementations}.

\begin{lemshifts} 
  \statelemshifts
\end{lemshifts}
\begin{proof}
  If $w = 1$, then $J$ is the graph with vertices $s$ and $t$ and no edges. In the rest of the proof we focus on the case $w \in (0,1)$.

  Recall that $x = 1 + q / (y-1)$. Since $q > 0$ and $y \in (1-q/2,1)$, we find that $x \in (-\infty, -1)$. Let $y_j = 1 + q / (x^j - 1)$. A $j$-stretch gives a shift from $(x,y)$ to $(x^j, y_j)$. If $j$ is even, then $x^j > 1$ and $y_j > 1$. Otherwise, $x^j < -1$ and $y_j \in (1-q/2,1)$. Moreover, the sequences $\left\{y_{2j+1}\right\}$  and $\left\{y_{2j}\right\}$ are increasing and decreasing, respectively, and $\left|y_j - 1 \right|$ can be made exponentially small as a function of $j$. We use these properties of $y_j$ to show that we can compute $y_{(e_1, \ldots, e_m)} = \prod_{j = 1}^m y_j^{e_j}$ such that $\left| y_{(e_1, \ldots, e_m)} - w \right| \le 2^{-n}$. Let $J$ be the parallel composition of the path graphs that $(x,y)$-implement $(x^j, y_j)$, each one repeated $e_j$ times, for $j \in \{1, \ldots, m\}$. Then $J$ is a theta graph and, in view of \eqref{eq:theta:implement}, we have $w(G; q, y-1) =  y_{(e_1, \ldots, e_m)} - 1$, that is, $J$ $(x,y)$-implements $(\hat{x}, \hat{y}) \in \mathcal{H}_q$ with $\hat{y} = y_{(e_1, \ldots, e_m)}$. The graph $J$ is the theta graph output by our algorithm.
  
  First, we define a sequence $\{d_j\}$ that will be related to the exponents $e_1, \ldots, e_m$. Since $q, x \in \mathsf{P}_{\mathbb{R}}$, we can compute rational upper bounds of $q$ and $x$ (Lemma~\ref{lem:pr:bounds}) and, with the help of these bounds, a positive integer $j_0$ such that $j_0 > \log_{\left| x \right|} q$. Let $d_j = 0$ for every positive integer $j$ with $j < j_0$ and let $d_j = 0$ for every even positive integer $j$.  For $j$ odd with $j \ge j_0$ we define $d_j$ recursively as the largest non-negative integer such that $y_{(d_1, \ldots, d_j)} \ge w$. The integer $d_j$ is well-defined because $0 < y_j < 1$ when $j$ is odd and $j \ge j_0$. An equivalent definition is that $\{d_j\}$ satisfies 
  \begin{equation} \label{eq:approx-real:dn:<}
    y_j < w / y_{(d_1, \ldots, d_j)} \le 1
  \end{equation}
  for every odd integer $j$ with $j \ge j_0$. A similar sequence $\{d_j\}$ is used in the proofs of~\cite[Lemma 22]{Goldberg2019} and~\cite[Lemma 3.28]{Goldberg2012}. For any odd integer $m$ with $m \ge \log_{|x|} \left( q 2^{n} - 1 \right)$ we have $0 \le 1- y_m \le 2^{-n}$ and, in light of \eqref{eq:approx-real:dn:<},
  \begin{equation*}
    0 \le 1 - w / y_{(d_1, \ldots, d_m)} \le 1 - y_m  \le 2^{-n}. 
  \end{equation*}
  Since $1 \ge y_{(d_1, \ldots, d_m)} \ge w$, it follows that
  \begin{equation} \label{eq:approx-real:dn}
    \left| w - y_{(d_1, \ldots, d_m)} \right| \le y_{(d_1, \ldots, d_m)} 2^{-n} \le 2^{-n}.
  \end{equation}  

  Now we study the size of the integers $d_1, \ldots, d_m$. We bound $d_j$ using an argument given in~\cite[Lemma 3.28]{Goldberg2012}. First, we show that $d_{j_0}$ is $O(k)$. We have $y_{j_0}^{d_{j_0}} \ge w$. We obtain
  \begin{equation*}
    d_{j_0} \le \log_{y_{j_0}} (w) = \log_y (w) \log_{y_{j_0}}(y).
  \end{equation*}
  Since $w \in [y^{k}, 1)$ and $\log_{y_{j_0}}(y) > 0$, it follows that $0 < \log_y (w) \le k$ and $d_{j_0} \in O(k)$. Now we show that $d_j$ is bounded for any $j > j_0$. By applying \eqref{eq:approx-real:dn:<} twice, we find that
  \begin{equation*}
    y_{j-2} < w / y_{(d_1, \ldots, d_{j-2})} = w  y_j^{d_j} / y_{(d_1, \ldots, d_j)} \le y_j^{d_j}
  \end{equation*}
  for every odd integer $j$ with $j > j_0$. It follows that $d_j \le \log (y_{j-2}) / \log(y_j)$ (here and in the rest of this paper $\log$ is taken in base $e$).  For every $x \in (1,5/4)$, we have $3 (x-1) / 4 \le \log (x) \le x - 1$. Hence, we obtain
  \begin{align*}
    d_j \le  \frac{\log (y_{j-2})}{\log(y_j)}  = \frac{\log (1/y_{j-2})}{\log(1/y_j)} & \le \frac{4}{3} \frac{1/ y_{j-2} - 1}{1 /y_j - 1} \\ & = \frac{4 y_{j}}{3 y_{j-2}} \frac{ 1 - y_{j-2}}{1 -  y_j} = \frac{4 y_j}{3 y_{j-2}} \frac{\left| x \right|^{j} + 1}{\left| x \right|^{j-2} + 1} \le \frac{4 y_{j}}{3 y_{j-2}}  \left| x \right|^{2},
  \end{align*}
  where the last inequality is a consequence of $\left| x \right|^2 (\left| x \right|^{j-2} + 1) \ge \left| x \right|^{j} + 1$. Since $y_j / y_{j-2}$ converges to $1$ and, thus, is bounded, it follows that $d_j$ is bounded. We conclude that $\sum_{j = 1}^m d_j = O(k+m)$.

  Let us assume that we can compute $d_1, \ldots, d_m$ for $m = \lceil 1 + \log_{|x|} \left( q 2^n - 1\right) \rceil$. In light of \eqref{eq:approx-real:dn}, we can return $J$ as the theta graph that implements the weight $w(J; q, y-1) = y_{(d_1, \ldots, d_m)} - 1$. Since $\sum_{j = 1}^m d_j = O(k+m)$ and $m = \Theta (n)$, the size of $J$ is at most a polynomial in $k$ and $n$.

  If $y$ were algebraic, computing $d_1, \ldots, d_m$ in polynomial time would be straightforward from their definition because we can efficiently check inequalities between real algebraic numbers as explained in Section~\ref{sec:complex-implementations:algebraic}. This is the approach followed in~\cite[Lemma 22]{Goldberg2019}. However, we only know that $y \in \mathsf{P}_{\mathbb{R}}$ and, thus, it is not clear how to efficiently determine whether $y_{(d_1,\ldots, d_{m-1}, d)} \ge w$ or not for any given $d$. In the rest of this proof, we show how to overcome this difficulty.

  Let $n$ be a positive integer, so $2^{-n}$ is the desired accuracy for our algorithm. Let us assume that we have computed the integers $d_1, \ldots, d_{j-1}$ and we want to compute $d_j$ for an odd positive integer $j$ with $j \ge j_0$. We are going to sequentially try all the values $d = 0, 1, \ldots$ until we have
  \begin{equation*}
    y_j < \frac{w}{y_{(d_1, \ldots, d_{j-1}, d)}} \le 1,
  \end{equation*}
  in which case we have found the value $d_j$ (see \eqref{eq:approx-real:dn:<}). Recall that $y_{(d_1, \ldots, d_{j-1}, d)}-1$ is the weight implemented by a theta graph $J_d$ whose size is bounded by a polynomial in $k$ and $j$. Therefore, by applying Lemma~\ref{lem:pc:essential} with $G = J_d$ and $\gamma = y - 1$, we can compute in polynomial time in $n$ and the size of $J_d$, a positive integer $f(n+2, J_d)$ with $f(n+2, J_d) = n + \Theta (\mathrm{size}(J_d))$ such that if $\left| \gamma - \hat{\gamma} \right| \le 2^{-f(n+2, J_d)}$, then $\left| w(G; q, \gamma) - w(G; q, \hat{\gamma}) \right| \le 2^{-n-2}$. Since $y \in \mathrm{P}_{\mathbb{R}}$, we can compute a rational number $\hat{\gamma}$ such that $\left| \gamma - \hat{\gamma} \right| \le 2^{-f(n+2, J_d)}$ in polynomial time in $n$ and the size of $J_d$. Let $\hat{y}_{(d_1,\ldots, d_{j-1}, d)} = w(G; q, \hat{\gamma}) + 1$. Then we have computed in polynomial time in $k, j$ and $n$ a rational number $\hat{y}_{(d_1,\ldots, d_{j-1}, d)}$ such that
  \begin{equation*}
    \left| \hat{y}_{(d_1, \ldots, d_{j-1}, d)} - y_{(d_1, \ldots, d_{j-1}, d)} \right| \le 2^{-n-2}.
  \end{equation*} 
  Because $\big| \hat{y}_{(d_1, \ldots, d_{j-1}, d)} - w \big|$ is a real algebraic number, we can check if the following inequality holds in polynomial time,
  \begin{equation}
    \label{eq:approx-real:condition}
    \big| \hat{y}_{(d_1, \ldots, d_{j-1}, d)} - w \big| \le 2^{-n-1}.
  \end{equation}
  If that is the case, then
  \begin{equation*}
    \left| y_{(d_1, \ldots, d_{j-1}, d)} - w \right| \le \left| y_{(d_1, \ldots, d_{j-1}, d)} - \hat{y}_{(d_1, \ldots, d_{j-1},d)} \right| + \left| \hat{y}_{(d_1, \ldots, d_{j-1}, d)} - w \right| \le 3 \cdot 2^{-n} / 4 < 2^{-n},
  \end{equation*}
  so $y_{(d_1, \ldots, d_{j-1}, d)}$ is a good enough approximation of $w$ and we can stop the algorithm (even though we have not computed $d_j$). Otherwise, we claim that $\hat{y}_{(d_1, \ldots, d_{j-1}, d)} \ge w$ if and only if $y_{(d_1, \ldots, d_{j-1}, d)} \ge w$. If $\hat{y}_{(d_1, \ldots, d_{j-1}, d)} \ge w$ and $w > y_{(d_1, \ldots, d_{j-1}, d)}$, then
  \begin{equation*}
    \big| \hat{y}_{(d_1, \ldots, d_{j-1}, d)} - w \big| \le \big| \hat{y}_{(d_1, \ldots, d_{j-1}, d)} - y_{(d_1, \ldots, d_{j-1}, d)} \big| \le 2^{-n-2}
  \end{equation*}
  and \eqref{eq:approx-real:condition} holds, a contradiction. The same reasoning applies when $\hat{y}_{(d_1, \ldots, d_{j-1}, d)} < w$ and $w \le y_{(d_1, \ldots, d_{j-1}, d)}$. Hence, we can check whether $y_{(d_1,\ldots, d_{j-1}, d)} \ge w$ or not by checking $\hat{y}_{(d_1,\ldots, d_{j-1}, d)} \ge w$, provided that \eqref{eq:approx-real:condition} does not hold. This gives a procedure to compute $d_j$ for odd $j$ with $j \ge j_0$:
  \begin{enumerate}[label=\arabic*.]
  \item Set $d = 0$.
  \item If \eqref{eq:approx-real:condition} holds, then return $d$. We have failed to compute $d_j$, but we have succeeded in finding an approximation of $w$.
  \item If $\hat{y}_{(d_1, \ldots, d_{j-1}, d+1)} \ge w$, then increase $d$ by $1$ and go to step $2$. Else, we have $d_j = d$.
  \end{enumerate}

  We repeat this procedure to compute $d_j$ sequentially until \eqref{eq:approx-real:condition} holds, in which case we stop and return the graph $J$ associated to $y_{(d_1, \ldots, d_{j-1}, d)}$.
  
  It remains to show that this procedure always halts and runs in polynomial time. In light of \eqref{eq:approx-real:dn}, we find that, for odd $m \ge \log_{|x|}( q2^{n+2} - 1 )$,
  \begin{align*}
    \left| \hat{y}_{(d_1, \ldots, d_m)} - w \right| & \le \left| \hat{y}_{(d_1, \ldots, d_m)} - y_{(d_1, \ldots, d_m)} \right| + \left| y_{(d_1, \ldots, d_m)} -w \right|  \le  2^{-n-1},
  \end{align*}
  that is, \eqref{eq:approx-real:condition} holds. Therefore, our procedure that computes non-negative integers $d_1, \ldots, d_{m-1}, d$ with $\left| y_{(d_1, \ldots, d_{m-1}, d)} - w \right| \le 2^{-n}$ halts for $m = O(n)$. As a consequence, the whole procedure runs in polynomial time in $k$ and $n$.
\end{proof}

The proof of Lemma~\ref{lem:implement-approx:pc} can be adapted to to the case $w \in (1, y^{-k}]$. The main difference is that this time we work with the decreasing sequence $\{y_{2j}\}$. We set $d_j = 0$ for odd $j$ and, for even $j$, we define $d_j$ recursively as the largest non-negative integer such that $y_{(d_1, \ldots, d_j)} \le w$. The details of the proof are left to the reader. When studying the hardness of approximating $Z_{\text{Tutte}}(G; q, \gamma)$ we only need the version stated in Lemma~\ref{lem:implement-approx:pc}.

\begin{thmshifts}
\statethmshifts
\end{thmshifts}

\begin{proof}
	First, let us assume that $y' \in (0, 1]$. By Lemma~\ref{cor:implement-approx:real-with-complex}, there is a polynomial-time approximate series-parallel shift from $(x,y)$ to $(\tilde{x}, \tilde{y})$ for some $(\tilde{x}, \tilde{y}) \in \mathcal{H}_q$ with $\tilde{x}, \tilde{y} \in \mathsf{P}_{\mathbb{R}}$ and $\tilde{y} \in (0,1)$. Since $q \ge 2$, we have $1-q/2 \le 0$ and $\tilde{y} \in (1-q/2, 1)$. Hence, Lemma~\ref{lem:implement-approx:pc} gives us a polynomial-time approximate theta shift from $(\tilde{x},\tilde{y})$ to $(x', y')$. Since $\tilde{y} \not \in 1 - q/2 + i \mathbb{R} =  1 - q/2 + i q\mathbb{R}$ and $\tilde{x}, \tilde{y} \in \mathrm{P}_{\mathbb{R}}$, the transitivity property of polynomial-time approximate shifts, Lemma~\ref{lem:approx-shifts:transitivity}, for $(x_1, y_1) = (x,y)$, $(x_2, y_2) = (\tilde{x},\tilde{y})$ and $(x_3, y_3) = (x',y')$ gives us a polynomial-time approximate series-parallel shift from $(x,y)$ to $(x',y')$.
	
	Now we treat the case $y' = 0$. As a consequence of what we have just shown in the paragraph above, there is a polynomial-time approximate series-parallel shift from $(x,y)$ to $(1-2q, 1/2) \in \mathcal{H}_q$. An $n$-thickening gives a shift from $(1-2q, 1/2)$ to $(x_n, 2^{-n})$, where $x_n = 1 + q / (2^{-n} - 1)$, so there is also a polynomial-time approximate theta shift from $(1-2q, 1/2)$ to $(1-q, 0)$. We conclude that there is a polynomial-time approximate series-parallel shift from $(x,y)$ to $(1-q, 0)$ by applying Lemma~\ref{lem:approx-shifts:transitivity} with  $(x_1, y_1) = (x,y)$, $(x_2, y_2) = (1-2q,1/2)$ and $(x_3, y_3) = (1-q,0)$. Note that we can indeed apply Lemma ~\ref{lem:approx-shifts:transitivity} because $1-2q, 1/2 \in \mathrm{P}_{\mathbb{R}}$ and $1/2 \not \in  1 - q/2 + i q\mathbb{R}$.
	
	Now we deal with the case $y' > 1$. We use again the polynomial-time approximate series-parallel shift from $(x,y)$ to $(x_1,y_1) = (1-2q, 1/2) \in \mathcal{H}_q$. We use a $2$-stretch to $(1-2q, 1/2)$-implement $(x_2, y_2)$ with $x_2 = (1-2q)^2 \ge 9$ and $y_2 = 1 + q/(x_2-1) > 1$. Hence, there is a polynomial-time approximate series-parallel shift from $(x,y)$ to $(x_2, y_2)$. Since $x_2$ and $y_2$ are real algebraic numbers with $y_2 > 1$ and $(x_2-1)(y_2-1) = q > 0$, in view of Lemma \ref{lem:implement-approx:literature:1}, we have a polynomial-time approximate theta shift from $(x_2,y_2)$ to $(x', y')$. Note that $y_2 \not \in \{1\} \cup (1-q/2+i q \mathbb{R})$. Hence, we can apply the transitivity property shown in Lemma \ref{lem:approx-shifts:transitivity} with $(x_1, y_1) = (x, y)$, $(x_2, y_2) = (x_2, y_2)$ and $(x_3, y_3) = (x',y')$ and find a polynomial-time approximate series-parallel shift from $(x,y)$ to $(x', y')$.
	
	Finally, we study the case $y' < 0$. In light of Corollary~\ref{cor:implement-weight:approx-1}, there is a polynomial-time approximate series-parallel shift from $(x,y)$ to $(0, 1-q)$. Note that $1-q \le -1$. In this proof we have already shown that there is a polynomial-time approximate series-parallel shift from $(x,y)$ to $(x_3, y_3) \in \mathcal{H}_q$ for $y_3 = y' / (1-q) > 0$. Since $y' = y_3 (1-q)$, by Lemma~\ref{lem:approx-shifts:composition} with parameters $(x_1, y_1) = (x, y)$, $(x_2, y_2) = (0, 1-q)$ and $(x_3, y_3) = (x_3, y_3)$, we conclude that there is a polynomial-time approximate series-parallel shift from $(x,y)$ to $(x', y')$.
\end{proof}

\section{Hardness results} \label{sec:hardness}
We begin with obtaining lower bounds on $Z_{\text{Tutte}}(G; q, \gamma)$  for algebraic numbers $q$ and $\gamma$. In Section~\ref{sec:hardness:lll}, we review the algorithm of~\cite{Kannan1988} for computing algebraic representations, and in Section~\ref{sec:pre:hardness} the exact $\numP$-hardness results that we will use. The rest of the section gives various ingredients that are needed in the reduction, which are put together in Section~\ref{sec:mainproof} where we prove all of our main theorems.
\subsection{Properties of $Z_{\text{Tutte}}(G; q, \gamma)$ for algebraic numbers $q$ and $\gamma$} \label{sec:hardness:lower-bound}

In this section we give a lower bound on $Z_{\text{Tutte}}(G; q, \gamma)$ and study the degree and height of $Z_{\text{Tutte}}(G; q, \gamma)$ when $q$ and $\gamma$ are algebraic numbers. First, we have to introduce some concepts and results from algebraic number theory. The \emph{degree} of an algebraic number $\gamma$ is the degree of its minimal polynomial $p$, and we denote it by $d(\gamma)$. Recall that the \emph{degree of a field extension} $F / K$ is the dimension of $F$ as a $K$-vector space, and it is denoted by $[F : K]$. It is well-known that if $\gamma$ is algebraic, then $[ K(\gamma) : K ]$ is the degree of the minimal polynomial of $\gamma$ over $K$~\cite[Chapter 5]{Stewart2004}. In particular, we have  $[ \mathbb{Q}(\gamma) : \mathbb{Q} ] = d(\gamma)$. The \emph{usual height} of a polynomial $f \in \mathbb{Z}[x_1, \ldots, x_m]$ is the  the largest value among the absolute values of its coefficients and it is denoted by $H(f)$. The \emph{usual height} of $\gamma$ is $H(\gamma) = H(p)$. One can find several (non-equivalent) definitions of the height of an algebraic number in the literature. Another one of these definitions is the absolute logarithmic height. First, we have to introduce the \emph{Mahler's measure} of a polynomial $f \in \mathbb{Z}[x]$, which is given by
\begin{equation*}
  M(f) = \left| a_d \right| \prod_{i = 1}^d \max \{1, \left|\alpha_i\right|\},
\end{equation*}
where $f(x) = \sum_{j = 0}^d a_j x^j$, $a_d \ne 0$, and  $\alpha_1, \ldots, \alpha_d$ are the roots of $f$. It is well-known that
\begin{equation} \label{eq:mahler-inequality}
  2^{-d(f)} H(f) \le M(f) \le H(f) \sqrt{d(f)+1},
\end{equation}
see~\cite[Lemma 3.11]{Waldschmidt2000}. The \emph{Mahler's measure} of an algebraic number $\gamma$ with minimal polynomial $p$ is $M(\gamma) = M(p)$. The \emph{absolute logarithmic height} of $\gamma$ is $h(\gamma) = d(\gamma)^{-1} \log M(\gamma)$. Note that $h(\gamma) \ge 0$ because $M(\gamma) \ge 1$. Now we can state a lower bound for the evaluation of a polynomial at algebraic numbers.

\begin{lemma}[{\cite[Section 3.5.4]{Waldschmidt2000}}] \label{lem:algebraic:lower-bound}
  Let $f \in \mathbb{Z}[x_1, \ldots, x_m]$ be a polynomial in $m$ variables and let $\gamma_1, \ldots, \gamma_m$ be algebraic numbers. If $f(\gamma_1, \ldots, \gamma_m) \ne 0$, then we have
  \begin{equation*}
    \left| f(\gamma_1, \ldots, \gamma_m) \right| \ge e^{-c T},
  \end{equation*}
  where $T = \deg f + \log H(f)$, $c = D (2 + h(\gamma_1) + \cdots + h(\gamma_m))$ and $D = [ \mathbb{Q}(\gamma_1, \ldots, \gamma_m) : \mathbb{Q} ]$.
\end{lemma}

\begin{corollary} \label{cor:lower-bound:tutte}
  Let $q$ and $\gamma$ be algebraic numbers. We can compute a rational number $C_{q, \gamma}$ with $C_{q, \gamma} > 1$ such that, for any graph $G$, either $Z_{\text{Tutte}}(G; q, \gamma) = 0$ or $\left| Z_{\text{Tutte}}(G; q, \gamma) \right| \ge C_{q, \gamma}^{-\mathrm{size}(G)}$.
\end{corollary}
\begin{proof}
  Recall that we represent an algebraic number $\gamma$ as its minimal polynomial $p$ and a rectangle of the complex plane where $\gamma$ is the only root of $p$. Let $G = (V, E)$ be a graph. Let $n = |V|$ and $m = |E|$. Let us assume that $Z_{\text{Tutte}}(G; q, \gamma) \ne 0$. We can apply Lemma~\ref{lem:algebraic:lower-bound} for $f(q,\gamma) = Z_{\text{Tutte}}(G; q, \gamma)$ to find that $\left| Z_{\text{Tutte}}(G; q, \gamma) \right| \ge e^{-c T}$, where $c$ and $T$ are as in Lemma~\ref{lem:algebraic:lower-bound}. We have $c = D (2 + h(q) + h(\gamma))$ and $D = [ \mathbb{Q}(q, \gamma) : \mathbb{Q} ]$, so $c \ge 2$. Note that, by definition of $Z_{\text{Tutte}}$, we have $H(f) \le 2^m$ and $\deg f \le n + m$. Hence, we find that $\left| Z_{\text{Tutte}}(G; q, \gamma) \right| \ge e^{-2 c \, \mathrm{size}(G)}$. It remains to compute a rational number $C_{q, \gamma}$ in $(e^{2 c}, \infty)$ to conclude the result. From $D = [ \mathbb{Q}(q, \gamma) : \mathbb{Q} ]$, we can compute $D$ exactly. Moreover, we can apply \eqref{eq:mahler-inequality} to upper bound $h(q)$ and $h(\gamma)$ in terms of the usual heights and degrees of $q$ and $\gamma$, and compute an appropriate rational number $C_{q, \gamma}$ with the help of these upper bounds.
\end{proof}

The case $q = 2$ (Ising model) of Corollary~\ref{cor:lower-bound:tutte} has previously been shown in~\cite[Lemma 6.4]{Goldberg2017}. Note that the approach followed in this section can be applied to obtain lower bounds for other partition functions.

In the rest of this section we upper bound the degree and the usual height of the algebraic number $Z_{\text{Tutte}}(G; q, \gamma)$ in terms of the usual heights and degrees of $q$ and $\gamma$. We will make use of these bounds in the proof of Lemma~\ref{lem:compute-fraction:q>1}.

Let $q$ and $\gamma$ be two algebraic numbers. By the tower law, we have $[ \mathbb{Q}(q, \gamma) : \mathbb{Q} ] = [ \mathbb{Q}(q, \gamma) : \mathbb{Q}(q) ] [ \mathbb{Q}(q) : \mathbb{Q} ] \le  d(q) d(\gamma)$, where we used that the degree of the minimal polynomial of $\gamma$ over $\mathbb{Q}(q)$ is bounded by $d(\gamma)$. Since $Z_{\text{Tutte}}(G; q, \gamma)$ is in $\mathbb{Q}(q, \gamma)$, it follows that its degree is bounded by $d(q) d(\gamma)$.

Now we argue how we can bound the usual height of $Z_{\text{Tutte}}(G; q, \gamma)$. A well-known property of the absolute logarithmic height is that $h(\alpha \beta) \le h(\alpha) + h(\beta)$, $h(\alpha + \beta) \le \log 2 +  h(\alpha) + h(\beta)$ and $h(1/\alpha) = h(\alpha)$~\cite[Property 3.3]{Waldschmidt2000}. Moreover, if $n$ is an integer, then $h(n) = \log |n|$. A more general property is the following one.

\begin{lemma}[{\cite[Lemma 3.7]{Waldschmidt2000}}] \label{lem:height:upper-bound}
  Let $f \in \mathbb{Z}[x_1, \ldots, x_t]$ be a non-zero polynomial in $t$ variables with integer coefficients. Let $\gamma_1, \ldots, \gamma_t$ be algebraic numbers. Then
  \begin{equation*}
    h \left( f \left( \gamma_1, \ldots, \gamma_t \right) \right) \le \log L(f) + \sum_{j = 1}^t \deg_{x_j} (f) h(\gamma_j), 
  \end{equation*}
  where $L(f)$ is the sum of the absolute values of the coefficients of $f$ and $\deg_{x_j} (f)$ is the degree of $f$ with respect to the $j$-th variable.
\end{lemma}

\begin{corollary} \label{cor:degree-height}
  Let $q$ and $\gamma$ be algebraic numbers. Then, for any graph $G =(V, E)$ with $n = |V|$ and $m = |E|$, we have
  \begin{equation*}
    d \left( \frac{Z_{s|t}(G; q, \gamma)}{Z_{st}(G; q, \gamma)}  \right) \le d \left( q \right) d \left( \gamma \right) \mbox{\ \  and\  \ }
    H \left( \frac{Z_{s|t}\left(G; q, \gamma\right)}{Z_{st}\left(G; q, \gamma\right)}  \right) \le \left(2^{m+1/2} e^{n h(q) + m h(\gamma)} \right)^{2 d(q) d(\gamma)}.
  \end{equation*}
\end{corollary}
\begin{proof}
  The degree bound on $Z_{s|t}(G; q, \gamma) / Z_{st}(G; q, \gamma)$ follows from the fact that that it is in $\mathbb{Q}(q, \gamma)$. For its absolute logarithmic height, we have
  \begin{equation*}
     h \left( \frac{Z_{s|t}(G; q, \gamma)}{Z_{st}(G; q, \gamma)} \right) \le h \left( Z_{st}(G; q, \gamma) \right) + h \left( Z_{s|t}(G; q, \gamma) \right).
  \end{equation*}
  Note that $L(Z_{st}(G; q, \gamma)) + L(Z_{s|t}(G; q, \gamma)) = 2^m$. As a consequence of Lemma~\ref{lem:height:upper-bound}, we find that
  \begin{equation*}
    h \left( Z_{st}(G; q, \gamma) \right) + h \left( Z_{s|t}(G; q, \gamma) \right)  \le  2 \left(m \log 2 + n h(q) + m h(\gamma)\right).
  \end{equation*}
  Recall that $M(\alpha) = \exp (d(\alpha) h(\alpha))$. Thus, the bounds on the Mahler's measure \eqref{eq:mahler-inequality} yield the inequality $H(\alpha) \le (2\exp ( h(\alpha)))^{d(\alpha)}$. We conclude that
  \begin{align*}
    H \left( \frac{Z_{s|t}(G; q, \gamma)}{Z_{st}(G; q, \gamma)} \right) & \le  \left(2  e^{2 \left(m \log 2 + n h(q) + m h(\gamma)\right) } \right)^{d(q) d(\gamma)} = \left(2^{m+1/2} e^{n h(q) + m h(\gamma)} \right)^{2d(q) d(\gamma)}. \qedhere 
  \end{align*}
\end{proof}

One could derive analogous bounds to those of Corollary~\ref{cor:degree-height} for the algebraic number $Z_{\text{Tutte}}(G; q, \gamma)$ by applying the same argument.

\subsection{Computing representations of algebraic numbers via approximations} \label{sec:hardness:lll}

Kannan, Lenstra and Lov\'asz~\cite{Kannan1988} showed how to reconstruct the minimal polynomial of an algebraic number from a certain number of digits of its binary expansion, and we will use their algorithm as a black-box in our reduction of Section~\ref{sec:hardness:real}, in the following form.

\begin{lemma}[{\cite[Theorem 1.19]{Kannan1988}}] \label{lem:lll}
  Let $\alpha$ be an algebraic number and let $d$ and $U$ be upper bounds on the degree and usual height, respectively, of $\alpha$. Suppose that we are given a rational approximation $\overline{\alpha}$ to $\alpha$ such that $\left| \alpha - \overline{\alpha} \right| \le 2^{-\nbits}/(12 d)$, where $\nbits$ is the smallest positive integer such that
  \begin{equation*}
    2^\nbits \ge 2^{d^2 / 2} (d + 1)^{(3d + 4)/2} U^{2d}.
  \end{equation*}
  Then the minimal polynomial of $\alpha$ can be determined in $O(d^5 (d + \log U ) )$ arithmetic operations on integers having $O(d^2 (d + \log U))$ binary bits.
\end{lemma}

The algorithm in Lemma~\ref{lem:lll} is based on the Lenstra–Lenstra–Lov\'asz lattice basis reduction algorithm, we refer the reader to~\cite{Yap2000} for more details.

\subsection{Exact Hardness results} \label{sec:pre:hardness}

We will use the following hardness results from~\cite{Jaeger1990} regarding the problem of exactly evaluating $Z_{\mathrm{Tutte}}(G;q,\gamma)$, given a graph $G$. We refer to this problem as $\ET(q, \gamma)$. Jaeger et al.~\cite{Jaeger1990} identify the following $9$ ``special'' points of the Tutte plane: $(1, -1)$, $(0,0)$, $(4, -2)$, $(2, -2)$, $(2, -1)$, $(2, -i-1)$, $(2, i-1)$, $(3, \omega_3^2-1)$, and $(3, \omega_3-1)$,  where $i = \sqrt{-1}$ and $\omega_3 = \exp(2 \pi i / 3)$.\footnote{In the $(x,y)$-parametrisation, the  special points are $(0,0)$, $(1,1)$, $(-1,-1)$,  $(0,-1)$, $(-1,0)$, $(i,-i)$, $(-i,i)$, $(\omega_3, \omega_3^2)$, and $(\omega_3^2, \omega_3)$.} With these special points in mind, their main result on the complexity of $\ET(q, \gamma)$ can be stated as follows.

\begin{theorem}[{\cite[Proposition 1]{Jaeger1990}}]  \label{thm:hardness:exact}
 Let $q$ and $\gamma$ be algebraic numbers. Then $\ET(q, \gamma)$ is $\# \mathsf{P}$-hard unless $q = 1$ or $(q, \gamma)$ is a special point, in which case $\ET(q, \gamma)$ is in $\mathsf{FP}$.
\end{theorem}

In~\cite{Vertigan2005}, Vertigan studied the complexity of the problem $\PT(q, \gamma)$, which also turns out to be hard for most parameters $q$ and $\gamma$.

\begin{theorem}[{\cite[Theorem 5.1]{Vertigan2005}}] \label{thm:hardness:exact-planar}
  Let $q$ and $\gamma$ be algebraic numbers. Then $\PT(q, \gamma)$ is $\# \mathsf{P}$-hard unless $q \in \{1, 2\}$ or $(q, \gamma)$ is a special point, in which case $\PT(q, \gamma)$ is in $\mathsf{FP}$.
\end{theorem}

\subsection{Computational problems} \label{sec:pre:problems}

In this section, we define a few computational problems that will be useful in our reductions; these were also considered in~\cite{Goldberg2017}.  Let $q$ be a real algebraic number,  $\gamma_1, \ldots, \gamma_k$ be algebraic numbers, and $K,\rho$ be real numbers with $K > 1$, $\rho>0$.

\prob{$\STz(q, \gamma_1, \ldots, \gamma_k)$.}{A (multi)graph $G$ and a weight function $\hat{\gamma} \colon E \to \{\gamma_1, \ldots, \gamma_k\}$.}{A correct statement of the form $Z_{\text{Tutte}}(G; q, \hat{\gamma}) \ge 0$ or  $Z_{\text{Tutte}}(G; q, \hat{\gamma}) \le 0$.}

\prob{$\FNTz(q, \gamma_1, \ldots, \gamma_k)$.}{A (multi)graph $G$ and a weight function $\hat{\gamma} \colon E \to \{\gamma_1, \ldots, \gamma_k\}$.}{If $Z_{\text{Tutte}}(G; q, \hat{\gamma})= 0$, the algorithm may output any rational number. Otherwise, it must output $\hat{N} \in \mathbb{Q}$ such that $\hat{N} / K \le |Z_{\text{Tutte}}(G; q, \hat{\gamma})| \le K \hat{N}$.}

\prob{$\pDAT \rho \sDATz(q, \gamma)$.}{A (multi)graph $G$.}{If $Z_{\text{Tutte}}(G; q, \gamma)= 0$, the algorithm may output any rational number. Otherwise, it must output $\hat{A} \in \mathbb{Q}$ such that, for some $a \in \arg(Z_{\text{Tutte}}(G; q, \gamma))$, we have $| \hat{A} - a| \le \rho$.}

 We also consider these problems for the Potts model (with parameters $q$ and $y = \gamma+1$), and we write $\textsc{Potts}$ instead of $\textsc{Tutte}$ in the name of these problems when we refer to the Potts ones.  We also consider all these problems restricted to planar graphs, in which case we write $\PT$ instead of $\ET$ in the name of the problem. It is a trivial observation that the planar case reduces to the general case.

\subsection{Reducing exact computation to sign and approximate computation} \label{sec:hardness:real}

In this section, we first review the binary search technique of~\cite{Goldberg2017}, which we will refer to as ``interval-shrinking''. Then, we use this to obtain several of our inapproximability theorems.

Let $f(\varepsilon) = -\varepsilon A + B$ be a linear function, where $A$ and $B$ are real algebraic numbers with $A \ne 0$. Let $\varepsilon^{*} = B/A$ be the zero of $f$. Let $(\varepsilon', \varepsilon'')$ be an open interval with length $l > 0$ such that $\varepsilon^{*}$ is in $(\varepsilon', \varepsilon'')$ or, equivalently, $f(\varepsilon') f(\varepsilon'') < 0$. We want to find a small open subinterval of $(\varepsilon', \varepsilon'')$ that contains $\varepsilon^{*}$. 

First, assume that we have an oracle that, on input $\varepsilon$, outputs the sign of $f(\varepsilon)$, unless when $f(\varepsilon) = 0$, in which case the output of the oracle is unreliable. Let $\varepsilon_0, \varepsilon_1, \ldots, \varepsilon_4$ be a partition of the interval $(\varepsilon', \varepsilon'')$ such that $\varepsilon_0 = \varepsilon'$, $\varepsilon_4 = \varepsilon''$ and $\varepsilon_{i+1} - \varepsilon_i \ge l / 10$ for every $i \in \{0, \ldots, 3\}$. We invoke the  oracle with input $\varepsilon_i$ to determine the sign of $f(\varepsilon_i)$ for every $i \in \{0,\ldots, 4\}$; let $s_i$ be the answer of the oracle. Then, we have a monotone sequence $s_0, \ldots, s_4$ of positive and negative signs with $s_0 \ne s_4$. Hence, there are two possibilities: either $s_0 = s_1 = s_2$, in which case $\varepsilon_1 < \varepsilon^{*}$ and we can recurse on $(\varepsilon_1, \varepsilon_4)$, or $s_2 = s_3 = s_4$, in which case $\varepsilon^{*} < \varepsilon_3$ and we can recurse on $(\varepsilon_0, \varepsilon_3)$. In any of these two cases, we can shrink the interval $(\varepsilon', \varepsilon'')$ to at most $9/10$ of its original length.  Then, recursively, we can  find an open subinterval of arbitrarily small length containing the zero of $f$.

Next, assume that we have an oracle that returns a multiplicative approximation to the norm of $f$. More accurately, let $\eta = 1/41$ and suppose that we have an oracle that, on input $\varepsilon$, returns a value $\hat{f}(\varepsilon)$ satisfying
\begin{equation*}
 \left(1-\eta \right) \left| f(\varepsilon) \right| < \frac{1}{1+\eta} \left| f(\varepsilon) \right| \le \hat{f}(\varepsilon) \le \left( 1 + \eta \right)\left| f(\varepsilon)\right| 
\end{equation*}
when $f(\varepsilon) \ne 0$ (otherwise the value  $\hat{f}(\varepsilon)$ is unreliable). The approach given in~\cite{Goldberg2017} by  Goldberg and Guo to shrink $(\varepsilon', \varepsilon'')$ is as follows. First, let us assume that $A > 0$, so $f$ is strictly decreasing. Let $\varepsilon_0, \varepsilon_1, \ldots, \varepsilon_{10}$ be a partition of the interval $(\varepsilon', \varepsilon'')$ such that $\varepsilon_0 = \varepsilon'$, $\varepsilon_{10} = \varepsilon''$ and $\varepsilon_{i+1} - \varepsilon_i \ge l / 20$ for every $i \in \{0, \ldots, 9\}$. These numbers are not chosen to be optimal but they suffice. We invoke our oracle to compute $\hat{f}(\varepsilon_i)$ for $i \in \{0, \ldots, 10\}$. Let $s_i$ be the sign (positive, negative, or zero) of $\hat{f}(\varepsilon_i) - \hat{f}(\varepsilon_{i+1})$ for each $i \in \{0, \ldots, 9\}$. We analyse the signs $s_i$ for $i \in \{0, \ldots, 9\}$.  First, we consider the case $\varepsilon_i < \varepsilon_{i+1} < \varepsilon^{*}$. Note that we have $f(\varepsilon_i) > f(\varepsilon_{i+1}) > 0$.  Moreover,
\begin{align*}
  \hat{f}\left(\varepsilon_i\right) - \hat{f}\left(\varepsilon_{i+1}\right) & \ge \left( 1 - \eta \right) f \left( \varepsilon_i \right) - \left( 1 + \eta \right) f \left( \varepsilon_{i+1} \right) \\
  & = A \left( \varepsilon_{i+1} - \varepsilon_i - \eta \left( 2 \varepsilon^{*} - \varepsilon_i - \varepsilon_{i+1} \right) \right).
\end{align*}
Note that $\varepsilon^* - \varepsilon_i$ and $\varepsilon^{*} - \varepsilon_{i+1}$ are both at most $l$ and, thus, we obtain $2 \varepsilon^{*} - \varepsilon_i - \varepsilon_{i+1} \le 2l$. So since $\eta = 1/41$ and $\varepsilon_{i+1} - \varepsilon_i \ge l / 20$, we conclude that $s_i$ is positive. Now we consider the case $\varepsilon^{*} < \varepsilon_i < \varepsilon_{i+1}$. This time we have $f(\varepsilon_{i+1}) < f(\varepsilon_i) < 0$,
\begin{align*}
  \hat{f}\left(\varepsilon_i\right) - \hat{f}\left(\varepsilon_{i+1}\right) & \le \left( 1 + \eta \right) \left(-f \left( \varepsilon_i \right)\right) - \left( 1 - \eta \right) \left(-f \left( \varepsilon_{i+1} \right)\right) \\
  & = -A \left( \varepsilon_{i+1} - \varepsilon_i - \eta \left(  \varepsilon_i + \varepsilon_{i+1} - 2 \varepsilon^{*} \right) \right),
\end{align*}
and $0 < \varepsilon_i + \varepsilon_{i+1} - 2 \varepsilon^{*} < 2l$. We conclude that $s_i$ is negative. If $\varepsilon_i \le \varepsilon^{*}$ and $\varepsilon^{*} \le \varepsilon_{i+1}$, then we do not know what the value of $s_i$ will be. However, this is true for at most two consecutive  values of $i$. With these properties of the signs $s_i$ in mind, let us study the sequence $s_0, \ldots, s_9$. There are two possibilities. The first one is that $s_0, s_1, s_2, s_3$ are all positive, in which case $\varepsilon_2 < \varepsilon^{*}$ and we can recurse on $(\varepsilon_2, \varepsilon_{10})$. The second possibility is that $s_6, s_7, s_8, s_9$ are all negative, in which case $\varepsilon^{*} < \varepsilon_8$ and we can recurse on $(\varepsilon_0, \varepsilon_8)$. In any of these two cases, we can shrink the interval $(\varepsilon', \varepsilon'')$ to at most $9/10$ of its original length. Again using binary search it is possible to find a small open subinterval containing the zero of $f$. Let us now assume that $A < 0$. In this case, one can analogously prove that the sign $s_i$ is positive when $\varepsilon_i < \varepsilon_{i+1} < \varepsilon^{*}$ and negative when $\varepsilon^{*} < \varepsilon_i < \varepsilon_{i+1}$, so the same procedure allows us to shrink $(\varepsilon', \varepsilon'')$.

Let $q$ and $\gamma$ be real algebraic numbers with $q \not \in \{0,1\}$ and $\gamma > 0$. Let $H$ be a graph and let $s$ and $t$ be two distinct connected vertices of $H$. We are going to apply these interval stretching techniques to the linear function
\begin{equation} \label{eq:f}
    f(\varepsilon; H, \gamma) = Z_{s|t}(H; q, \gamma)\left( 1 - \frac{1}{q} \right) + \varepsilon \left( Z_{st}(H; q, \gamma) + \frac{1}{q} Z_{s|t}(H; q, \gamma) \right).
  \end{equation}
  Let us write this function as $f(\varepsilon; H, \gamma) = B(H, \gamma) - \varepsilon A(H, \gamma)$, where $B(H, \gamma) = Z_{s|t}(H; q, \gamma) ( 1 - 1 / q )$ and $A(H, \gamma) =  -Z_{st}(H; q, \gamma) - q^{-1} Z_{s|t}(H; q, \gamma)$. We have
  \begin{equation}
    \label{eq:f:endpoints}
    \begin{aligned}
      f(0; H, \gamma) & = Z_{s|t}(H; q, \gamma)\left( 1 - \frac{1}{q} \right); \\
      f(1-q; H, \gamma) & =  (1-q)Z_{st}(H; q, \gamma).
    \end{aligned}
  \end{equation}
  Under certain hypotheses, we are going to prove that $f(0; H, \gamma) f(1-q; H, \gamma) < 0$, so $A(H, \gamma) \ne 0$ and $f(-;H, \gamma)$ has a zero between $0$ and $1-q$. This allows us to find a suitable interval where we can perform interval-shrinking. For this purpose we will also need Lemma~\ref{lem:f:subinterval}, that tells us that the zero of $f(-; H, \gamma)$ is not close to either $0$ or $1-q$. 

  \begin{lemma} \label{lem:f:subinterval}
    Let $q$ and $\gamma$ be real algebraic numbers with $q \not \in \{0,1\}$ and $\gamma > 0$. Let $H = (V, E)$ be a graph and let $s$ and $t$ be two distinct connected vertices of $H$. Let $n = |V|$, $m = |E|$, $r = \max \{n, m\}$ and $c = 2 \max \{|q|, 1/|q|\}\max \{\gamma, 1/\gamma \}$. Let $\varepsilon^{*}$ be the zero of the function $f(\varepsilon; H, \gamma) = B(H, \gamma) - A(H, \gamma)$, defined as in \eqref{eq:f}. Let us assume that $|Z_{st}(H; q, \gamma)| \ge c^{-r}$,  $|Z_{s|t}(H; q, \gamma)| \ge c^{-r}$ and $A(H, \gamma) \ne 0$.  Then we have $\left| 1-q - \varepsilon^{*} \right| \ge \left| 1-q \right| c^{-2r}$ and $\left| \varepsilon^{*} \right| \ge \left| 1-1/q \right| c^{-2r}$.
  \end{lemma}
  
  \begin{proof}
     In view of the definition of $f(\varepsilon; H, \gamma)$ and equation \eqref{eq:f:endpoints}, we have
     \begin{equation*}
       \left|1-q - \varepsilon^{*}\right| = \frac{\left|f \left( \varepsilon^{*}; H, \gamma \right) - f(1-q)\right|}{|A(H, \gamma)|} = \frac{\left|1-q\right| \left|Z_{st}\left( H; q, \gamma \right)\right|}{|A(H, \gamma)|}. 
    \end{equation*}
    Note that 
    \begin{equation} \label{eq:bound-A}
      \left|A(H, \gamma)\right| \le \sum_{A \subseteq E} \max \{|q|, 1/|q|\} \left| q \right|^{k(A)-1} \left| \gamma \right|^{|A|} \le c^r.
    \end{equation}
    Moreover, we have $|Z_{st}( H; q, \gamma )| \ge c^{-r}$ by hypothesis, so we conclude that $|1 - q - \varepsilon^{*}| \ge |1-q| c^{-2r}$.  Analogously, we find that
     \begin{equation*}
       \left|\varepsilon^{*}\right| = \frac{\left|f \left( \varepsilon^{*}; H, \gamma \right) - f(0)\right|}{|A(H, \gamma)|} = \frac{\left|1-1/q\right| \left|Z_{s|t}\left( H; q, \gamma \right)\right|}{|A(H, \gamma)|} \ge \left| 1 - \frac{1}{q} \right| c^{-2r}. \qedhere 
    \end{equation*}    
  \end{proof}

  \begin{lemma} \label{lem:compute-fraction:q>1} 
    Let $K$ be a real number with $K > 1$. Let $q$, $\gamma_1$ and $\gamma_2$ be real algebraic numbers such that $q > 1$, $\gamma_1 \in (-2, -1)$ and $\gamma_2 > 0$. Let us assume that we have access to an oracle for $\FNPTz(q, \gamma_1, \gamma_2)$. Then there exists an algorithm that takes as input a positive integer $\rho$ and a planar graph $H$ along with two distinct connected vertices $s$ and $t$ of $H$, and, for $\gamma = (\gamma_2+1)^\rho-1$, this algorithm computes a representation of the algebraic number $Z_{s|t}(H; q, \gamma) / Z_{st}(H; q, \gamma)$ in polynomial time in $\rho$ and the size of $H$. Moreover, if we have access to the more powerful oracle $\FNTz(q, \gamma_1, \gamma_2)$, then we can remove the constraint that $H$ is planar.
  \end{lemma}

\begin{proof}
   Since $\FNPTz(q, z)$ is equally hard for any $K > 1$ (see Section~\ref{sec:intro:main-results}), we may assume that $K = 1 + \eta$ for $\eta = 1/41$. 
  
  Let $\rho$, $H = (V, E)$ and $s, t$ be the inputs of our algorithm. Let $n = |V|$ and $m = |E|$. Let $c = 2 \max \{|q|, 1/|q|\}\max \{\gamma, 1/\gamma \}$, so $c \ge 2$, and let $r = \max \{n, m\}$. Let $H'$ be a copy of $H$ with an extra edge from $s$ to $t$. Let $\gamma'$ be a weight that we can implement and let $\varepsilon = \gamma'+1$, so the pair $(1 + q/(\varepsilon -1), \varepsilon)$ is $(q, \gamma')$ written in $(x,y)$ coordinates. We will choose $\gamma'$ and argue how we can implement $\gamma'$ later in the proof. When we say we implement $\varepsilon$, we mean that we implement the pair $(1 + q/(\varepsilon -1), \varepsilon)$ in $(x,y)$ notation or, equivalently, $\gamma'$.
  
  Let $\mathbf{\tau}$ be the weight function on $H'$ that assigns the weight $\gamma$ to the edges of $H$ and the weight $\gamma'$ to the new edge. Then, as was observed in~\cite[Lemma 2]{Goldberg2014}, we have 
  \begin{equation} \label{eq:reduction:1}
    \begin{aligned}
    Z_{\text{Tutte}}(H'; q, \mathbf{\tau}) & = Z_{st}(H; q, \gamma)(1+ \gamma') + Z_{s|t}(H; q, \gamma) \left( 1 + \frac{\gamma'}{q} \right) \\
                                        & = Z_{s|t}(H; q, \gamma)\left( 1 - \frac{1}{q} \right) + \varepsilon \left( Z_{st}(H; q, \gamma) + \frac{1}{q} Z_{s|t}(H; q, \gamma) \right) \\ & = f(\varepsilon; H, \gamma), 
    \end{aligned}
  \end{equation}
  where $f(\varepsilon; H, \gamma)$ was introduced in \eqref{eq:f}. Hence, $Z_{\text{Tutte}}(H'; q, \tau)$ can be seen as a function, with variable $\varepsilon$, of the form $f(\varepsilon; H, \gamma) = B(H, \gamma) - \varepsilon A(H, \gamma)$, where $B(H, \gamma) = Z_{s|t}(H; q, \gamma) ( 1 - 1 / q )$ and $A(H, \gamma) =  -Z_{st}(H; q, \gamma) - q^{-1} Z_{s|t}(H; q, \gamma)$. This construction will be used several times in this section. Now we analise $f(-; H, \gamma)$ for our particular setting ($q > 1$). Since $q$ and $\gamma$ are positive, the quantities $Z_{st}(H; q, \gamma)$ and $Z_{s|t}(H; q, \gamma)$ are positive, so $A(H, \gamma)$ is negative. From $q > 1$ and \eqref{eq:f:endpoints}, it follows that $f(0; H, \gamma) = B(H, \gamma) > 0$ and $f(1-q; H, \gamma) < 0$, so $f(0; H, \gamma)f(1-q; H, \gamma) < 0$ as we wanted. We conclude that the zero $\varepsilon^{*}$ of $f(\varepsilon; H, \gamma)$ is in $(1-q, 0)$. Note that $\varepsilon \in (1-q, 0)$ if and only if $\gamma' \in (-q, -1)$. Moreover, we have
  \begin{equation} \label{eq:bounds-zst:q>1}
    \begin{aligned}
      Z_{st}(H; q, \gamma) & \ge q \gamma^m \ge c^{-r}, \\
      Z_{s|t}(H; q, \gamma) & \ge q^n \ge c^{-r}.
    \end{aligned}
  \end{equation}
  This allow us to apply Lemma~\ref{lem:f:subinterval}. Once we have all these properties of $f(\varepsilon; H, \gamma)$ at our disposal, we can proceed to describe our algorithm. Our algorithm also works for $q \in (- \infty, 0) \cap (0,1)$ as long as $f(0; H, \gamma)f(1-q; H, \gamma) < 0$ and the hypotheses of Lemma~\ref{lem:f:subinterval} hold. In the rest of the proof we will only use the fact that $q > 1$ one more time, but this will be made explicit and can easily be adapted to the case $q < 1$ as we will explain in Lemma~\ref{lem:compute-fraction:0<q<1}.

  Our algorithm computes a positive integer $j_0$ such that $c^{-j_0} \le |q-1|/2$. Let $j$ be an integer with $j \ge j_0$. We will first show how to additively approximate $Z_{s|t}( H; q, \gamma) / Z_{st}( H; q, \gamma)$ with error at most $2 |q| c^{-j} / |q-1|$. 

  If we could efficiently implement the point $(1-q/(\varepsilon-1), \varepsilon)$ (in $(x,y)$ coordinates) for any $\varepsilon \in [1-q,0]$ using only planar graphs, then our algorithm could perform the interval-shrinking technique explained at the beginning of this section. This would allow us to compute an interval of length at most $c^{-j-4r}$ where the linear function $f(\varepsilon)$ has a zero, which would, in turn, provide us with the desired additive approximation, as we will see later. However, some difficulties arise since we do not know how to implement any specific real algebraic weight. This difficulty was overcome by Goldberg and Jerrum by developing Lemmas~\ref{lem:implement-approx:literature:1} and~\ref{lem:implement-approx:literature:2}. Here we use the version of these lemmas given in  Corollary~\ref{cor:implement-approx:literature}. Let $y_1 = \gamma_1 + 1$, $x_1 = 1 + q / (y_1 - 1)$, $y_2 = \gamma_2 +1$ and $x_2 = 1 + q / (y_2 - 1)$. Note that $y_1 \in (-1,0)$, $y_2 > 1$ and $q \ne 0$. Hence, Corollary~\ref{cor:implement-approx:literature} allows us to efficiently implement approximations of real algebraic numbers when applied with the parameters $x_1, y_1, x_2, y_2$. Every time we invoke Corollary~\ref{cor:implement-approx:literature} we will be using these parameters. We are going to use this corollary to implement approximations of $\varepsilon \in (1-q, 0)$. This is the only point where our algorithm uses the fact that $\gamma_1 \in (-2,-1)$ or, equivalently, $y_1 \in (-1,0)$. In further lemmas where we study the case  $q < 1$, we will have to implement approximations of $\varepsilon \in (0, 1-q)$ and, hence, we will get away with the weaker hypothesis $\gamma_1 \in (-1,0)$, or, equivalently, $y_1 \in (0,1)$. (This hypothesis is ``weaker'' in the sense that a $2$-thickening of a $y_1 \in (-1,0)$ implements a $y_1 \in (0,1)$.)

  We want to implement numbers $\varepsilon'$ and $\varepsilon''$ so that $\varepsilon^{*} \in (\varepsilon', \varepsilon'') \subseteq (1-q, 0)$. Note that here we are using that $q > 1$. When $q< 1$ our algorithm would work on the interval $(0, 1-q)$ instead of $(1-q, 0)$. This paragraph is the last time that we use the hypothesis $q > 1$ in this proof. The argument given in this paragraph will be revisited when we deal with the case $q < 1$ in further lemmas. Our algorithm first applies the algorithm given in Corollary~\ref{cor:implement-approx:literature} with $y' = - ( 1-1/q ) c^{-2r}/2$, $k$ such that $|y_1|^k < |y'| < |y_1|^{-k}$ and $n = \lceil 2r \log_2(c) - \log_2(1-1/q) + 2 \rceil$.  Note that $k = O(r)$ and $n = O(r)$. This procedure computes a theta graph and a weight function taking weights in $\{\gamma_1, \gamma_2\}$ that implement a point $(1 + q/(\varepsilon''-1), \varepsilon'')$ such that $|y' - \varepsilon'' | \le 2^{-n} \le (1-1/q)c^{-2r}/4$ in polynomial time in $r = O ( \mathrm{size}(H) )$. We have $ -3(1-1/q)c^{-2r}/4 \le \varepsilon'' \le  -(1-1/q)c^{-2r}/4$, so, by Lemma~\ref{lem:f:subinterval}, we find that $\varepsilon^{*} < \varepsilon'' < 0$. Now our algorithm invokes again Corollary~\ref{cor:implement-approx:literature}, this time with inputs $y' = 1-q + (q-1)c^{-2r}/2$,  $k$ such that $|y_1|^k < |y'| < |y_1|^{-k}$ and $n = \lceil 2r \log_2(c) - \min \{0, \log_2(q-1)\} + 2 \rceil$. This implements $(1 + q/(\varepsilon'-1), \varepsilon')$ with $|y' - \varepsilon' | \le (q-1)c^{-2r}/4$, which gives $1-q+(q-1)c^{-2r}/4 \le \varepsilon' \le  1-q+3(q-1)c^{-2r}/4$. Again by Lemma~\ref{lem:f:subinterval}, we find that $1-q < \varepsilon' < \varepsilon^{*}$.  The interval $(\varepsilon', \varepsilon'')$ is the starting interval for the interval-shrinking procedure.

  Let us assume that we are carrying out the interval-shrinking technique explained at the beginning of this section, so we have an interval $(\varepsilon', \varepsilon'')$ of length $l$ where $f$ changes sign. Let us also assume that we can implement the endpoints $\varepsilon'$ and $\varepsilon''$. We want to find a subinterval of length at most $9l/10$ where $f$ changes sign. We can assume that $l > c^{-j-4r}$, since otherwise we do not need to shrink the interval further. Let $p = 10$ be the number of subintervals into which $(\varepsilon', \varepsilon'')$ is partitioned by the interval-shrinking technique. We want to find numbers $\varepsilon_1, \ldots, \varepsilon_{p-1}$ such that we can implement the point $(1 +q/(\varepsilon_i -1), \varepsilon_i)$ for every $i \in \{1, \ldots, p-1\}$ and, for $\varepsilon_0=\varepsilon'$ and $\varepsilon_p = \varepsilon''$, we have $\varepsilon_i - \varepsilon_{i-1} \ge l / 2 p$ for every $i \in \{1, \ldots, p\}$, which is what is required to perform interval-shrinking. For each $i \in \{1, \ldots, p-1\}$, our algorithm computes $\varepsilon_i' = \varepsilon' + i l / p$ and then it applies the algorithm given in Corollary~\ref{cor:implement-approx:literature} with $y'= \varepsilon_i'$, $k$ such that $|y_1|^k < |y'| < |y_1|^{-k}$ and $n = \lceil (j + 4 r) \log_2(c) + \log_2(4 p) \rceil$. This procedure computes a graph and a weight function taking weights in $\{\gamma_1, \gamma_2\}$ that implement a point $(1 + q/(\varepsilon_i -1), \varepsilon_i)$ such that $|\varepsilon_i' - \varepsilon_i| \le 2^{-n} \le c^{-j - 4 r} /(4 p)$. This application of the procedure given in Corollary~\ref{cor:implement-approx:literature} takes polynomial time in $j$, $r$ and $k$. Note that $k$ is polynomial in $r$ and $j$ because $|1-q| \ge |\varepsilon'_i| \ge l/p \ge c^{-j - 4 r}/ p$ for any $i \in \{1, \ldots, p-1\}$. The algebraic numbers $\varepsilon', \varepsilon_1, \ldots, \varepsilon_{p-1}, \varepsilon''$ form a partition the interval $(\varepsilon', \varepsilon'')$. Our algorithm has computed theta (and, thus, planar) graphs that implement $(1 + q/(\varepsilon_i -1), \varepsilon_i)$, so it can use the oracle $\FNPTz(q, \gamma_1, \gamma_2)$ to multiplicatively approximate $f(\varepsilon_i)$  for every $i \in \{0, \ldots, p\}$. Note that
  \begin{equation*}
    \varepsilon_i - \varepsilon_{i-1} \ge \varepsilon_i' -  \varepsilon_{i-1}' - c^{-j - 4 r} \frac{1}{2p} \ge \frac{l}{2p}
  \end{equation*}
  for every $i \in \{1, \ldots, p\}$. Therefore, our algorithm can apply the interval-shrinking technique discussed at the beginning of this section to shrink $(\varepsilon', \varepsilon'')$.

  To guarantee that this interval-shrinking technique computes an interval of length at most $c^{-j - 4 r}$, it suffices to subdivide the original interval $\lceil (j + 4r) \log_{10/9} (c) + \log_{10/9}|1-q|\rceil$ times due to the fact that each iteration shrinks the interval to $9/10$ of its size. In~\cite{Goldberg2014} and~\cite{Goldberg2017} the authors used the information provided by this interval-shrinking procedure to solve the problem $\# \textsc{Minimum Cardinality } (s,t)\textsc{-Cut}$ for arbitrary graphs (not-necessarily planar). Here we follow a different approach that allows us to compute the representation of $Z_{s|t}( H; q, \gamma )/ Z_{st}( H; q, \gamma )$.

   Once our algorithm has computed an interval of length at most $c^{-j - 4 r}$ where $f$ has a zero, it implements a point $(1+q/(\hat{\varepsilon}-1), \hat{\varepsilon})$ such that $\hat{\varepsilon}$ is in this interval. This can be done by applying Corollary~\ref{cor:implement-approx:literature} with the same parameters as before other than $y'$, which is set as the middle point of the computed interval. Let $\varepsilon^{*}$ be the zero of $f$. Note that $|\hat{\varepsilon}-\varepsilon^{*}| \le c^{-j - 4 r}$. Recall that $f(\varepsilon; H, \gamma) = B(H, \gamma) - A(H, \gamma) \varepsilon$. For a graph $H'$ and a weight function $\tau$ as in \eqref{eq:reduction:1}, with $\gamma' = \hat{\varepsilon}-1$ (which we can now implement as promised before \eqref{eq:reduction:1}), we obtain
  \begin{align}  \label{eq:exact-reduction:2}
    \left| Z_{Tutte}\left( H'; q,\tau \right) \right| & = \left| f(\hat{\varepsilon}) \right| =  \left| f(\hat{\varepsilon}) - f(\varepsilon^{*}) \right| \le |A(H, \gamma)| c^{-j - 4 r} \le c^{-j- 3r},
  \end{align}
  where we used the elementary bound $|A(H, \gamma)| \le c^r$, which has been established in \eqref{eq:bound-A}. By dividing by $Z_{st}\left( H; q, \gamma \right)$ in \eqref{eq:reduction:1}, which is non-zero, and rearranging the terms we find that
  \begin{equation*} 
    \frac{Z_{\text{Tutte}}\left( H'; q, \mathbf{\tau} \right)}{Z_{st}\left( H; q, \gamma \right)} =  \hat{\varepsilon} + \left( 1 + \frac{\hat{\varepsilon}-1}{q} \right) \frac{Z_{s|t}\left( H; q, \gamma \right)}{Z_{st}\left( H; q, \gamma \right)}.
  \end{equation*}
  Dividing by $1 + (\hat{\varepsilon}-1)/q = (q - 1 + \hat{\varepsilon}) / q$ yields
  \begin{equation} \label{eq:exact-reduction:3}
    \frac{q Z_{\text{Tutte}}\left( H'; q, \mathbf{\tau} \right)}{(q-1+\hat{\varepsilon}) Z_{st}\left( H; q, \gamma \right)} = -\frac{\hat{\varepsilon} q}{1-q-\hat{\varepsilon}} + \frac{Z_{s|t}\left( H; q, \gamma \right)}{Z_{st}\left( H; q, \gamma \right)}.
  \end{equation}
  We claim that $|1-q-\hat{\varepsilon}| \ge |1-q| c^{-2r} /2$. Recall that in view of Lemma~\ref{lem:f:subinterval}, we have $|1-q - \varepsilon^{*}| \ge |1-q| c^{-2r}$. Hence, we obtain
  \begin{equation*}
    \left|1-q-\hat{\varepsilon}\right| \ge \left|1 - q - \varepsilon^{*}\right| - \left|\varepsilon^{*} - \hat{\varepsilon} \right| \ge \left|1-q\right| c^{-2r} - c^{-j-4r} \ge \frac{\left|1-q\right|}{2} c^{-2r},
  \end{equation*}
  where we used that $c^{-j-4r} \le c^{-j_0} c^{-4r} \le \left|q-1\right| c^{-4r} / 2$ by definition of $j_0$. Therefore, we can apply this lower bound in conjunction with \eqref{eq:bounds-zst:q>1}, \eqref{eq:exact-reduction:2} and \eqref{eq:exact-reduction:3} to conclude that
  \begin{align*}
    \left| \frac{Z_{s|t}\left( H; q, \gamma \right)}{Z_{st}\left( H; q, \gamma \right)} -  \frac{\hat{\varepsilon} q }{1-q - \hat{\varepsilon}} \right| \le \frac{2|q| \left|Z_{\text{Tutte}}\left( H'; q, \mathbf{\tau} \right)\right|}{\left| 1-q\right|} c^{3r} \le \frac{2|q|}{\left|1-q\right|} c^{-j}.
   \end{align*}
 Our algorithm then computes $\hat{\varepsilon} q  /(1-q - \hat{\varepsilon})$ as an approximation of $\alpha = Z_{s|t}\left( H; q, \gamma \right) / Z_{st}\left( H; q, \gamma \right)$. We have shown that $\alpha$ is a real algebraic number that we can additively approximate up to an error at most $2|q| c^{-j} / |1-q|$ in polynomial time in $j$ and the size of $H$. Technically, our approximation $\hat{\varepsilon} q /(1-q - \hat{\varepsilon})$ is another algebraic number. For this reason, our algorithm approximates $\hat{\varepsilon} q /(1-q - \hat{\varepsilon})$ by a rational number $\overline{\alpha}$ (with additive error at most  $2|q| c^{-j} / \left|q-1\right|$) and uses this rational number as our approximation of $\alpha$. The overall error that we make is then $ |\alpha - \overline{\alpha}| \le 4|q| c^{-j} / \left|1-q\right|$.

 In view of Corollary~\ref{cor:degree-height}, we have $d(\alpha) \le d(q) d(\gamma) \le d(q) d(\gamma_2)$, where we have used that $\gamma \in \mathbb{Q}(\gamma_2)$ and, thus, $d(\gamma) \le d(\gamma_2)$. Moreover, Corollary~\ref{cor:degree-height} yields 
 \begin{equation*}
   H(\alpha) \le \left(2^{m+1/2} e^{n h(q) + m h(\gamma)} \right)^{2d(q) d(\gamma)}.
 \end{equation*}
 Since $h(\gamma) = h((\gamma_2 - 1)^\rho -1) \le \rho (1 + h(\gamma_2))$ by Lemma~\ref{lem:height:upper-bound}, our algorithm can compute a rational number $D_{q, \gamma_2}$ with $D_{q,\gamma_2} > 1$ such that $H(\alpha) \le D_{q, \gamma_2}^{\rho\, \mathrm{size}(H)}$. The only non-trivial step of this computation is upper bounding $h(q)$ and $h(\gamma_2)$ in terms of the degrees and usual heights of $q$ and $\gamma_2$ as in \eqref{eq:mahler-inequality}. Let $d = d(q) d(\gamma_2) = O(1)$ and $U = D_{q, \gamma_2}^{\rho\, \mathrm{size}(H)}$. Let $\nbits$ be as in Lemma~\ref{lem:lll}. Then we have $2^\nbits = O(D_{q, \gamma_2}^{2d\, \rho \, \mathrm{size}(H)})$, so $\nbits = O(\rho\, \mathrm{size}(H))$. By choosing $j$ appropriately, we can use the algorithm that we have developed in this proof to find a rational approximation $\overline{\alpha}$ with $\left| \alpha - \overline{\alpha} \right| \le 2^{-\nbits}/ (12 d)$. As we have argued, this takes polynomial time in $\nbits$ and $\mathrm{size}(H)$. Since $\nbits = O (\rho \, \mathrm{size}(H))$, we conclude that the computation of $\overline{\alpha}$ runs in polynomial time in $\rho$ and $\mathrm{size}(H)$. Once we have computed this approximation, our algorithm invokes the algorithm given in Lemma~\ref{lem:lll} to determine the minimal polynomial of $\alpha$ in time $O(d^5 (d + \log U ) ) = O( \rho \, \mathrm{size}(H))$. Finally, it remains to compute an interval of the real line where $\alpha$ is the only root of its minimal polynomial. Since $\alpha$ is a real algebraic number and we know its minimal polynomial, our algorithm can use Sturm sequences to isolate the real roots of this minimal polynomial. Then, by approximating $\alpha$ it decides which one of the computed intervals corresponds to $\alpha$.

  Finally, note that our algorithm also works for arbitrary graphs (not-necessarily planar) as long as our oracle provides us with reliable answers for any graph.
\end{proof}

  \begin{lemma} \label{lem:compute-fraction:q>1:sign} 
    Let $q$, $\gamma_1$ and $\gamma_2$ be real algebraic numbers such that $q > 1$, $\gamma_1 \in (-2, -1)$ and $\gamma_2 > 0$. Let us assume that we have access to an oracle for the computational problem $\SPTz(q, \gamma_1, \gamma_2)$. Then there exists an algorithm that takes as input a positive integer $\rho$ and a planar graph $H$ along with two distinct connected vertices $s$ and $t$ of $H$, and, for $\gamma = (\gamma_2+1)^\rho-1$, this algorithm computes a representation of the algebraic number $Z_{s|t}(H; q, \gamma) / Z_{st}(H; q, \gamma)$ in polynomial time in $\rho$ and the size of $H$. Moreover, if we have access to the more powerful oracle $\STz(q, \gamma_1, \gamma_2)$, then we can remove the constraint that $H$ is planar.
  \end{lemma}
  \begin{proof}
    The algorithm is exactly the same one of Lemma~\ref{lem:compute-fraction:q>1}. The proof is analogous too. The only difference is in the interval-shrinking technique, where we split $(\varepsilon',\varepsilon'')$ into $4$ intervals instead of $10$ (so $p = 4$ in the proof), but this has been discussed at the beginning of this section.
  \end{proof}

\begin{lemma} \label{lem:compute-fraction:0<q<1} 
  Let $K$ be a real number with $K > 1$. Let $q$, $\gamma_1$ and $\gamma_2$ be real algebraic numbers such that $0 < q < 1$, $\gamma_1 \in (-1, 0)$ and $\gamma_2 > 0$. Let us assume that we have access to an oracle for $\FNPTz(q, \gamma_1, \gamma_2)$. Then there exists an algorithm that takes as input a positive integer $\rho$ and a planar graph $H$ along with two distinct connected vertices $s$ and $t$ of $H$, and, for $\gamma = (\gamma_2+1)^\rho-1$, this algorithm computes a representation of the algebraic number $Z_{s|t}(H; q, \gamma) / Z_{st}(H; q, \gamma)$ in polynomial time in $\rho$ and the size of $H$. Moreover, if we have access to the more powerful oracle $\FNTz(q, \gamma_1, \gamma_2)$, then we can remove the constraint that $H$ is planar.
\end{lemma}
\begin{proof}
  We claim that the algorithm presented in Lemma~\ref{lem:compute-fraction:q>1} also works in this setting. Let $f(\varepsilon; H, \gamma) = B(H,\gamma)  - \varepsilon A(H, \gamma)$ as in \eqref{eq:f}. As we pointed out in the proof of Lemma~\ref{lem:compute-fraction:q>1}, the algorithm works as long as $f(0; H, \gamma) f(1-q; H, \gamma) < 0$ and the hypothesis of Lemma~\ref{lem:f:subinterval} hold. First, since  $q$ and $\gamma$ are positive, equations \eqref{eq:bounds-zst:q>1} hold. It follows that $A(H, \gamma) = -Z_{st}(H; q, \gamma) - q^{-1}Z_{st}(H; q, \gamma) \ne 0$. Hence, the hypothesis of Lemma~\ref{lem:f:subinterval} hold. In view of \eqref{eq:f:endpoints} and the fact that $q \in (0,1)$ and $\gamma$ is positive, we have $f(0; H, \gamma) < 0$ and $f(1-q; H, \gamma) > 0$. We conclude that $f(0; H, \gamma) f(1-q; H, \gamma) < 0$, as we wanted.

  This time the interval-stretching technique applied in Lemma~\ref{lem:compute-fraction:q>1} runs on a subinterval $(\varepsilon', \varepsilon'')$ of $(0, 1-q)$, so we only need to implement positive values of $\varepsilon$. For this reason, we can get away with the hypothesis $\gamma_1 \in (-1,0)$ instead of the hypothesis $\gamma_1 \in (-2,-1)$, as was announced in the proof of  Lemma~\ref{lem:compute-fraction:q>1}. Finally, we must indicate how our algorithm implements the numbers $\varepsilon'$ and $\varepsilon''$ so that $\varepsilon^{*} \in (\varepsilon', \varepsilon'') \subseteq (0, 1-q)$, as this was only done in Lemma \ref{lem:compute-fraction:q>1} for $q > 1$. The argument that we give here also applies when $q < 0$. Let $y_1 = \gamma_1 + 1$, $x_1 = 1 + q / (y_1 - 1)$, $y_2 = \gamma_2 +1$ and $x_2 = 1 + q / (y_2 - 1)$. We have $y_1 \in (0,1)$, $y_2 > 1$, $q < 1$ and $q \ne 0$. Our algorithm first applies the algorithm given in Corollary~\ref{cor:implement-approx:literature} with $y' = \left| 1-1/q \right| c^{-2r}/2$, $k$ such that $|y_1|^k < |y'| < |y_1|^{-k}$ and $n = \lceil 2r \log_2(c) - \min\{0,\log_2|1-1/q|\} + 2 \rceil$.  Note that $k = O(r)$ and $n = O(r)$. This procedure computes a theta graph and a weight function taking weights in $\{\gamma_1, \gamma_2\}$ that implement a point $(1 + q/(\varepsilon'-1), \varepsilon')$ such that $|y' - \varepsilon'| \le 2^{-n} \le |1-1/q|c^{-2r}/4$ in polynomial time in $r = O ( \mathrm{size}(H) )$. We obtain $ |1-1/q|c^{-2r}/4 \le \varepsilon' \le  3|1-1/q|c^{-2r}/4$, so, by Lemma~\ref{lem:f:subinterval}, we find that $0 < \varepsilon' < \varepsilon^{*}$. Next our algorithm invokes again Corollary~\ref{cor:implement-approx:literature}, this time with inputs $y' = 1-q - (1-q)c^{-2r}/2$, $k$ such that $|y_1|^k < |y'| < |y_1|^{-k}$ and $n = \lceil 2r \log_2(c) - \min \{0,\log_2(1-q)\} + 2 \rceil$. This implements $(1 + q/(\varepsilon''-1), \varepsilon'')$ with $|y' - \varepsilon'' | \le (1-q)c^{-2r}/4$, which gives $1-q-3(1-q)c^{-2r}/4 \le \varepsilon'' \le  1-q+(1-q)c^{-2r}/4$. Again by Lemma~\ref{lem:f:subinterval}, we find that $\varepsilon^{*} < \varepsilon'' < 1-q$. The interval $(\varepsilon', \varepsilon'')$ is the starting interval for the interval-shrinking procedure that we needed.
\end{proof}

  \begin{lemma} \label{lem:compute-fraction:0<q<1:sign} 
    Let $q$, $\gamma_1$ and $\gamma_2$ be real algebraic numbers such that $0 < q < 1$, $\gamma_1 \in (-1, -0)$ and $\gamma_2 > 0$. Let us assume that we have access to an oracle for the computational problem $\SPTz(q, \gamma_1, \gamma_2)$. Then there exists an algorithm that takes as input a positive integer $\rho$ and a planar graph $H$ along with two distinct connected vertices $s$ and $t$ of $H$, and, for $\gamma = (\gamma_2+1)^\rho-1$, this algorithm computes a representation of the algebraic number $Z_{s|t}(H; q, \gamma) / Z_{st}(H; q, \gamma)$ in polynomial time in $\rho$ and the size of $H$. Moreover, if we have access to the more powerful oracle $\STz(q, \gamma_1, \gamma_2)$, then we can remove the constraint that $H$ is planar.
  \end{lemma}
  
  \begin{proof}
    The algorithm is exactly the same the one of Lemma~\ref{lem:compute-fraction:0<q<1}, the only difference being in the interval-stretching technique as we have already explained.
  \end{proof}

\begin{lemma} \label{lem:compute-fraction:q<0} 
  Let $K$ be a real number with $K > 1$. Let $q$, $\gamma_1$ and $\gamma_2$ be real algebraic numbers such that $q < 0$, $\gamma_1 \in (-1, 0)$ and $\gamma_2 > 0$. Let us assume that we have access to an oracle for $\FNPTz(q, \gamma_1, \gamma_2)$. Then there exists an algorithm that takes as input:  
  \begin{itemize}
  \item a positive integer $\rho$ ;
  \item a planar graph $H = (V, E)$ such that, for $\gamma = (\gamma_2+1)^\rho-1$, we have $\gamma \ge (8 \max \{|q|, 1/|q|\})^r$, where $r=\max \{|V|, |E|\}$;
  \item two distinct connected vertices $s$ and $t$ of $H$.
  \end{itemize}
 This algorithm computes a representation of the algebraic number $Z_{s|t}(H; q, \gamma) / Z_{st}(H; q, \gamma)$ in polynomial time in $\rho$ and the size of $H$. Moreover, for such inputs $\rho$, $H$ and $s,t$, we have $Z_{st}(H; q, \gamma) \ne 0$ and $Z_{\text{Tutte}}(H; q, \gamma) \ne 0$. If we have access to the more powerful oracle $\FNTz(q, \gamma_1, \gamma_2)$, then we can remove the constraint that $H$ is planar.
\end{lemma}
\begin{proof}
  We claim that the algorithm presented in Lemmas~\ref{lem:compute-fraction:q>1} and~\ref{lem:compute-fraction:0<q<1} also works in this setting. Let $n = |V|$ and $m = |E|$. Let $c = 2 \max \{|q|,1/|q|\} \gamma$. We may assume that $r\ge2$. First, let us assume that $H$ is connected. Let $f(\varepsilon; H, \rho) = B(H, \gamma) - \varepsilon A(H, \gamma)$ as in \eqref{eq:f}, so $B(H, \gamma) = Z_{s|t}(H; q, \gamma) ( 1 - 1 / q )$ and $A(H, \gamma) = -Z_{st}(H; q, \gamma) - q^{-1} Z_{s|t}(H; q, \gamma)$. Recall that we have to prove that the conditions of Lemma~\ref{lem:f:subinterval} hold, as well as the inequality $f(0; H, \gamma) f(1-q; H, \gamma) < 0$. Let $\delta = (2 \max \{|q|, 1/|q|\})^r / \gamma$. Note that $0 < \delta \le 1/4$ because $\gamma \ge (8 \max \{|q|, 1/|q|\})^r$. Each one of the (at most $2^m$) terms in $Z_{st}(H; q, \gamma)$, other than the term with all edges in $A$, has absolute value at most $\gamma^{m-1} |q| \max \{|q|, 1\}^{n-1} \le \delta 2^{-m}\gamma^m |q|$. Since $H$ is connected, the term with all edges in $A$ is $q \gamma^m$. Thus, we have the inequalities
  \begin{equation} \label{eq:bounds-zst}
    \gamma^{m} q - \delta \gamma^m |q| \le Z_{st}(H; q, \gamma) \le  \gamma^m q + \delta \gamma^m |q| < 0.
  \end{equation}
  In particular, $Z_{st}(H; q, \gamma) \ne 0$. It also follows that
  \begin{equation*}
    \left| Z_{st}(H; q, \gamma) \right| \ge \gamma^m |q| (1 - \delta) \ge \gamma^m |q| 3/4 \ge c^{-r},
  \end{equation*}
  which is one of the conditions of  Lemma~\ref{lem:f:subinterval}. Recall that an $(s,t)$-cut of $H$ is a subset $A$ of edges of $H$ such that any path from $s$ to $t$ in $H$ has an edge in $A$. The size of this $(s,t)$-cut is the cardinality of $A$. Let $k$ be the size of a minimum cardinality $(s,t)$-cut in $H$, and let $C$ be the number of $(s,t)$-cuts of size $k$. We study the terms $q^{k(A)}\gamma^{|A|}$ appearing in $Z_{s|t}(H; q, \gamma)$, so $A$ is a subset of $E$ such that $s$ and $t$ are not connected in $(V, A)$. Note that such an $A$ is the complement of an $(s,t)$-cut and, hence, $|A| \le m-k$. Moreover, if $A$ is not the complement of an $(s,t)$-cut of size $k$, then the absolute value of $q^{k(A)}\gamma^{|A|}$ is at most $\gamma^{m-k-1} q^2 \max \{1, |q|\}^{n-2} \le \delta 2^{-m} \gamma^{m-k} q^2$. Thus, we have the inequalities
  \begin{equation} \label{eq:bounds-zs|t}
    0 < C \gamma^{m-k} q^2 - \delta \gamma^{m-k} q^2 \le Z_{s|t}(H; q, \gamma) \le  C\gamma^{m-k} q^2 + \delta \gamma^{m-k} q^2.
  \end{equation}
  The inequalities \eqref{eq:bounds-zst} and \eqref{eq:bounds-zs|t} have been previously given in the proof of~\cite[Lemma 2]{Goldberg2014}. As a consequence, we find that 
  \begin{equation*}
    \left| Z_{s|t}(H; q, \gamma) \right| \ge C \gamma^{m-k} q^2 (1 - \delta) \ge C\gamma^{m-k} q^{2} 3/4 \ge \gamma^{m-k} q^{2} 3/4  \ge c^{-r},
  \end{equation*}
  which is another one of the conditions of Lemma~\ref{lem:f:subinterval}. In view of \eqref{eq:f:endpoints} and the facts that $q < 0$ and we know the signs of $Z_{s|t}(H; q, \gamma)$ and $Z_{st}(H; q, \gamma)$, it follows that  $f(0; H, \gamma) > 0$ and $f(1-q; H, \gamma) < 0$. Hence, we find that $f(0; H, \gamma) f(1-q; H, \gamma) < 0$, as we wanted. Note that $A(H, \gamma)$ has to be non-zero because $f(-;H,\gamma)$ is non-constant as $f(0; H, \gamma) f(1-q; H, \gamma) < 0$. This is the last condition of Lemma~\ref{lem:f:subinterval} that we had to check.  We conclude that we can apply the algorithm given in the proof of Lemma~\ref{lem:compute-fraction:0<q<1} to compute $Z_{s|t}( H; q, \gamma ) / Z_{st}( H; q, \gamma)$ in polynomial time in $\rho$ and the size of $H$. Finally, we show that $Z_{\text{Tutte}}(H; q, \gamma) \ne 0$. This is not needed for the algorithm of Lemma~\ref{lem:compute-fraction:0<q<1}, but is part of the statement of the current lemma. In light of \eqref{eq:bounds-zst} and \eqref{eq:bounds-zs|t}, we have $\left| Z_{st}(H; q, \gamma) \right| \ge \gamma^{m} |q| (1 - \delta)$ and $\left| Z_{s|t}(H; q, \gamma) \right| \le C \gamma^{m-k} q^2(1+\delta)$. Note that
  \begin{align*}
     \gamma^{m} |q| (1-\delta) \ge \frac{3}{4}\gamma^{m} |q|  > \frac{5}{4}C\gamma^{m-k} q^2 \ge \gamma^{m-k} q^2 (1+\delta),
  \end{align*}
  where we used that $\gamma \ge (8 \max \{|q|, 1/|q|\})^r \ge 8 \cdot 2^{m} |q| > 5 C |q|$ since $r \ge 2$.  Therefore, we find that $\left| Z_{st}(H; q, \gamma) \right| > \left| Z_{s|t}(H; q, \gamma) \right|$. We conclude that
  \begin{equation*}
    Z_{\text{Tutte}}(H; q, \gamma) =   Z_{st}(H; q, \gamma) + Z_{s|t}(H; q, \gamma) < 0.
  \end{equation*}
  
  It remains to consider the case where $H$ is not connected. Let $H_1, \ldots, H_l$ be the connected components of $H$, and let us assume that the vertices $s$ and $t$ are in $H_1$ without loss of generality. We have
  \begin{align*}
    Z_{st}\left(H; q, \gamma\right) &=  Z_{st}\left(H_1; q, \gamma\right) Z_{\text{Tutte}}\left(H_2; q, \gamma\right) \cdots Z_{\text{Tutte}}\left(H_l; q, \gamma\right); \\
    Z_{s|t}\left(H; q, \gamma\right) &=  Z_{s|t}\left(H_1; q, \gamma\right) Z_{\text{Tutte}}\left(H_2; q, \gamma\right) \cdots Z_{\text{Tutte}}\left(H_l; q, \gamma\right); \\
    Z_{\text{Tutte}}\left(H; q, \gamma\right) &=  Z_{\text{Tutte}}\left(H_1; q, \gamma\right) Z_{\text{Tutte}}\left(H_2; q, \gamma\right) \cdots Z_{\text{Tutte}}\left(H_l; q, \gamma\right).    
  \end{align*}
  We have already shown that $Z_{st}\left(H_1; q, \gamma\right)$, $Z_{st}\left(H_1; q, \gamma\right)$ and $Z_{\text{Tutte}}\left(H_j; q, \gamma\right)$ are non-zero for all $j$. Hence, we obtain $Z_{s|t}( H; q, \gamma ) / Z_{st}( H; q, \gamma) = Z_{s|t}( H_1; q, \gamma ) / Z_{st}( H_1; q, \gamma)$, and we can apply our algorithm to $H_1$ instead of $H$. Moreover, we have $Z_{st}\left(H; q, \gamma\right) \ne 0$ and $Z_{\text{Tutte}}(H; q, \gamma) \ne 0$ as we wanted. This finishes the proof.
\end{proof}

\begin{lemma}  \label{lem:compute-fraction:q<0:sign} 
  Let $q$, $\gamma_1$ and $\gamma_2$ be real algebraic numbers such that $q < 0$, $\gamma_1 \in (-1, 0)$ and $\gamma_2 > 0$. Let us assume that we have access to an oracle for $\SPTz(q, \gamma_1, \gamma_2)$. Then there exists an algorithm that takes as input:  
  \begin{itemize}
  \item a positive integer $\rho$ ;
  \item a planar graph $H = (V, E)$ such that, for $\gamma = (\gamma_2+1)^\rho-1$, we have $\gamma \ge (8 \max \{|q|, 1/|q|\})^r$, where $r=\max \{|V|, |E|\}$;
  \item two distinct connected vertices $s$ and $t$ of $H$.
  \end{itemize}
 This algorithm computes a representation of the algebraic number $Z_{s|t}(H; q, \gamma) / Z_{st}(H; q, \gamma)$ in polynomial time in $\rho$ and the size of $H$. Moreover, for such inputs $\rho$, $H$ and $s,t$, we have $Z_{st}(H; q, \gamma) \ne 0$ and $Z_{\text{Tutte}}(H; q, \gamma) \ne 0$. If we have access to the more powerful oracle $\STz(q, \gamma_1, \gamma_2)$, then we can remove the constraint that $H$ is planar.
\end{lemma}
\begin{proof}
  The algorithm is the same one as that of Lemma~\ref{lem:compute-fraction:q<0}, the only difference being in the interval-stretching technique, as we have already explained. 
\end{proof}

Now we deal with the last part of our reduction, where we reduce the computation of $Z_{\text{Tutte}}(G; q, \gamma)$ to the computation of $Z_{s|t}(H; q, \gamma ) / Z_{st}( H; q, \gamma )$ on the subgraphs $H$ of $G$. First, let us introduce some notation.
\begin{definition}
We say that a pair $(q, \gamma)$ of algebraic numbers is \emph{zero-free} for a graph $G$ if $q \ne 0$ and, for every subgraph $H$ of $G$ and every pair of distinct vertices $s$ and $t$ in the same connected component of $H$, the quantities $Z_{st}(H, q, \gamma)$ and $Z_{\mathrm{Tutte}}(H, q, \gamma)$ are non-zero.  
\end{definition}
Note that if $(q,\gamma)$ is zero-free for $G$, then $(q, \gamma)$ is also zero-free for any subgraph of $H$. We consider the following computational problems.

\prob{$\RT(q, \gamma)$.}{A (multi)graph $G = (V, E)$ such that $(q,\gamma)$ is zero-free for $G$ and two distinct vertices $s$ and $t$ in the same connected component of $G$.}{A representation of the algebraic number $Z_{s|t}(G; q, \gamma ) / Z_{st}(G; q, \gamma )$.}

\prob{$\ZT(q, \gamma)$.}{A (multi)graph $G = (V, E)$ such that $(q,\gamma)$ is zero-free for $G$.}{A representation of the algebraic number $Z_{\text{Tutte}}(G; q, \gamma )$.}

We also consider the planar versions of these problems, $\RPT(q, \gamma)$ and $\ZPT(q, \gamma)$. Then we can express the last part of our reduction as a reduction between these two computational problems.

\begin{lemma} \label{lem:reduce-frac-to-exact}
  Let $q$ and $\gamma$ be algebraic numbers with $q \ne 0$. Then we have the reductions
  \begin{align*}
    \ZPT(q, \gamma) & \le_T \RPT(q, \gamma), \\
    \ZT(q, \gamma) & \le_T \RT(q, \gamma).
  \end{align*}
\end{lemma}

\begin{proof}
  First, we show $\ZPT(q, \gamma) \le_T \RPT(q, \gamma)$. Let $G$ be the input of $\ZPT(q, \gamma)$. The reduction computes a representation of $Z_{\text{Tutte}}(G; q, \gamma)$ as follows. We assume that $G$ is not a tree since it is known how to compute the Tutte polynomial of a tree in polynomial time~\cite[Example 2.1]{Sokal2005}.  Then we can find an edge $e = (s,t)$ of $G$ that is not a bridge. We are going to use the oracle for $\RPT(q, \gamma)$ to reduce the computation of $Z_{\text{Tutte}}(G; q, \gamma)$ to that of $Z_{\text{Tutte}}(G \setminus e; q, \gamma)$, where $G \setminus e$ is formed from $G$ by deleting $e$. Note that if $G$ is planar, then $G \setminus e$ is also planar. Since $(q, \gamma)$ is zero-free for $G$, we have $Z_{st}( G; q, \gamma ) \ne 0$.  Let $\alpha = Z_{s|t}(G; q, \gamma ) / Z_{st}( G; q, \gamma )$. First, note that
   \begin{equation*}
      Z_{\text{Tutte}}\left(G; q, \gamma\right) =  Z_{st}\left( G; q, \gamma \right) + Z_{s|t}\left( G; q, \gamma \right) = Z_{st}\left( G; q, \gamma \right) \left(1 + \alpha\right).
    \end{equation*}
   By calling the oracle the algorithm obtains a representation of the factor $1 + \alpha$. Since $e$ is not a bridge, $s$ and $t$ are connected in $G \setminus e$, so, by calling the oracle again, the algorithm has access to a representation of the algebraic number $\beta = Z_{s|t}( G \setminus e; q, \gamma ) / Z_{st} ( G \setminus e; q, \gamma )$. We have 
   \begin{align*}
     Z_{st}\left(G; q, \gamma\right) & = Z_{st}\left(G \setminus e; q, \gamma \right) \left(1 + \gamma \right)  + \gamma q^{-1}  Z_{s|t}\left( G \setminus e; q, \gamma \right) \\
                         & = Z_{\text{Tutte}}\left(G \setminus e; q, \gamma \right) \left( \frac{1 + \gamma}{1 + \beta}  +  \gamma q^{-1}\frac{ \beta}{1 + \beta}\right),
   \end{align*}
   where we multiplied and divided by $Z_{\text{Tutte}}(G \setminus e; q, \gamma ) = Z_{st}(G \setminus e; q, \gamma )(1+\beta)$, which is non-zero since $(q, \gamma)$ is zero-free for $G$. Note that the fact that $Z_{\text{Tutte}}\left(G \setminus e; q, \gamma \right) \ne 0$ is equivalent to $\beta \ne -1$. We obtain
   \begin{equation} \label{eq:reducing-exact-to-quotient}
     Z_{\text{Tutte}}\left(G; q, \gamma\right) = Z_{\text{Tutte}}\left(G \setminus e; q, \gamma \right) \left( 1 + \gamma  +  \gamma q^{-1} \beta \right) \frac{1 + \alpha}{1 + \beta}.
   \end{equation}
   The algorithm then computes a representation of $Z_{\text{Tutte}}(G \setminus e; q, \gamma)$ recursively. Note that this reduction also works between the non-planar versions of the problems.
\end{proof}  

In the rest of this section we put our reduction together. There is one result for each one of the cases $q > 1$, $0 < q < 1$ and $q< 0$ (see Lemmas~\ref{lem:exact-reduction:q>1},~\ref{lem:exact-reduction:0<q<1} and~\ref{lem:exact-reduction:q<0}).

\begin{lemma} \label{lem:exact-reduction:q>1}
  Let $K$ be a real number with $K > 1$. Let $q$, $\gamma_1$ and $\gamma_2$ be real algebraic numbers such that $q >1$, $\gamma_1 \in (-2, -1)$ and $\gamma_2 > 0$.  Then we have the following reductions:
  \begin{align*}
    \PT(q, \gamma_2) & \le_T \FNPTz(q, \gamma_1, \gamma_2) \\
    \PT(q, \gamma_2) & \le_T \SPTz(q, \gamma_1, \gamma_2),
  \end{align*}
  where $\le_T$ denotes a Turing reduction. Moreover, these reductions also hold for the analogous non-planar problems.
\end{lemma}
\begin{proof}
  We claim that the problems $\PT(q, \gamma_2)$ and $\ZPT(q, \gamma_2)$ are equivalent. This follows from the fact that $(q, \gamma_2)$ is zero-free for every graph $G$. Lemma~\ref{lem:reduce-frac-to-exact} gives us a reduction from $ \ZPT(q, \gamma_2)$ to $\RPT(q, \gamma_2)$. Recall that we have $q>0$, $\gamma_1 \in (-2,-1)$ and $\gamma_2 > 0$. Thus, we can apply Lemma~\ref{lem:compute-fraction:q>1} with $\rho = 1$ to obtain a reduction from the problem $\RPT(q, \gamma_2)$ to the problem $\FNPTz(q, \gamma_1, \gamma_2)$, which gives the first reduction of the statement. The second reduction is derived analogously, but this time we apply Lemma~\ref{lem:compute-fraction:q>1:sign} instead of Lemma~\ref{lem:compute-fraction:q>1}. Finally, note that our reductions also hold for the non-planar version of the problems since the algorithms given in Lemma~\ref{lem:compute-fraction:q>1} and Lemma~\ref{lem:compute-fraction:q>1:sign} work for arbitrary graphs (non-necessarily planar) as long as the oracle does.
\end{proof}

\begin{lemma} \label{lem:exact-reduction:0<q<1}
  Let $K$ be a real number with $K > 1$.  Let $q$, $\gamma_1$ and $\gamma_2$ be real algebraic numbers such that $0 < q < 1$, $\gamma_1 \in (-1, 0)$ and $\gamma_2 > 0$. Then we have the following reductions:
  \begin{align*}
    \PT(q, \gamma_2) & \le_T \SPTz(q, \gamma_1, \gamma_2), \\
    \PT(q, \gamma_2) & \le_T \FNPTz(q, \gamma_1, \gamma_2).
  \end{align*}
  Moreover, these reductions also hold for the analogous non-planar problems.
\end{lemma}
\begin{proof}
  The proof is analogous to that of Lemma~\ref{lem:exact-reduction:q>1}; now, we instead combine Lemmas~\ref{lem:reduce-frac-to-exact},~\ref{lem:compute-fraction:0<q<1} and~\ref{lem:compute-fraction:0<q<1:sign}.
\end{proof}

So far we have obtained reductions when $q > 1$ or $0 < q < 1$. To obtain a similar result when $q < 0$ we have to introduce the following variant of $\ET(q, \gamma)$, where $q$ is an algebraic number and $\gamma$ is a positive real algebraic number.

\prob{$\TT(q, \gamma)$.}{A (multi)graph $G = (V, E)$.}{A representation of the algebraic number $Z_{\text{Tutte}}(G; q, (\gamma+1)^{\rho(G)}-1)$, where $\rho(G)$ is the smallest positive integer such that $(\gamma+1)^{\rho(G)}-1 > M(G)$ for $M(G) = (8 \max \{ |q|, 1/|q| \})^r$ and $r = \max \{|V|, |E|\}$.}

We also consider the planar version of this problem, $\TPT(q, \gamma)$, where the input graph is promised to be planar. 

\begin{lemma} \label{lem:thickened-tutte}
  Let $q$ be an algebraic number and let $\gamma$ be a real algebraic number with $\gamma > 0$. Then the problem $\TPT(q, \gamma)$ is $\# \mathsf{P}$-hard unless $q \in \{1, 2\}$, and the problem $\TT(q, \gamma)$ is $\# \mathsf{P}$-hard unless $q = 1$.
\end{lemma}
\begin{proof}
  We are going to reduce $\PT(q, 2)$ to $\TPT(q, \gamma)$. The result then follows from the $\# \mathsf{P}$-hardness of $\PT(q, 2)$, cf. Theorem~\ref{thm:hardness:exact-planar}.

  Let $G$ be an $m$-edge instance of $\PT(q, 2)$. For $j=1,\hdots,m$, let $G_j$ be the graph obtained from $G$ by $j$-thickening each of its edges. We have $M(G_j) = (8 \max \{|q|, 1 / |q|\})^{\max\{n, j m\}}$ so $M(G_j)$, and therefore $\rho(G_j)$, are non-decreasing in $j$. Let $\gamma_j =  (\gamma+1)^{j\rho(G_j)}-1$ and note that $Z_{\text{Tutte}}(G_j; q, (\gamma+1)^{\rho(G_j)}-1) = Z_{\text{Tutte}}(G; q, \gamma_j)$.  Note that the points $\gamma_1, \ldots, \gamma_m$ are distinct because $j \rho(G_j) \le j \rho(G_{j+1}) < (j+1)\rho(G_{j+1})$ for every $j$. Moreover, their representation is polynomial in the size of $G$, and hence so is the representation of $Z_{\text{Tutte}}(G; q, \gamma_j)$.

The reduction constructs $G_1, \ldots, G_m$ and computes $Z_{\text{Tutte}}(G; q, \gamma_j)$ using the oracle for $\TPT(q, \gamma)$ with input $G_j$.  By interpolation, we then recover the polynomial $Z_{\text{Tutte}}(G; q, x)$, whose degree is  $m$ when $q$ is viewed as a constant,  in time polynomial in the size of $G$. The reduction is then completed by evaluating $Z_{\text{Tutte}}(G; q, x)$ at $x =2$.

  Finally note that this reduction also works from $\ET(q, 2)$ to $\TT(q, \gamma)$. The only difference is that $\ET(q, 2)$ is also $\# \mathsf{P}$-hard for $q = 2$ (see Theorem~\ref{thm:hardness:exact}), so we also get $\# \mathsf{P}$-hardness in this case.
\end{proof}

We are going to reduce the problem $\TPT(q, \gamma_2)$ to the problem $\FNPTz(q, \gamma_1, \gamma_2)$ for appropriate $\gamma_1$ and $\gamma_2$. In order to do so, we need to adapt Lemma~\ref{lem:reduce-frac-to-exact} to this context. For this purpose, we consider the following computational problems.

\prob{$\RTT(q, \gamma)$.}{A (multi)graph $G = (V, E)$, two distinct connected vertices $s$ and $t$ of $G$, and a positive integer $\rho$ such that, for $\gamma_{\rho} = (\gamma+1)^{\rho}-1$, $(q,\gamma_{\rho})$ is zero-free for $G$ and $\gamma_{\rho} > M(G)$, where $M(G) = (8 \max \{ |q|, 1/|q| \})^r$ and $r = \max \{|V|, |E|\}$.}{A representation of the algebraic number  $Z_{s|t}(G; q, \gamma_{\rho}) / Z_{st}(G; q, \gamma_{\rho} )$.}

\prob{$\ZTT(q, \gamma)$.}{A (multi)graph $G = (V, E)$ and a positive integer $\rho$ such that, for $\gamma_{\rho} = (\gamma+1)^{\rho}-1$, $(q,\gamma_{\rho})$ is zero-free for $G$ and $\gamma_{\rho} > M(G)$, where $M(G) = (8 \max \{ |q|, 1/|q| \})^r$ and $r = \max \{|V|, |E|\}$.}{A representation of the algebraic number $Z_{\text{Tutte}}(G; q, \gamma_{\rho} )$.}

We also consider the planar versions of these problems, $\RTPT(q, \gamma)$ and $\ZTPT(q, \gamma)$.

\begin{lemma} \label{lem:reduce-frac-to-exact:thickened}
  Let $q$ and $\gamma$ be algebraic numbers with $q \ne 0$. Then we have the reductions
  \begin{align*}
    \ZTPT(q, \gamma) & \le_T \RTPT(q, \gamma), \\
    \ZTT(q, \gamma) & \le_T \RTT(q, \gamma).
  \end{align*}
\end{lemma}
\begin{proof}
  The reduction is almost exactly the one explained in Lemma~\ref{lem:reduce-frac-to-exact}. The only difference is that, for an input $(G, \rho)$, each call to the oracle has as parameters a subgraph $H$ of $G$, two vertices $s$ and $t$ determined in the reduction, and the same positive integer $\rho$. 
\end{proof}

\begin{lemma} \label{lem:exact-reduction:q<0}
  Let $K$ be a real number with $K > 1$. Let $q$, $\gamma_1$ and $\gamma_2$ be real algebraic numbers such that $q < 0$, $\gamma_1 \in (-1, 0)$ and $\gamma_2 > 0$. Then we have the following reductions:
  \begin{align*}
    \TPT(q, \gamma_2 ) & \le_T \SPTz(q, \gamma_1, \gamma_2), \\
    \TPT(q, \gamma_2 ) & \le_T \FNPTz(q, \gamma_1, \gamma_2).
  \end{align*}
  Moreover, these reductions also hold for the analogous non-planar problems.  
\end{lemma}
\begin{proof}
  Let $G$ and $\rho$ be the inputs of $\TPT(q, \gamma_2)$. Let $H$ be a subgraph of $G$ and let $s$ and $t$ be two distinct connected vertices of $H$. By applying Lemma~\ref{lem:compute-fraction:q<0} we find that $Z_{st}( H; q, (\gamma_2+1)^{\rho}-1)$ and  $Z_{\text{Tutte}}( H; q, (\gamma_2+1)^{\rho}-1)$ are non-zero. Hence, $(q, (\gamma_2+1)^{\rho}-1)$ is zero-free for $G$. This shows that $\TPT(q, \gamma_2)$ reduces to $\ZTPT(q, \gamma_2)$. Lemma~\ref{lem:reduce-frac-to-exact:thickened} gives us a reduction from $ \ZTPT(q, \gamma_2)$ to $\RTPT(q, \gamma_2)$. Recall that we have $q<0$, $\gamma_1 \in (-1,0)$ and $\gamma_2 > 0$. Thus, Lemma~\ref{lem:compute-fraction:q<0} gives a reduction from $\RTPT(q, \gamma_2)$ to $\FNPTz(q, \gamma_1, \gamma_2)$, which completes the proof for the first reduction of the statement. The second reduction is analogous, but this time we apply Lemma~\ref{lem:compute-fraction:q<0:sign} instead of Lemma~\ref{lem:compute-fraction:q<0}. Finally, note that our reductions also hold for the analogous non-planar problems since the algorithms given in Lemma~\ref{lem:compute-fraction:q<0} and Lemma~\ref{lem:compute-fraction:q<0:sign} work for arbitrary graphs (non-necessarily planar) as long as the oracle does.
\end{proof}

\subsection{The connection between approximate shifts and reductions} \label{sec:hardness:approximate-shifts}

In this section we show how a polynomial-time approximate shift from $(q,\gamma_1)$ to $(q, \gamma_2)$ may allow us to reduce the problems of approximating the norm of the Tutte polynomial at $(q, \gamma_2)$ to the same problem at $(q, \gamma_1)$ (see Lemma~\ref{lem:reduction:norm:outline}). We also derive a similar result for the problem $\pDAT \rho \sDATz(q, \gamma)$ in Lemma~\ref{lem:reduction:distance}.

\begin{lemma} \label{lem:approximate-shift:reduction}
  Let $q$, $\gamma_1$ and $\gamma_2$ be algebraic numbers with $q \ne 0$ such that there is a polynomial-time series-parallel approximate shift from $(q, \gamma_1)$ to $(q, \gamma_2)$. Then there is an algorithm that has as input a graph $G$ and a positive integer $k$ and computes, in polynomial time in $k$ and the size of $G$, a graph $H$ and a representation of an algebraic number $D$ with $D \ne 0$ such that
    \begin{align*} 
    \Big|Z_{\text{Tutte}}\left(G; q, \gamma_2\right) - \frac{ Z_{\text{Tutte}}\left(H; q, \gamma_1\right)}{ D} \Big| \le 2^{-k}.
    \end{align*}
 Moreover, if the graph $G$ is planar, then the graph $H$ is also planar, and if $q$ and $\gamma_1$ are real, then $D$ is also real.
\end{lemma}

\begin{proof}
  Let $G = (V, E)$ and $k$ be the inputs of the algorithm. Let $n = |V|$ and $m = |E|$. By the definition of series-parallel polynomial-time approximate shifts, for any positive integer $j$, one can compute, in polynomial time in $j$, a series-parallel graph $J_j$ that $\gamma_1$-implements $\hat{\gamma}$ with $\left| \gamma_2 - \hat{\gamma} \right| \le 2^{-j}$ for terminals $s$ and $t$. By definition of implementations, we have $\hat{\gamma} = q Z_{st}( J_j; q, \gamma_1) / Z_{s|t}( J_j; q, \gamma_1)$ and $Z_{s|t}( J_j; q, \gamma_1) \ne 0$. We construct a graph $G_j$ that is a copy of $G$ where every edge $f$ in $G$ has been replaced by a copy of $J_j$ as in Lemma~\ref{lem:implementations:multivariate}, identifying the endpoints of $f$ with $s$ and $t$. In light of Lemma~\ref{lem:implementations:multivariate}, we have
  \begin{equation*}
    Z_{\text{Tutte}}\left(G_j; q, \gamma_1\right) = \left( \frac{Z_{s|t}\left( J_j; q, \gamma_1 \right)}{q^2} \right)^{m} Z_{\text{Tutte}}\left(G; q, \hat{\gamma}\right).
  \end{equation*}
  We can compute a representation of $D_j = Z_{s|t}\left( J_j; q, \gamma_1 \right) / q^2$ in polynomial time in the size of $J_j$ because $J_j$ is a series-parallel graph. However, note that this hypothesis is not essential as  long as there is some way to compute a representation of $D_j$ while constructing $J_j$. Note that $\left| \hat{\gamma} \right| \le \left| \gamma_2 \right| + 2^{-j}$, so $\left|Z_{\text{Tutte}}(G; q, \gamma_2) -  Z_{\text{Tutte}}(G; q, \hat{\gamma}) \right| $ is upper bounded by
  \begin{align*}
	\sum_{A \subseteq E} \left|q\right|^{k(A)} \left| \gamma_2^{|A|} - \hat{\gamma}^{|A|} \right|&\le \sum_{A \subseteq E} \left|q\right|^{k(A)} \left| \gamma_2 - \hat{\gamma}\right| \sum_{t = 0}^{|A|-1}  \left|\gamma_2^{|A| - 1 -t} \hat{\gamma}^{t}  \right|  \\
	&\le  \sum_{A \subseteq E} \left|q\right|^{k(A)} \left| \gamma_2 - \hat{\gamma}\right| \left(\lvert A \rvert-1\right) \Big( \left| \gamma_2 \right| + \frac{1}{2} \Big)^{|A|-1} \\ 
	& \le  \left| \gamma_2 - \hat{\gamma}\right|  \left|q\right|^{n} 2^{m} (m-1) \Big( \left| \gamma_2 \right| + \frac{1}{2} \Big)^{m-1}.
  \end{align*}
  Hence, for $j$ such that $2^{-j} |q|^{n} 2^{m} (m-1) ( | \gamma_2 | + 1/2 )^{m-1} \le 2^{-k}$, which can be achieved for $j = O(\mathrm{size}(G) + k)$, we obtain
  \begin{align*} \label{eq:approximate-shift:reduction}
    \Big|Z_{\text{Tutte}}\left(G; q, \gamma_2\right) - \frac{ Z_{\text{Tutte}}\left(G_j; q, \gamma_1\right)}{ D_j^{m}} \Big|=\left|Z_{\text{Tutte}}(G; q, \gamma_2) -  Z_{\text{Tutte}}(G; q, \hat{\gamma}) \right| \le 2^{-k}.
  \end{align*}
  The algorithm returns $H = G_j$ and $D = D_j^m \ne 0$. Note that if $G$ is planar, then $H=G_j$ is also planar by construction. If $q$ and $\gamma_1$ are real, then the number $D =  (Z_{s|t}( J_j; q, \gamma_1) / q^2)^m$ is clearly real too.
\end{proof}

In the rest of this section we use  Lemma~\ref{lem:approximate-shift:reduction} to translate information about the function $Z_{\text{Tutte}}(-; q, \gamma_1)$ for certain graphs to information about $Z_{\text{Tutte}}(G; q, \gamma_2)$. This leads to the reductions given in Lemmas~\ref{lem:reduction:norm:outline} and~\ref{lem:reduction:distance}. These results are stated for polynomial series-parallel approximate shifts, but they would also hold even if the shifts are not series-parallel as long as, in the proof of Lemma~\ref{lem:approximate-shift:reduction}, the graphs $J_j$ are planar and we can compute $D_j = Z_{s|t}\left( J_j; q, \gamma_1 \right) / q^2$ in polynomial time in the size of $J_j$.

We are now ready to prove Lemma~\ref{lem:reduction:norm:outline}, which was stated in the Proof Outline, and which we restate here for convenience.
\begin{lemout}
\statelemredout
\end{lemout}
\begin{proof}
  We are going to solve $\textsc{Factor-}4K\textsc{-NonZero-NormTutte}(q, \gamma_2)$ in polynomial time with the help of an oracle for $\FNTz(q, \gamma_1)$. Recall that hardness of these problems does not depend on $K$ (see Section~\ref{sec:intro:main-results}). Let $C_{q,\gamma_2} > 1$ be the constant computed in Corollary~\ref{cor:lower-bound:tutte} for the algebraic numbers $q$ and $\gamma= \gamma_2$; so, for any graph $G$, either  $Z_{\text{Tutte}}(G; q, \gamma_2) = 0$ or  $\lvert Z_{\text{Tutte}}(G; q, \gamma_2) \rvert \ge C_{q, \gamma_2}^{-\mathrm{size}(G)}$. Let $G = (V, E)$ be the input of the computational problem $\textsc{Factor-}4K\textsc{-NonZero-NormTutte}(q, \gamma_2)$. We assume that $Z_{\text{Tutte}}(G; q, \gamma_2) \ne 0$ since otherwise we can output anything.  Let $k$ be the smallest integer such that $2^{-k} \le C_{q,\gamma_2}^{-\mathrm{size}(G)}/ 2$. The reduction uses the algorithm given in Lemma~\ref{lem:approximate-shift:reduction} to compute a graph $H$ and a representation of an algebraic number $D$ with $D \ne 0$ such that
  \begin{equation} \label{eq:reduction:lower-bound}
    \left| Z_{\text{Tutte}}\left(G; q, \gamma_2\right) - \frac{Z_{\text{Tutte}}\left(H; q, \gamma_1\right)}{D} \right| \le 2^{-k} \le  \frac{C_{q, \gamma_2}^{-\mathrm{size}(G)}}{2} \le \frac{ \left| Z_{\text{Tutte}}\left(G; q, \gamma_2\right) \right| }{2}.
  \end{equation}
  Therefore, we have
  \begin{equation*}
    \frac{1}{2} \le \frac{\left| Z_{\text{Tutte}}\left(H; q, \gamma_1\right) \right|}{D \left| Z_{\text{Tutte}}\left(G; q, \gamma_2\right) \right|} \le \frac{3}{2}.
  \end{equation*}
  By invoking the oracle for $\FNTz(q, \gamma_1)$, the reduction computes a rational number $N$ with $N / K \le \left| Z_{\text{Tutte}}(H; q, \gamma_1) \right| \le K N$. The reduction also computes a non-zero rational number $\hat{D}$ such that $1/2  \le D/ \hat{D} \le 2$. Then $\hat{N} = N / \hat{D}$ satisfies
  \begin{equation*}
    \frac{\hat{N}}{4K} \le  \left| Z_{\text{Tutte}}\left(G; q, \gamma_2\right) \right| \le 4K\hat{N},
  \end{equation*}
  so the reduction outputs $\hat{N}$ for $\textsc{Factor-}4K\textsc{-NonZero-NormTutte}(q, \gamma_2)$. Note that this reduction analogously applies to the planar case since the graph $H$ is planar when $G$ is planar (see Lemma~\ref{lem:approximate-shift:reduction}).
\end{proof}

We next give the analogue of Lemma~\ref{lem:reduction:norm:outline} for the argument.

\begin{lemma} \label{lem:reduction:distance}
  Let $q$, $\gamma_1$ and $\gamma_2$ be algebraic numbers with $q \ne 0$. If there is a polynomial-time series-parallel approximate shift from $(q, \gamma_1)$ to $(q, \gamma_2)$, then we have the following reduction, $\pDAT 5 \pi/12 \sDATz(q, \gamma_2) \le_T \pDAT \pi/3 \sDATz(q, \gamma_1)$. This reduction also holds for the planar version of the problem.
\end{lemma}

\begin{proof}
  Let $C_{q,\gamma_2}$ be the constant computed in Corollary~\ref{cor:lower-bound:tutte} for $\gamma= \gamma_2$. Let $G = (V, E)$ be the input of $\pDAT \pi/2 \sDATz(q, \gamma_2)$. We assume that $ Z_{\text{Tutte}}(G; q, \gamma_2)  \ne 0$ since otherwise we can output anything. The reduction proceeds again similarly to that of Lemma~\ref{lem:reduction:norm:outline}. First, it applies Lemma~\ref{lem:approximate-shift:reduction} for appropriate $k$ as in \eqref{eq:reduction:lower-bound} to compute a graph $H$ and a representation of a real algebraic number $D$ with $D \ne 0$ such that
  \begin{equation} \label{eq:reduction:lower-bound:2}
    \left| Z_{\text{Tutte}}\left(G; q, \gamma_2\right) - \frac{Z_{\text{Tutte}}\left(H; q, \gamma_1\right)}{D} \right| \le 2^{-k-2} \le  \frac{C_{q, \gamma_2}^{-\mathrm{size}(G)}}{8} \le \frac{ \left| Z_{\text{Tutte}}\left(G; q, \gamma_2\right) \right| }{8}.
  \end{equation}
  Let $\alpha = Z_{\text{Tutte}}(G; q, \gamma_2)$ and $\beta = Z_{\text{Tutte}}(H; q, \gamma_1 ) / D $, so \eqref{eq:reduction:lower-bound:2}  can be rewritten as $\left| \alpha - \beta \right| \le \left| \alpha \right|/8$. We claim that $|\mathrm{Arg}(\alpha) - \mathrm{Arg}(\beta) | \le \pi / 24$. Since $\beta$ is in the disc of centre $\alpha$ and radius $|\alpha| / 8$,  by basic geometry, we have
  
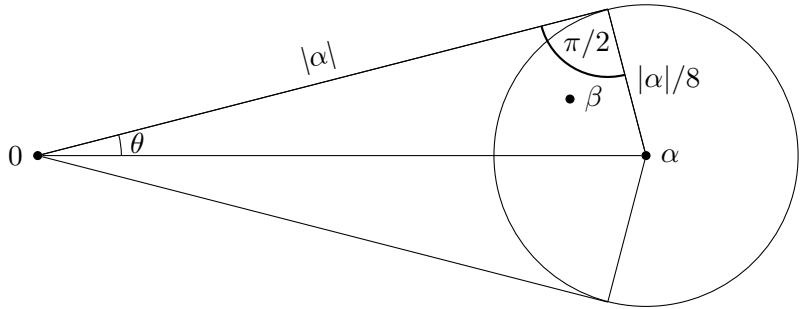
\begin{center}
\begin{tikzpicture}
\node [fill, draw, circle, minimum width=3pt, inner sep=0pt, label=left:$0$] (0) at (0,0) {};

\node [fill, draw, circle, minimum width=3pt, inner sep=0pt, label=right:$\alpha$] (o) at (8, 0) {};

\node [fill, draw, circle, minimum width=3pt, inner sep=0pt, label=right:$\beta$] (beta) at (7, 0.75) {};

\node [circle,draw,name path=circle] (c) at (o) [minimum size=4cm] {};

\draw (0)  --  (tangent cs:node=c,point={(0)},solution=1) 
coordinate (B)
       (0) -- (tangent cs:node=c,point={(0)},solution=2)
       coordinate (C)
       (o)  -- (tangent cs:node=c,point={(0)},solution=2) 
       (o)  -- (tangent cs:node=c,point={(0)},solution=1) 
       (0) -- (o);

\draw (0) -- (C) node [midway, above, sloped] (TextNode1) {$|\alpha|$};
\draw (o) -- (C) node [midway, right] (TextNode1) {$|\alpha|/8$};
 \pic [draw, -, angle radius=11mm, angle eccentricity=1.2, "$\theta$"] {angle = o--0--C};
\pic [draw,thick,angle radius=9mm, "$\pi/2$"]   {angle = 0--C--o};
\end{tikzpicture}
\end{center}
so $\sin(\theta) = 1/8$, where $\theta$ is the angle between $0$, $\alpha$ and the intersection of the circle of radius $|\alpha|/8$ and center $\alpha$ with the tangent line that goes through $0$. Since $\sin(\pi/24) > 1/8$, we conclude that $|\mathrm{Arg}(\alpha) - \mathrm{Arg}(\beta) | \le \theta \le \pi/24$ as we claimed. By invoking the oracle for $\pDAT \pi/3 \sDATz(q, \gamma_1)$, the reduction computes a rational number $\hat{A}_1$ such that, for some $a_1 \in \arg(Z_{\text{Tutte}}(H; q, \gamma_1))$, we have $|a_1  - \hat{A}_1  | \le \pi/3$. Since the reduction has at its disposal a representation of the algebraic number $D$, it can compute (in polynomial time in the length of this representation) a rational number $\hat{A}_2$ such that, for some $a_2 \in \arg(D)$, we have $| a_2 - \hat{A}_2 | \le \pi/24$. The reduction outputs $\hat{A} = \hat{A}_1 - \hat{A}_2$. We claim that there is an argument $a$ of $\alpha$ such that $|a - \hat{A}| \le 5\pi/12$. Note that $b = a_1 - a_2$ is an argument of $\beta$. By the triangle inequality, we have $| b - \hat{A} | \le | a_1 - \hat{A}_1 | + | \hat{A}_2 - a_2 | \le 9\pi/24$. Let $a = \mathrm{Arg}(\alpha) + (b - \mathrm{Arg}(\beta))$, which is an argument of $\alpha$. We conclude that
\begin{equation*}
  |a - \hat{A} | \le | b - \hat{A} | + | a - b | = | b - \hat{A} | + | \mathrm{Arg}(\alpha) - \mathrm{Arg}(\beta) | \le 5 \pi / 12.
\end{equation*}
This reduction analogously applies to the planar case since the graph $H$ is planar when $G$ is planar (see Lemma~\ref{lem:approximate-shift:reduction}).
\end{proof}

One could actually change the angles $\rho_2 = 5 \pi /12$ and $\rho_1 = \pi/3$ in the statement of Lemma~\ref{lem:reduction:distance} as long as $\rho_1 < \rho_2$, but $\rho_2 = 5 \pi /12$ and $\rho_1 = \pi/3$ will suffice for our purposes.

\subsection{Hardness for the Tutte polynomial} \label{sec:hardness:tutte}

In this section we  use the reductions of Section~\ref{sec:hardness:real} to obtain intermediate hardness results that will be used to obtain our main theorems in the upcoming sections. We start with the following corollary which strengthens previous results of~\cite{Goldberg2014} (that applied to general graphs rather than planar).

\begin{corollary} \label{cor:first-hardness-result:q>1}
  Let $K>1$ be a real number. Let $q\neq 0,2$, and $\gamma_1,\gamma_2$ be real algebraic numbers  with $\gamma_2\in (-\infty,-2)\cup (0,\infty)$ and either 
	\begin{itemize}
	\item $q >1$, $\gamma_1 \in (-2, -1)$, or
	\item $q < 1$, $\gamma_1 \in (-1, 0)$.
	 \end{itemize}
	Then, $\FNPTz(q, \gamma_1, \gamma_2)$ and $\SPTz(q, \gamma_1, \gamma_2)$ are $\numP$-hard.
	\end{corollary}
 \begin{proof}
We consider first the case when $\gamma_2>0$. For $q,\gamma_1,\gamma_2$ as in the first item, the conclusion  follows from Theorem~\ref{thm:hardness:exact-planar} and the reductions given in Lemma~\ref{lem:exact-reduction:q>1}. For the second item: when $q\in (0,1)$, the result follows from the reductions given in Lemma~\ref{lem:exact-reduction:0<q<1} and Theorem~\ref{thm:hardness:exact-planar}, while for $q<0$, the result follows from Lemmas~\ref{lem:thickened-tutte} and~\ref{lem:exact-reduction:q<0}.

The other case is when $\gamma_2 < -2$. Then, we can $\gamma_2$-implement $(\gamma_2+1)^2 -1 > 0$ with a $2$-thickening and proceed as in the previous case.
 \end{proof}

\begin{lemma} \label{lem:regionB}
   Let $K$ be a real number with $K > 1$. Let $x, y$ be a real algebraic numbers such that $(x,y) \ne (-1, -1)$, $\min \{x, y\} \le -1$ and $\max \{x, y\} < 0$. Let $q = (x-1)(y-1)$ and $\gamma = y-1$. Then $\FNPTz(q, \gamma)$ and $\SPTz(q, \gamma)$ are $\# \mathsf{P}$-hard.
\end{lemma}
\begin{proof}
  Note that $q > 2$. We claim that we can $(x,y)$-implement $(x_1, y_2)$ with $y_1 \in (-1,0)$, and $(x_2, y_2)$ with $|y_2|>1$  using planar (in fact, series-parallel) graphs. The result then follows by invoking Corollary~\ref{cor:first-hardness-result:q>1} with $\gamma_1 = y_1- 1$ and $\gamma_2 = y_2 -1$. 
	
	The case $\min \{x,y\} < -1$ is treated in~\cite[Lemmas 8--11]{Goldberg2014}. Hence, we may assume that $-1 \le x < 0$ and $-1 \le y < 0$. Since $(x,y) \ne (-1,-1)$ by hypothesis, there are two cases:
  \begin{itemize}
  \item $x = -1$ and $-1 < y < 0$. As pointed out in~\cite[Corollary 26]{Goldberg2014}, a $3$-thickening from $(x,y)$ implements the point
$(x', y') = \left( 1 -\frac{2}{1+y+y^2}, y^3 \right)$  with $x' < -1$ and $y' \in (-1,0)$, so the point $(x', y')$ has already been studied in this proof.
  \item $-1 < x < 0$ and $y = -1$. This time we perform a $3$-stretching from $(x,y)$ to implement a point $(x',y')$ with $x' \in (-1,0)$ and $y' < -1$. \qedhere
  \end{itemize}
\end{proof}

\begin{lemma} \label{lem:regionG}
   Let $K>1$ be a real number and $q,x, y$ be real algebraic numbers with $\max \{|x|, |y|\} < 1$ and $q = (x-1)(y-1) > 32/27$. 
	
Then, for $\gamma=y-1$, $\FNPTz(q, \gamma)$ and  $\SPTz(q, \gamma)$ the following problems are $\# \mathsf{P}$-hard, unless $q=2$.
\end{lemma}
\begin{proof}
  In view of~\cite[Lemmas 12 and 15]{Goldberg2014}, we can $(x,y)$-implement points $(x_1,y_1)$ and  $(x_2,y_2)$ with $y_1 \in (-1,0)$ and $y_2 > 1$. These implementations only use series-parallel graphs. Hence, we can apply (the first item of) Corollary~\ref{cor:first-hardness-result:q>1} with $\gamma_1 = y_1- 1$ and $\gamma_2 = y_2 -1$ to finish the proof.
\end{proof}

\subsection{Proofs of our main theorems}\label{sec:mainproof}

In this section we show how our main Theorems~\ref{thm:mainone},~\ref{thm:maintwo},~\ref{thm:mainthree} and~\ref{thm:mainfour} follow from the $\numP$-hardness results of Section~\ref{sec:hardness:tutte}. We start with Theorem~\ref{thm:mainfour}.

\begin{thmmainfour} 
\statethmmainfour
\end{thmmainfour}
\begin{proof}
Let $(x,y)\in \mathcal{H}_q$ be such that $y=\gamma+1$. Consider  the point $(x_2, y_2) \in H_q$ with $y_2 = -1/2$ and $x_2 = 1 + q/(y_2 -1)$.  Note that $x_2 = 1 - 2q/3 \le 1 - 4/3 < 0$. There are two cases. Either $x_2 \le -1$ and the point $(x_2, y_2)$ satisfies the hypothesis of Lemma~\ref{lem:regionB}, or $-1 < x_2 < 0$ and the point $(x_2, y_2)$ satisfies the hypothesis of Lemma~\ref{lem:regionG}. In any case, we conclude that  $\SPTz(q, \gamma_2)$ and $\FNPTz(q, \gamma_2)$ are $\numP$-hard for $\gamma_2 = y_2 -1$ when $q > 2$. 

  By Lemma~\ref{lem:reduction:norm:outline} (for $\gamma_1 = \gamma$ and $\gamma_2 = \gamma_2$), we see that $\FNPTz(q, \gamma_2)$ reduces to $\FNPTz(q, \gamma)$, proving that the latter is $\numP$-hard too. The proof for $\pDAT \pi/3 \sDAPTz(q, \gamma)$ is analogous: first observe that since $q,\gamma_2$ are real and $5\pi/12<\pi/2$, the problem $\SPTz(q, \gamma_2)$ reduces (trivially) to $\pDAT 5\pi/12 \sDAPTz(q, \gamma_2)$. Moreover, applying Theorem~\ref{thm:approx-shifts} with $x$ and $y$ as above, $y' = y_2\in (-1,0)$ and $x' = x_2$, we have a polynomial-time series-parallel approximate shift from $(x,y)$ to $(x',y')$ or, equivalently, from $(q, \gamma)$ to $(q, \gamma_2)$.  Using Lemma~\ref{lem:reduction:distance} with $\gamma_1 = \gamma$ and $\gamma_2 = \gamma_2$, we conclude that $\pDAT 5\pi/12 \sDAPTz(q, \gamma_2)$ reduces to $\pDAT \pi/3 \sDAPTz(q, \gamma)$, proving that the latter is $\numP$-hard, as wanted. 
\end{proof}

\begin{thmmainthree} 
\statethmmainthree
\end{thmmainthree}
\begin{proof}
Let $q=2$, $\gamma=y-1$, $y_2=-1/2$, $\gamma_2=y_2-1$. From the result of Goldberg and Guo~\cite{Goldberg2017},  the problems $\FNI(y_2)$ and $\pDAT \pi/3 \sDAI(y_2)$ are $\numP$-hard, hence  $\FNPTz(q, \gamma_2)$ and $\pDAT \pi/3 \sDAPTz(q, \gamma_2)$ are $\numP$-hard as well, using that $Z_{\mathrm{Ising}}(G; y_2) = Z_{\mathrm{Tutte}}(G; 2, \gamma_2)$.

By applying Lemma~\ref{lem:reduction:norm:outline} and Theorem~\ref{thm:approx-shifts} analogously to the proof of Theorem~\ref{thm:mainfour}, we conclude that $\FNPTz(q, \gamma)$ and $\pDAT \pi/3 \sDAPTz(q, \gamma)$ are \#P-hard, and hence $\FNI(y)$ and $\pDAT \pi/3 \sDAI(y)$, using that $Z_{\mathrm{Ising}}(G; y) = Z_{\mathrm{Tutte}}(G; 2, \gamma)$.
\end{proof}

\begin{thmmainone}
\statethmmainone
\end{thmmainone}
\begin{proof}
Just apply Theorem~\ref{thm:mainfour} to the integer $q$, and use $Z_{\mathrm{Potts}}(G; q, y) = Z_{\mathrm{Tutte}}(G; q, y-1)$.
\end{proof} 

\begin{thmmaintwo}
\statethmmaintwo
\end{thmmaintwo}
\begin{proof}
Let $y \in (-q+1, 0)$. The point $(x,y)$ with $x = 1+q/(y-1)$ satisfies $x \in (1-q, 0)$, $(x,y) \ne (-1,-1)$ and $y < 0$. If $x \le -1$ or $y \le -1$, $\numP$-hardness follows from Lemma~\ref{lem:regionB}. Otherwise, we have $q \ge 3$ and $ x,y\in (-1,0)$, so hardness follows from Lemma~\ref{lem:regionG}.
\end{proof}

\section{Further consequences of our results} \label{sec:further-consequences}

In this final section, we discuss some further consequences of our techniques, as mentioned in Section~\ref{sec:intro:tutte}. First, in Section~\ref{sec:further-consequences:real-summary}, we explain how our results can be used to obtain hardness for  $\SPTz(q, \gamma)$ and $\FNPTz(q, \gamma)$ (and the non-planar version of these problems) at other parameters than the ones studied in Section~\ref{sec:hardness:tutte}, building on work of Goldberg and Jerrum~\cite{Goldberg2014}. Secondly, in Section~\ref{sec:further-consequences:bqp}, we apply our results to the problem of approximating the Jones polynomial of an alternating link, which is connected to the quantum complexity class $\mathsf{BQP}$ as explained in~\cite{Bordewich2005}.

\subsection{Hardness results for real algebraic parameters in the Tutte plane} \label{sec:further-consequences:real-summary}

The regions studied in Lemmas~\ref{lem:regionB} and~\ref{lem:regionG} have been studied by Goldberg and Jerrum~\cite{Goldberg2014}, where they showed $\#\mathsf{P}$-hardness of $\SPT(q, \gamma)$ at several regions of the real algebraic plane. As we explained in Section~\ref{sec:hardness:tutte}, we obtain hardness at a point $(q,\gamma)$ as long as we can $\gamma$-implement algebraic numbers $\gamma_1$ and $\gamma_2$ as in Corollary~\ref{cor:first-hardness-result:q>1}. Goldberg and Jerrum came up with multiple implementations that achieve the conditions of  Corollary~\ref{cor:first-hardness-result:q>1}. By applying their implementations, we obtain $\numP$-hardness for $\FNT(q, y-1)$ in the same regions where they obtained $\numP$-hardness of $\SPT(q, \gamma)$ in~\cite[Theorem 1]{Goldberg2014}.

 Some of the implementations developed in~\cite{Goldberg2014} consist of planar graphs (as those used in Lemmas~\ref{lem:regionB} and~\ref{lem:regionG}), so we can extend their results to the planar version of the problems for some of the previous regions. 

\begin{theorem} \label{thm:hardness:real:planar} 
  Let $q $ and $\gamma$ be real algebraic numbers with $q\neq 0,1,2$. Let $y = \gamma+1$ and $x = 1 + q/ (y-1)$. The problems $\SPTz(q, \gamma)$ and $\FNPTz(q, \gamma)$ are $\# \mathsf{P}$-hard when $x,y$ are real algebraic numbers satisfying one of the following: 
\begin{enumerate}
\item\label{it:qxwecca}  $\min(x,y) \leq -1$, $\max(x,y)<0$ and $(x,y) \ne (-1,-1)$,
\item\label{it:qxweccb}  $|x|>1$, $|y|>1$ and $xy < 0$,
\item\label{it:qxweccc}  $\max(|x|,|y|)<1$ and $q>32/27$,
\item\label{it:qxweccd}  $\max(|x|,|y|)<1$, $q\leq 32/27$ and $x<-2y-1$,  
\item\label{it:qxwecce}  $\max(|x|,|y|)<1$, $q\leq 32/27$ and $y<-2x-1$.
\end{enumerate}
\end{theorem}
\begin{proof}
  The proof follows from the following results of~\cite{Goldberg2014}, which show how to implement $\gamma_1$ and $\gamma_2$ with a planar (actually series-parallel) graph as in Corollary~\ref{cor:first-hardness-result:q>1} for each of the regions in the statement.
  
	Item~\ref{it:qxwecca} follows from  Lemma~\ref{lem:regionB}. For Item~\ref{it:qxweccb}, note that $q  < 0$, so we have to implement $\gamma_1 \in (-1,0)$ and $\gamma_2 \not \in [-2, 0]$. We choose $\gamma_2 = y-1$ and $\gamma_1$ as implemented in~\cite[Lemma 16]{Goldberg2014}. Item~\ref{it:qxweccc} follows from  Lemma~\ref{lem:regionB}. For Item~\ref{it:qxweccd}, we implement $\gamma_1 \in (-1,0)$ and $\gamma_2 \not \in [-2, 0]$; the implementations are as in~\cite[Lemmas 14 and 15]{Goldberg2014}. For Item~\ref{it:qxwecce}, we implement $\gamma_1 \in (-1,0)$ and $\gamma_2 \not \in [-2, 0]$; the implementations are as in~\cite[Lemmas 13 and 15]{Goldberg2014}.
\end{proof}

The complexity of approximating the Tutte polynomial of a planar graph has previously been studied in~\cite{Goldberg2012} and~\cite{Kuperberg2015}. Our result on this matter (Theorem~\ref{thm:hardness:real:planar}) strengthens the results of~\cite{Goldberg2012} in three directions. First, we also study the complexity of determining the sign of the Tutte polynomial. Secondly, we find new regions where the approximation problem is hard. These regions are 3, 4 and 5, as well as the points in region 1 such that $q \le 5$ and $q \ne 3$. Finally, we prove $\# \mathsf{P}$-hardness, whereas in~\cite{Goldberg2012} hardness was obtained under the hypothesis that $\mathsf{RP} \ne \mathsf{NP}$.

For $q \in \mathbb{Z}^+$, let  $P(G;q)$ count the number of proper $q$-colourings of a graph $G$. The \emph{chromatic polynomial} of $G$ is the only polynomial that agrees with  $P(G;q)$  on positive integers. It is well-known that $P(G;q) =  Z_{\text{Tutte}}(G; q, -1)$, see for instance \cite{Sokal2005}. The value $q = 32/27$ appearing in Theorem~\ref{thm:hardness:real:planar} is, in some sense, a phase transition for the complexity of computing the sign of $P(G,q)$: this sign depends upon $G$ in an essentially trivial way for $q < 32/27$ \cite[Theorem 5]{jackson1993} and its computation is $\numP$-hard for $q > 32/27$, see ~\cite{Goldberg2014} for an in detail discussion of the relevance of the phase transition $q = 32/27$. 



\subsection{Hardness results for the Jones polynomial} \label{sec:further-consequences:bqp}

We briefly review some relevant facts about links and the Jones polynomial that relate it to the Tutte polynomial on graphs, see~\cite{Welsh1993} for their definitions. Let $V_L(T)$ denote the Jones polynomial of a link $L$.  By a result of Thistlethwaite, when $L$ is an alternating link with associated planar graph $G(L)$, we have $V_L(t) = f_L(t) T(G(L); -t, -t^{-1})$, where $f_L(t)$ is an easily-computable factor that is plus or minus a half integer power of $t$, and $T(G; x, y)$ is the Tutte polynomial of $G$ in the $(x,y)$-parametrisation~\cite{Thistlethwaite1987,Welsh1993}. Moreover, every planar graph is the graph of an alternating link~\cite[Chapter 2]{Welsh1993}. Hence, we can translate our results on the complexity of approximating the Tutte polynomial of a planar graph to the complexity of approximating the Jones polynomial of an alternating link, and obtain $\numP$-hardness results for approximating $V_L(t)$. More formally, we consider the following problems.

\prob{$\FNJ(t)$.}{A link $L$.}{If $V_L(t)= 0$, the algorithm may output any rational number. Otherwise, it must output $\hat{N} \in \mathbb{Q}$ such that $\hat{N} / K \le |V_L(t)| \le K \hat{N}$.}
  
\prob{$\DAJ(q, \gamma)$.}{A link $L$.}{If $V_L(t)= 0$, the algorithm may output any rational number. Otherwise, it must output $\hat{A} \in \mathbb{Q}$ such that, for some $a \in \arg(V_L(t))$, we have $| \hat{A} - a| \le \rho$.}

\begin{corollary} \label{cor:jones}
Let $K$ be a real number with $K > 1$. Let $t$ be an algebraic number with $\mathrm{Re}(t) > 0$. Then $\FNJ(t)$ and $\DAJpi(t)$ are $\numP$-hard unless $t \in \{1, -e^{2\pi i/3}, -e^{4\pi i/3}\}$ when both problems can be solved exactly.
\end{corollary}
\begin{proof}
  Let us consider the point $(x,y)=(-t,-t^{-1})$ in the Tutte plane. Note that $t \in \{1, -e^{2\pi i/3}, -e^{4\pi i/3}\}$ if and only if $(x,y)$ is one of the special points $(-1,-1)$, $(e^{4\pi i/3}, e^{2\pi i/3})$ and $(e^{2\pi i/3}, e^{4\pi i/3})$, where the Jones polynomial of a link can be exactly evaluated in polynomial time in the size of the link \cite{Jaeger1990}. Let us assume that $t$ is not one of these three values. We have $q = (-t-1)(-t^{-1} -1) = 2 + 2\mathrm{Re}(t) > 2$. When $t$ is non-real, in view of  Theorem~\ref{thm:mainfour}, $\FNPT(q, y-1)$ and $\pDAT \pi/3 \sDAPT(q,y-1)$ are $\numP$-hard and the result follows. When $t$ is real, note that $y < 0$, $x < 0$ and $q > 2$. Thus, either $(x,y)$ is such that $\max \{|x|, |y|\} \ge 1$ and $(x,y) \ne (-1,-1)$, so hardness is covered in region 1 of Theorem~\ref{thm:hardness:real:planar}, or $\max \{|x|, |y|\} < 1$, so hardness is covered in region 3 of Theorem~\ref{thm:hardness:real:planar}.
\end{proof}

The case $t = e^{2\pi i /5}$ of Corollary~\ref{cor:jones} is particularly relevant due to its connection with quantum computation. This connection between approximate counting and the quantum complexity class $\mathsf{BQP}$ was explored by Bordewich, Freedman, Lov\'{a}sz and Welsh in~\cite{Bordewich2005}, where they posed the question of determining the complexity of the following problem:

\prob{$\SRPT(q, \gamma)$}{A planar (multi)graph $G$.}{Determine whether $\mathrm{Re} (Z_{\text{Tutte}}(G; q, \gamma)) \ge 0$ or $\mathrm{Re}(Z_{\text{Tutte}}(G; q, \gamma)) \le 0$.}

 The non-planar version of $\SRPT(q, \gamma)$ has been studied in~\cite[Section 5]{Goldberg2017}, where it was shown that determining the sign of the real part of the Tutte polynomial is $\# \mathsf{P}$-hard in certain cases that include $t =e^{2 \pi i /5}$. Our results on the complexity of $\SPTz(q, \gamma)$ allow us to adapt the argument in \cite{Goldberg2017} to answer the question asked in~\cite{Bordewich2005}.

\begin{corollary} \label{cor:sign-real}
Consider the point $(x,y) = (\exp(-a \pi i /b), \exp(a \pi i /b))$, where $a$ and $b$ are positive integers such that $1/2 < a/b < 3/2$ and $a \ne b$. Let $q = (x-1)(y-1)$ and $\gamma = y-1$. Then $q \in (2,4)$ and $\SRPT(q, \gamma)$ is $\# \mathsf{P}$-hard.
\end{corollary}
\begin{proof}
  The proof is essentially the same one as that of~\cite[Theorem 1.7]{Goldberg2017}. First, note that
  \begin{equation*}
    q = (x-1) (y-1) = 2 - x - y = 2 - \exp(-a \pi i /b) - \exp(a \pi i /b) = 2 - 2 \cos (a \pi / b),
  \end{equation*}
  which is real. Since $1/2 < a/b < 3/2$ and $a \ne b$, we have $q \in (2, 4)$. A $b$-thickening allows us to $(x,y)$-implement $(1-q/2, -1)$. Since $\SPTz(q, -2)$ is $\# \mathsf{P}$-hard (see Theorem~\ref{thm:hardness:real:planar}), we conclude that  $\SRPT(q, \gamma)$ is $\# \mathsf{P}$-hard.
\end{proof}

 Corollary~\ref{cor:sign-real} includes the case where $a = 3$ and $b =5$. In this case, we have $x = \exp(-a \pi i /b) = -\exp(\pi i) \exp(-3 \pi i /5) = -\exp(2 \pi i /5)$ and $y =x^{-1}$. That is, $(x,y) = (-t, -t^{-1})$ for $t = \exp(2 \pi i /5)$, which is the point of interest in~\cite{Bordewich2005}.

\section*{Acknowledgement}

We thank Ben Green, Joel Ouaknine and Oliver Riordan for useful discussions on Section~\ref{sec:hardness:lower-bound}. We also thank Miriam Backens for useful conversations and suggestions about this work.

\bibliographystyle{plain}
\bibliography{\jobname}

\end{document}